\documentclass[a4paper,12pt, twoside]{article}

\usepackage[margin=1in]{geometry} 
\usepackage{amsmath, amsthm, slashed} 
\usepackage{authblk}
\usepackage{xparse} 
\usepackage{amssymb} 
\usepackage{enumerate} 
\usepackage{tikz}
\usetikzlibrary{decorations.pathmorphing,calc,intersections}
\usepackage[margin=1in]{geometry}
\usepackage{accents} 
\usepackage{tensor} 
\usepackage{bbm} 
\usepackage{mathrsfs} 
\usepackage{comment} 
\usepackage[normalem]{ulem} 
\usepackage{wasysym} 
\usepackage{relsize} 

\newcommand{\newcustomtheorem}[2]{%
	\newenvironment{#1}[1]
	{%
		\renewcommand\customgenericname{#2}%
		\renewcommand\theinnercustomgeneric{##1}%
		\innercustomgeneric
	}
	{\endinnercustomgeneric}
}
\newtheorem{innercustomgeneric}{\customgenericname}
\providecommand{\customgenericname}{}

\newtheorem{theorem}{Theorem}[section] 
\newtheorem*{theorem*}{Theorem}
\newcustomtheorem{customthm}{Theorem}

\newtheorem{proposition}{Proposition}[section]
\newtheorem*{proposition*}{Proposition}

\newtheorem{corollary}{Corollary}[section]
\newtheorem*{corollary*}{Corollary}

\newtheorem{lemma}{Lemma}[section]
\newtheorem*{lemma*}{Lemma*}

\newtheorem{remark}{Remark}[section]
\newtheorem*{remark*}{Remark}

\newtheorem{definition}{Definition}[section]
\newtheorem*{definition*}{Definition}

\newtheorem*{conjecture*}{Conjecture}

\newtheorem*{question*}{Question}

\newcommand{\Id}{\textnormal{Id}}
\newcommand{\og}{\overline{g}}

\newcommand{\onabla}{\overline{\nabla}}
\newcommand{\reals}{\mathbb{R}}

\newcommand{\stkout}[1]{\ifmmode\text{\sout{\ensuremath{#1}}}\else\sout{#1}\fi}

\newcommand{\Mg}{\big(\mcalm, g_M\big)}
\newcommand{\twosphere}{\mathbb{S}^2}

\newcommand{\roundmetric}{g_{\twosphere}}
\newcommand{\bg}{\boldsymbol{g}}
\newcommand{\bog}{\boldsymbol{\overline{g}}}
\newcommand{\qg}{\widetilde{g}_M}
\newcommand{\sg}{\slashed{g}_M}

\newcommand{\qmsm}{Q\astrosun S}

\newcommand{\ric}{\textnormal{Ric}}

\newcommand{\eh}{\mcalH^+}
\newcommand{\nullinf}{\mcalI^+}

\newcommand{\taus}{\tau^\star}

\newcommand{\depsilon}{\mathring{\epsilon}}

\newcommand{\smfun}{\Gamma(\mcalm)}
\newcommand{\smfunrad}{\Gamma_{{l\geq 2}}(\mcalm)}
\newcommand{\smonecov}{\Gamma(T^*\mcalm)}
\newcommand{\smonecovQ}{\Gamma(T^*Q)}
\newcommand{\smonecovS}{\Gamma(T^*S)}
\newcommand{\smonecon}{\Gamma(T\mcalm)}
\newcommand{\smoneconQ}{\Gamma(TQ)}
\newcommand{\smoneconS}{\Gamma(TS)}
\newcommand{\smtwocov}{\Gamma(T^2T^*\mcalm)}
\newcommand{\smsymtwocov}{\Gamma(S^2T^*\mcalm)}
\newcommand{\smsymtwocovQ}{\Gamma(S^2T^*Q)}
\newcommand{\smsymtwocovS}{\Gamma(S^2T^*S)}

\newcommand{\smsymtratwocovS}{\Gamma(\widehat{S}^2T^*S)}
\newcommand{\smthreecov}{\Gamma(T^3T^*\mcalm)}
\newcommand{\smncovQ}{\Gamma(T^nT^*Q)}
\newcommand{\smsymncovQ}{\Gamma(S^nT^*Q)}
\newcommand{\smncovS}{\Gamma(T^nT^*S)}
\newcommand{\smsymncovS}{\Gamma(S^nT^*S)}

\newcommand{\smcovS}[1]{\Gamma(T^{#1}T^*S)}
\newcommand{\smsymcovS}[1]{\Gamma(S^{#1}T^*S)}
\newcommand{\smsymtfcovS}[1]{\Gamma(\widehat{S}^{#1}T^*S)}
\newcommand{\smcovQ}[1]{\Gamma(T^{#1}T^*Q)}
\newcommand{\smsymcovQ}[1]{\Gamma(S^{#1}T^*Q)}
\newcommand{\smsymtfcovQ}[1]{\Gamma(\widehat{S}^{#1}T^*Q)}

\newcommand{\smqmsm}{\Gamma(T^*Q)\astrosun\Gamma(T^*S)}

\newcommand{\fgaumap}{f_{\textnormal{lin}}}
\newcommand{\frw}{f_{\textnormal{RW}}}
\newcommand{\fkerr}{f_{\textnormal{Kerr}}}
\newcommand{\fmrw}{f_{\textnormal{mRW}}}

\newcommand{\exd}{\textnormal{d}}
\newcommand{\pt}{\partial_{t^*}}
\newcommand{\pr}{\partial_r}
\newcommand{\px}{\partial_x}

\newcommand{\qtr}{\textnormal{tr}_{\qg}}
\newcommand{\qhd}{\widetilde{\star}}
\newcommand{\qsharp}{\widetilde{\sharp}}

\newcommand{\qexd}{\widetilde{\textnormal{d}}}
\newcommand{\qn}{\widetilde{\nabla}}
\newcommand{\qdiv}{\widetilde{\delta}}

\newcommand{\qbox}{\widetilde{\Box}}
\newcommand{\qepsilon}{\widetilde{\epsilon}}

\newcommand{\str}{\slashed{\textnormal{tr}}}
\newcommand{\shd}{\slashed{\star}}
\newcommand{\otimeshat}{\widehat{{\otimes}}}
\newcommand{\astrosunhat}{\widehat{{\astrosun}}}
\newcommand{\even}[1]{{#1}_{\textnormal{e}}}
\newcommand{\odd}[1]{{#1}_{\textnormal{o}}}
\newcommand{\sepsilon}{\slashed{\epsilon}}
\newcommand{\sdiv}{\slashed{\textnormal{div}}}
\newcommand{\scurl}{\slashed{\textnormal{curl}}}

\newcommand{\sexd}{\slashed{\textnormal{d}}}
\newcommand{\sn}{\slashed{\nabla}}
\newcommand{\sdso}{\slashed{\mathcal{D}}_1^{\star}}
\newcommand{\sdst}{\slashed{\mathcal{D}}_2^{\star}}
\newcommand{\sdo}{\slashed{\mathcal{D}}_1}
\newcommand{\sdt}{\slashed{\mathcal{D}}_2}
\newcommand{\slap}{\slashed{\Delta}}

\newcommand{\zslapinv}[1]{\slap^{-{#1}}_{\mathfrak{Z}}}

\newcommand{\qX}{\widetilde{X}}

\newcommand{\sX}{\slashed{X}}

\newcommand{\qP}{\widetilde{P}}

\newcommand{\qtau}{\widetilde{\tau}}
\newcommand{\qhattau}{\widehat{\widetilde{\tau}}}
\newcommand{\qtrtau}{\qtr\qtau}
\newcommand{\mtau}{\stkout{\tau}}
\newcommand{\stau}{\slashed{\tau}}
\newcommand{\shattau}{\widehat{\slashed{\tau}}}
\newcommand{\strtau}{\str\stau}

\newcommand{\sxi}{\slashed{\xi}}
\newcommand{\shatxi}{\widehat{\slashed{\xi}}}

\newcommand{\ssigma}{\slashed{\sigma}}
\newcommand{\srho}{\slashed{\rho}}

\newcommand{\qq}{\widetilde{q}}

\newcommand{\mdpart}[1]{{#1}_{l=0,1}}
\newcommand{\rpart}[1]{{#1}_{l\geq 2}}

\newcommand{\smcA}[1]{\slashed{\mathcal{A}}_f^{[#1]}}
\newcommand{\vmcA}[1]{\slashed{\mathcal{A}}_\xi^{[#1]}}
\newcommand{\tmcA}[1]{\slashed{\mathcal{A}}_\theta^{[#1]}}

\newcommand{\lin}[1]{\accentset{\scalebox{.6}{\mbox{\tiny\textnormal{(1)}}}}{#1}}

\newcommand{\Philin}{\lin{\Phi}}
\newcommand{\Psilin}{\lin{\Psi}}
\newcommand{\plin}{\lin{p}}
\newcommand{\qlin}{\lin{q}}
\newcommand{\quhatlin}{\lin{\underaccent{\check}q}}
\newcommand{\puhatlin}{\lin{\underaccent{\check}p}}
\newcommand{\trglin}{\textnormal{tr}_{{g_M}}\glin}

\newcommand{\qfnorwlin}{\lin{{\tilde{f}}}_{\diamond}}
\newcommand{\sfnorwlin}{\lin{\slashed{f}}_{\diamond}}

\newcommand{\Plin}{\lin{P}}
\newcommand{\Puhatlin}{\lin{\underaccent{\check}P}}

\newcommand{\flin}{\lin{f}}

\newcommand{\qglin}{\lin{\widetilde{g}}}
\newcommand{\qtrglin}{\qtr{\qglin}}
\newcommand{\qhatglin}{\lin{\widehat{\widetilde{g}}}}
\newcommand{\mglin}{\text{\sout{\ensuremath{\lin{g}}}}}
\newcommand{\strglin}{\str{\lin{\slashed{g}}}}
\newcommand{\shatglin}{\lin{\widehat{\slashed{g}}}}

\newcommand{\qplin}{\lin{\widetilde{p}}}
\newcommand{\slplin}{\lin{\slashed{p}}}

\newcommand{\glin}{\lin{g}}

\newcommand{\mfa}{\mathfrak{a}}
\newcommand{\mfm}{\mathfrak{m}}
\newcommand{\gkerr}{\glin_{\textnormal{Kerr}}}

\DeclareDocumentCommand\norm{m m o o o} {{\mathbb{#1}}_{\IfNoValueF{#4}{\mathsmaller{#4}}}^{\IfNoValueF{#3}{#3}}[#2]{\IfNoValueF{#5}{({#5})}}}

\newcommand{\inttwosphere}[3]{\int_{\twosphere_{{\taus_{#1}}, {#2}}}{#3}\,\sepsilon}

\newcommand{\mcalH}{\mathcal{H}}
\newcommand{\mcalI}{\mathcal{I}}
\newcommand{\mcalL}{\mathcal{L}}
\newcommand{\mcalN}{\mathcal{N}}
\newcommand{\mcalm}{\mathcal{M}}

\newcommand{\mfZ}{\slashed{\mathfrak{Z}}}

\newcommand{\uV}{{V}^{\perp}}
\newcommand{\oV}{\overline{V}}
\newcommand{\smfunsig}{\Gamma(\Sigma_0)}
\newcommand{\smonecovsig}{\Gamma(T^*\Sigma_0)}
\newcommand{\gammalin}{\lin{\gamma}}
\newcommand{\qexdL}{\qexd^{\mathsmaller{\mathcal{I}}}}
\newcommand{\qgamma}{\widetilde{\gamma}}
\newcommand{\mgamma}{\stkout{\gamma}}
\newcommand{\sgamma}{\slashed{\gamma}}
\newcommand{\seedspace}{\mathcal{D}}
\newcommand{\solnspace}{\mathcal{S}}
\newcommand{\solnmap}{\mathscr{S}}

\newcommand{\radseedspace}{\mathcal{D}_{\mathcal{R}}}

\newcommand{\iotalin}{\lin{\iota}}
\newcommand{\qiotalin}{\lin{\widetilde{\iota}}}
\newcommand{\qtriotalin}{\qtr\qiotalin}
\newcommand{\qhatiotalin}{\lin{\widehat{\widetilde{\iota}}}}
\newcommand{\miotalin}{\stkout{\lin{\iota}}}
\newcommand{\siotalin}{\lin{\slashed{\iota}}}
\newcommand{\striotalin}{\str\siotalin}

\newcommand{\qastrosunhat}{\widehat{\widetilde{\astrosun}}}
\newcommand{\sastrosunhat}{\widehat{\slashed{\astrosun}}}
\newcommand{\zetalin}{\lin{\zeta}}
\newcommand{\varphilin}{\lin{\varphi}}
\newcommand{\philin}{\lin{\phi}}
\newcommand{\psilin}{\lin{\psi}}
\newcommand{\Vlin}{\lin{V}}
\newcommand{\smfunsigrad}{\Gamma_{l\geq 2}(\Sigma_0)}
\newcommand{\smonesig}{\Gamma(T^*\Sigma_0)}

\newcommand{\gidnlin}{\lin{\underline{g}}}
\newcommand{\normtwosphere}[3]{||{#1}||^2_{{\twosphere_{{\taus_{#2}}, {#3}}}}}

\newcommand{\geomtwosphere}[2]{\twosphere_{{\taus_{#1}}, {#2}}}
\newcommand{\opmu}{(1+\mu)}
\newcommand{\ommu}{(1-\mu)}
\newcommand{\qm}{\mcalq}
\newcommand{\mcalq}{\mathcal{Q}}
\newcommand{\pa}{\partial_\alpha}
\newcommand{\pb}{\partial_\beta}
\newcommand{\pg}{\partial_\gamma}
\newcommand{\pA}{\partial_A}
\newcommand{\pB}{\partial_B}
\newcommand{\pC}{\partial_C}
\newcommand{\qomega}{\widetilde{\omega}}
\newcommand{\somega}{\slashed{\omega}}

\newcommand{\tr}{\textnormal{tr}_{g_M}}
\newcommand{\qhatgamma}{\widehat{\widetilde{\gamma}}}
\newcommand{\qtrgamma}{\qtr\qgamma}
\newcommand{\shatgamma}{\widehat{\slashed{\gamma}}}
\newcommand{\strgamma}{\str\sgamma}
\newcommand{\qzeta}{\widetilde{\zeta}}
\newcommand{\szeta}{\slashed{\zeta}}
\newcommand{\qV}{\widetilde{V}}
\newcommand{\sV}{\slashed{V}}

\newcommand{\ssharp}{\slashed{\sharp}}
\newcommand{\g}{g_M}
\newcommand{\mff}{\mathfrak{f}}
\newcommand{\mfg}{\mathfrak{g}}
\newcommand{\momega}{\text{\sout{\ensuremath{\omega}}}}
\newcommand{\mcala}{\slashed{\mathcal{A}}}

\title{The linear stability of the Schwarzschild solution to gravitational perturbations in the generalised wave gauge}
\author{Thomas William Johnson}
\affil{University of Cambridge\footnote{The majority of this work was carried out whilst the author was at Imperial College London.}}

\begin{document}

\maketitle

\begin{abstract}
We prove in this paper that the Schwarzschild family of black holes are linearly stable as a family of solutions to the system of equations that result from expressing the Einstein vacuum equations in a generalised wave gauge. In particular we improve on our recent work \cite{Johnsonlinstabschwarzold} by modifying the generalised wave gauge employed therein so as to establish asymptotic flatness of the associated linearised system. The result thus complements the seminal work \cite{DHRlinstabschwarz} of Dafermos--Holzegel--Rodnianski in a similar vein as to how the work \cite{LRstabmink} of Lindblad--Rodnianski complemented that of Christodoulou--Klainerman \cite{CKstabmink} in establishing the nonlinear stability of the Minkowski space. 
\end{abstract}

\tableofcontents

\section{Introduction}\label{Introduction}

The Schwarzschild exterior family of spacetimes $\Mg$ with $M>0$ comprise a 1-parameter family of Lorentizian manifolds with boundary that solve the Einstein vacuum equations of general relativity,
\begin{align}\label{introEE}
\text{Ric}[g]=0.
\end{align}
They each describe the region of spacetime exterior to the black hole region of a member of the Schwarzschild family \cite{Schwarzschild} of stationary black hole spacetimes with mass parameter $M$ -- the boundary of $\Mg$ then corresponds to the event horizon. The stability of this Schwarzschild exterior family as a family of solutions to \eqref{introEE} is thus fundamental to the physical significance of black holes:
\begin{question*}
Is $\Mg$ stable as a family of solutions to \eqref{introEE}?
\end{question*}

Although originally geometrically obscured by the diffeomorphism invariance of the theory, classical work \cite{CBlocwellpos} of Choquet-Bruhat showed that the correct way to pose this question is in the context of the associated hyperbolic initial value formulation of \eqref{introEE}. This question is further complicated however by the fact that the family $\Mg$ actually sits as a subfamily within the more elaborate 2-parameter family of stationary Kerr \cite{Kerr} exterior spacetimes $\big(\mcalm, g_{M,a}\big)$ with $|a|\leq M$. The stability of the Schwarzschild exterior family thus fits more correctly within the conjectured stability of the subextremal\footnote{The extremal case is more subtle -- see \cite{Aretakisextremal}.}  ($|a|<M$) Kerr exterior family:
\begin{conjecture*}
	The subextremal Kerr exterior family $\big(\mcalm, g_{M,a}\big)$ is stable as a family of solutions to \eqref{introEE}.
\end{conjecture*}
The precise mathematical formulation of this conjecture can be found in \cite{DHRscattering}. Note that as a consequence of their stationarity it is the exterior Kerr family itself which is posited to be stable as opposed to a single member of this family\footnote{This also agrees with the physical expectation that a nontrivial perturbation of a black hole should add both mass and angular momentum (the parameters $M$ and $a$ of the Kerr exterior family).}.

At the level of a \emph{global} statement about the Einstein vacuum equations \eqref{introEE} the above conjecture in particular demands that the maximal Cauchy development under \eqref{introEE} of smooth geometric data $(\Sigma, h, k)$ suitably close to the geometric data $\big(\Sigma_{M,a}, h_{M,a}, k_{M,a}\big)$ for a member of the subextremal Kerr exterior family possesses a complete future null infinity $\nullinf$ in addition to a non-empty future affine-complete null boundary $\eh$. Yet the only known mechanism for treating the nonlinearities in \eqref{introEE} so as to obtain global control of solutions is to exploit the dispersion provided by waves radiating towards $\nullinf$. Since moreover in 1+3 dimensions the expected rate of this dispersion is borderline it follows that one must identify a special structure in the nonlinear terms in \eqref{introEE} if this scheme is to prove suitable for resolving the conjecture.

One way of identifying this required structure is to express \eqref{introEE} relative to a \emph{generalised wave gauge}. For in this gauge the Einstein vacuum equations \eqref{introEE} reduce to a system of \emph{quasilinear wave equations}:
\begin{align}
\big(\widetilde{\Box}_{g, \og}g\big)_{\mu\nu}&=\mcalN_{\mu\nu}\big(g, \overline{\nabla}g, \onabla f(g)\big),\label{introwaveeqn}\\
g_{\mu\nu}(g^{-1})^{\xi o}\big(\Gamma^\nu_{\xi o}-\overline{\Gamma}^\nu_{\xi o}\big)&=\big(f(g)\big)_\mu\label{introwavegauge}.
\end{align}
Here $g$ and $\og$ are smooth Lorentzian metrics on a smooth manifold $\mcalL$ with $(x^\mu)$ any local coordinate system, $\mcalN$ is a nonlinear expression in its arguments and $f:\Gamma(T^2T^*\mcalL)\rightarrow \Gamma(T^*\mcalL)$ is a given map. We have also defined the wave operator $\widetilde{\Box}_{g, \og}:=(g^{-1})^{ab}\onabla_a\onabla_b$ with $\onabla$ the Levi-Civita connection of $\og$. The condition \eqref{introwavegauge} is thus equivalent to $g$ being in a generalised wave gauge \emph{with respect to $(\og, f)$} from which equation \eqref{introwaveeqn} follows after imposing \eqref{introEE} on $(\mcalL, g)$. Note that given \emph{any} Lorentzian metric $g$ on $\mcalL$ then there exists a wave-map operator $\Box_{g,\og}^f$ such that if $\phi:\mcalL\rightarrow\mcalL$ is a smooth solution to $\Box_{g,\og}^f\phi=0$ then $\phi^*g$ is in a generalised wave gauge with respect to $(\og, f)$.

Indeed the pioneering work of Lindblad--Rodnianski \cite{LRstabmink} established that the non-linearities in the coupled system \eqref{introwaveeqn}-\eqref{introwavegauge} when expressed on $\reals^4$ with $\og=\eta$ the Minkowski metric and $f=0$ satisfy a \emph{hierarchical} form of the \emph{weak null condition} -- see \cite{Keirweaknull} for the precise definition. This in principle provides sufficient structure so as for a ``small data global existence'' result to be established in a neighbourhood of geometric data for a globally hyperbolic, ``global'' solution $g_{\text{ini}}$ to the system \eqref{introwaveeqn}-\eqref{introwavegauge} purely by exploiting the dispersion embodied in a sufficiently robust statement of \emph{linear stability} for the solution $g_{\text{dyn}}$ one expects to approach in evolution. Such a scheme was successfully implemented by Lindblad--Rodnianski in \cite{LRstabmink} for the case where $g_{\text{ini}}=g_{\text{dyn}}=\og=\eta$. Here ``global existence'' of the constructed globally hyperbolic spacetimes $(\reals^4, g)$ was the statement that all causal geodesics in $(\reals^4, g)$ are future affine-complete. Since moreover the classical work \cite{CBlocwellpos} of Choquet-Bruhat established the local well-posedness of the generalised wave gauge with respect to $(\eta, 0)$ on any globally hyperbolic spacetime the result \cite{LRstabmink} also provided an additional proof of the fact that Minkowski space is nonlinearly stable as a solution to \eqref{introEE}, a statement which was originally provided by the monumental work of Christodoulou--Klainerman \cite{CKstabmink}. We note that the approach pioneered by Lindblad--Rodnianski in \cite{LRstabmink} has since been extended to various matter models \cite{LRstabmink}, \cite{SpeckstabminkEM}, \cite{LMstabminkEKG}, \cite{LTstabminkEV} and \cite{FJSstabminkEV} or different asymptotics \cite{Huneauthesis}. See also \cite{Limprovedasymptotics} and \cite{HVstabmink}. We further mention the recent mammoth work \cite{Keirweaknull} of Keir which establishes small data global existence results for a large class of nonlinear wave equations on $\reals^4$ satisfying the hierarchical weak null condition that includes the Einstein vacuum equations \eqref{introEE} expressed relative to a generalised wave gauge with respect to $(\eta, 0)$ as a special case. 

The statement of linear stability implicitly exploited by Lindblad--Rodnianski in \cite{LRstabmink} was that solutions to the linearised system behave like solutions to the free scalar wave equation $\Box_{\eta}\psi=0$. One is therefore lead to consider the following question:
\begin{question*}
	Is there a pair $(\og, f)$ for which residual pure gauge and linearised Kerr normalised solutions to the linearisation of  \eqref{introwaveeqn}-\eqref{introwavegauge} about any fixed member of the subextremal Kerr exterior family $\big(\mcalm, g_{M,a}\big)$ behave like solutions to the free scalar wave equation $\Box_{g_{M,a}}\psi=0?$
\end{question*}
Indeed the idea would then be to use the dispersion embodied in a positive answer to the above question to treat the nonlinear terms in \eqref{introwaveeqn}-\eqref{introwavegauge} in a similar way as to the treatment employed by Lindblad--Rodnianski in \cite{LRstabmink} -- that solutions to $\Box_{g_{M,a}}\psi=0$ do indeed disperse was established in the seminal \cite{DRS-Rwaveeqnkerr}. Combining this with a statement of well-posedness for the associated generalised wave gauge would then yield a positive resolution to the conjectured stability of the Kerr family. Note that the gauge-normalisation in the above merely reflects the fact that there is no unique way to express \eqref{introEE} relative to a generalised wave gauge with respect to $(\og, f)$ without first imposing extra gauge conditions -- one is therefore free to exploit this freedom in view of the fact that the stability of the Kerr family is a statement about the \emph{maximal} Cauchy development of geometric data under \eqref{introEE}. Moreover the Kerr-normalisation reflects the fact that solutions to the linearised system should in fact only disperse to a stationary linearised Kerr solution.

In this paper we will provide a positive answer to this question for any fixed member of the Schwarzschild family $\Mg$ with $(\og, f)=(g_M, \fgaumap)$ where $\fgaumap$ is an explicit \emph{gauge-map}. Note that the gauge-map $\fgaumap$ is $\reals$-linear and satisfies $\fgaumap(g_M)=0$ each of which ensure that the linearisation of the system \eqref{introwaveeqn}-\eqref{introwavegauge} around $g_M$ is well defined. The statement we are to prove is then given as follows.

\begin{theorem*}
	We consider the equations of linearised gravity around Schwarzschild, namely the system of equations that result from linearising the Einstein vacuum equations \eqref{introEE}, as expressed in a generalised wave gauge with respect to the pair $(g_M, \fgaumap)$, about a fixed member of the Schwarzschild exterior family $\Mg$. Then all solutions arising from smooth, asymptotically flat and gauge-normalised seed data prescribed on a Cauchy hypersurface $\Sigma_0$:
	\begin{enumerate}[i)]
		\item remain uniformly bounded on $\mcalm$ (up to and including the boundary $\eh$) and in fact decay at an inverse polynomial rate to a linearised Kerr solution which is itself determined from  initial data on $\Sigma_0$
		\item remain asymptotically flat on $\mcalm$.
	\end{enumerate}
\end{theorem*}
In the above seed data is a collection of freely prescribed quantities on $\Sigma$ which fully parametrises the solution space -- note that full Cauchy data cannot be prescribed in view of constraints inherited from the nonlinear theory. Gauge normalisation of the seed then reflects the fact that we obtain decay only after the addition of a residual pure gauge solution to a general solution which serves to normalise the seed of this latter solution. Here residual pure gauge solutions to the equations of linearised gravity arise from pulling back $\g$ by infinitesimal diffeomorphisms preserving the generalised wave gauge with respect to $(g_M, \fgaumap)$. In addition the linearised Kerr solutions of the theorem are those that arise from linearising the subextremal Kerr exterior family $g_{M,a}$ in the parameters. We stress therefore that the conclusion of our theorem is thus consistent with the statement that the \emph{maximal} Cauchy development of suitably small perturbations of the geometric data $(\Sigma, h, k)$ induced by $\g$ on $\Sigma$ under \eqref{introwaveeqn}-\eqref{introwavegauge} with $(\og, f)=(g_M, \fgaumap)$ dynamically asymptotes to a nearby member of the subextremal Kerr exterior family. We moreover emphasize that part $i)$ of our theorem should be viewed as a boundedness statement at the level of certain natural energy fluxes which does not lose derivatives -- see section \ref{OVTheorem2} of the overview for a more comprehensive version.

We in addition stress that the choice of the gauge-map $\fgaumap$ is crucial if the above theorem is to hold. Indeed whereas the generalised wave gauge as a whole reveals a special structure in the non-linear terms the gauge-map $\fgaumap$ is designed to unlock a special structure in the linear terms. We note also our previous \cite{Johnsonlinstabschwarzold} which provided a version of the above theorem without the asymptotic flatness criterion of part $ii)$ and with a different choice of gauge-map $f$. 

To understand this special structure we briefly discuss the proof of the theorem-- for further details one should consult the overview. First we show that any smooth solution to the equations of linearised gravity arising from initial data as in the theorem statement can be decomposed into the sum of a linearised Kerr solution, a residual pure gauge solution and a symmetric 2-covariant tensor field  whose components are given by derivatives of two scalar waves $\big(\Philin, \Psilin\big)$ on $\Mg$ each of which both completely decouple and vanish for all linearised Kerr and residual pure gauge solutions. We then show using the methods developed by Dafermos--Rodnianski in \cite{DRlecturenotes}--\cite{DRredshift} for analysing the scalar wave equation $\Box_{g_M}\psi=0$ on $\mcalm$ that the asymptotic flatness of the initial data implies that the invariant pair $\big(\Philin, \Psilin\big)$ decay at an inverse polynomial rate towards the future on $\mcalm$. Moreover, the ``gauge-conditions'' on $\Sigma$ ensure that the residual pure gauge part of the solution vanishes. It is then a simple matter to show that the decay bounds on the pair $\big(\Philin, \Psilin\big)$ yield the desired bounds on the solution of the theorem statement.

Key to the above decomposition is the gauge-map $\fgaumap$. Indeed consider instead the system of equations that result from linearising \eqref{introwaveeqn}-\eqref{introwavegauge} with $(\og, f)=(g_M, 0)$ about $\Mg$. Then the best one can show is that a general solution decomposes as a linearised Kerr solution plus a solution determined by six scalar waves $\big(\Philin, \Psilin,\plin, \puhatlin, \qlin, \quhatlin\big)$ with the invariant pair $\big(\Philin, \Psilin\big)$ appearing as inhomogeneous terms in the subsystem satisfied by $\big(\plin, \puhatlin, \qlin, \quhatlin\big)$\footnote{One shows this by first extracting a solution to Maxwell's equations in a generalised Lorentz gauge from the associated linearised system. We leave it to the reader to confirm this however.}. Since by definition it is impossible to `gauge-away' the pair $\big(\Philin, \Psilin\big)$ it follows that one is forced to derive decay estimates on a coupled system of \emph{linear wave equations} -- the subsequent loss of decay this yields is then problematic for potential nonlinear applications\footnote{Perhaps more important however is the fact that one would then expect a loss of derivatives in the associated energy norms due to an `amplification' of the celebrated trapping effect on Schwarzschild when passing from the scalar wave equation to \emph{linear systems of waves}.}. In particular, solutions to this linearised system do not behave like solutions to the free scalar wave equation. More generally then, we see that the purpose of introducing the gauge-map $\fgaumap$ is to negate a certain undesirable coupling in the linearised system\footnote{Note that this coupling occurs due to the fact that the linearised operator $g_M^{ab}\nabla_a\nabla_b$, with $\nabla$ the Levi-Civita connection of $g_M$, has tensorial structure.}.

The fact that one can extract two fully decoupled scalar waves from the linearised Einstein equations around Schwarzschild has been well known in the literature since the works of Regge--Wheeler in \cite{RWrweqn} and Zerilli in \cite{Zzeqn} where it was discovered that two \emph{gauge-invariant} scalars decouple from the full system into the Regge--Wheeler and Zerilli equations respectively. Decay estimates on these quantities were subsequently derived in the independent works \cite{BSdecayrw}, \cite{Holzegelultschwarz} and \cite{Johnsondecayz}, \cite{HKWlinstabschwarz} respectively although we shall reprove them. Moreover, it is also well known \cite{RWrweqn} that given any smooth solution to the \emph{linearised Einstein equations around Schwarzschild} then one could subtract from it a pure gauge and linearised Kerr solution such that the resulting solution can be expressed in terms of these Regge--Wheeler and Zerilli quantities. It was not known however until our previous \cite{Johnsonlinstabschwarzold} that one could in fact realise this ``Regge--Wheeler gauge'' within the context of the linear theory associated to the Einstein vacuum equations expressed in a generalised wave gauge. It was also not known how to modify the gauge so as to achieve asymptotic flatness. An additional aspect of our work is therefore to both modify and identify this remarkably useful ``Regge--Wheeler gauge'' as a ``gauge'' within a formulation of linearised gravity around Schwarzschild that has direct consequences for the associated nonlinear theory\footnote{In particular note that no dispersion for the linearisation of \eqref{introEE} can help overcome the structural deficit in the nonlinear terms!}.

Of course imposing a generalised wave gauge is not the only way to study the nonlinear terms in \eqref{introEE}. Indeed the monumental work of Christodoulou--Klainerman \cite{CKstabmink} showed that expressing \eqref{introEE} relative to a \emph{double null frame} reveals a certain \emph{null structure} in the nonlinear terms which is in principle sufficiently good so as to treat via the dispersion embodied in a statement of linear stability. This latter statement on Schwarzschild was provided in the seminal paper \cite{DHRlinstabschwarz} of Dafermos--Holzegel--Rodnianski. It is interesting to note however that there decay was obtained only in a ``gauge'' normalised along a `future' hypersurface (namely the horizon) in contrast to the ``initial-data gauge'' employed here. This discrepancy in ``gauge'' is consistent with the comparison between \cite{LRstabmink} with \cite{CKstabmink} and in fact provided one of the main reasons why the proof of the former was dramatically simply than that of the latter.

It is moreover worthwhile to contrast the difficulty involved in proving our theorem with that of proving the analogous result of Dafermos--Holzegel--Rodnianski in \cite{DHRlinstabschwarz}. To this end, we briefly summarise their approach as follows. They begin by extracting a pair of \emph{gauge-invariant} quantities $\big(\Plin, \Puhatlin\big)$ from the linearised system under consideration, each of which completely decouple into a wave equation that can be analysed using the methods of \cite{DRlecturenotes}--\cite{DRredshift}. A collection of quantities $X$ are then identified which fully determine solutions to the linearised equations in the sense that decay bounds on the former translate to decay bounds on the latter. An additional feature of the collection $X$ is that it can be arranged hierarchically so that each member of $X$ is estimable from a previous member of the hierarchy by solving a \emph{transport equation}. What's more, the pair $\big(\Plin, \Puhatlin\big)$ serve as originators for this hierarchy. Dafermos--Holzegel--Rodnianski are then able to upgrade decay bounds on the pair $\big(\Plin, \Puhatlin\big)$ to decay bounds for the full system by ascending this hierarchy \emph{when supplemented with certain `gauge conditions' along the horizon} -- it is ascending this hierarchy that comprises the bulk of the work in \cite{DHRlinstabschwarz}. In contrast, the structure of the linearised system we consider is such that, under a judicious choice of `gauge', the analogous task of estimating the ``gauge-dependent'' part of the solution is trivial\footnote{In particular, it seems that one is unable to replicate this in the approach of \cite{DHRlinstabschwarz} as the underlying equations do not permit the gauge-invariant pair $\big(\Plin, \Puhatlin\big)$ to formally decouple from the collection $X$.}. 

We turn now to a brief discussion of other results related to our work. We begin by noting that the first instance of employing a generalised wave gauge as a means to analyse the Einstein equations was that of Friedrich's in \cite{Friedrichgenwavgau}. Moreover, a more recent application of said gauge can be found in the seminal work \cite{HVstabkerrDS} of Hintz--Vasy where the nonlinear stability of the subextremal Kerr--de Sitter exterior family of black holes was established for small values in the rotation parameter\footnote{Note that the exponential decay expected from the presence of a positive cosmological constant makes this problem less complex than that of the analogous problem for the subextremal Kerr exterior family.}. See also the paper \cite{HintzstabKNdS} of Hintz for the analogous result for the subextremal Kerr--Newman--de Sitter exterior family. In regards to the Schwarzschild exterior family, recent work \cite{HKWlinstabschwarz} of Hung--Keller--Wang showed that sufficiently regular solutions to the linearisation of \eqref{introEE} about $\Mg$ decay to the sum of a pure gauge and linearised Kerr solution. Here the ``gauge'' adopted was a so-called Chandrasekhar gauge which, similarly to the ``gauge'' we adopt, expresses the solution in terms of the Regge--Wheeler and Zerilli quantities. In contrast to our work however it is not clear how one is to exploit this decay statement in order to treat the nonlinear terms in \eqref{introEE}. In addition, recent work \cite{Hunglinstabschwarz} of Hung gave a partial result towards establishing a decay statement for solutions to the linearisation of \eqref{introwaveeqn}-\eqref{introwavegauge} with $(g_M, f)=(g_M, 0)$ and which arise from a restricted class of initial data on $\Mg$. Here the restriction on the data ensures that the solutions under consideration are effectively governed by solutions to the Regge--Wheeler equation\footnote{This class of ``perturbations'' are more commonly known in the literature as odd perturbations owing to how they transform under a parity transformation of Schwarzschild. Note also that we attach the precondition partial as decay was only obtained in \cite{Hunglinstabschwarz} for the $l\geq 3$ angular frequencies.}. Whilst we emphasize that the full problem is significantly more complicated in any case our previous discussion indicates the potential obstructions to employing such a gauge in the fully nonlinear problem. Finally, we collect the following references pertaining to the stablilty problem on Schwarzschild at large: \cite{Jezierskizerilli}, \cite{STschwarz}, \cite{Moncriefschwarz}, \cite{Moncrieflinstab}, \cite{GStwoplustwo}, \cite{COSschwarz}, \cite{Holzegelconslaw}, \cite{ABWzer}, \cite{KWlinstabschwarz}, \cite{MPlinstabschwarz}, \cite{Vishveshwaralinstabschwarz}, \cite{BPlinstabschwarz}, \cite{Chandrasekharlinstabschwarz}, \cite{Chandrasekharbook}, \cite{Dottilinstabschwarz}, \cite{FSteu}, \cite{FMschwarz} and \cite{Tteu}.

Let us now conclude this introduction where it began, namely the question of the nonlinear stability of the Schwarzschild exterior family. Indeed, in view of the fact that one must linearise about the solution one expects to approach, providing a positive answer to this question would require first upgrading the linear theory established here to the subextremal Kerr exterior family (with small rotation parameter $a$). However, in \cite{DHRlinstabschwarz} Dafermos--Holzegel--Rodnianski formulated a restricted nonlinear stability conjecture regarding the Schwarzschild exterior family which should in principle be sufficient to resolve by exploiting the rate of dispersion embodied in part i) of our Theorem to treat the nonlinearities present in expressing \eqref{introEE} relative to a generalised wave gauge. We note that a proof of said conjecture in the symmetry class of axially symmetric and polarised perturations has recently been announced by Klainerman--Szeftel over a series of three papers, the first of which is to be found here \cite{KSstabschwarzsymm}.

\section{Overview}\label{Overview}

\noindent We shall now give a complete overview of this paper. 

We begin in section \ref{OVTheequationsoflinearisedgravityaroundSchwarzschild} by presenting the equations of linearised gravity around Schwarzschild, namely the system of equations that result from expressing the Einstein vacuum equations in a generalised wave gauge and then linearising about a fixed member of the Schwarzschild family. Then in section \ref{OVSpecialsolutionstothelinearisedequations} we discuss special solutions to the equations of linearised gravity arising from both residual gauge freedom and the existence of an explicit family of stationary solutions in the nonlinear theory. Then in section \ref{OVDecouplingtheequationsoflinearisedgravity} we exploit certain classical insights to show that a smooth solution to the equations of linearised can be decomposed into the sum of the special solutions of the previous section with a solution determined by derivatives of two scalar waves satisfying the Regge--Wheeler and Zerilli equations respectively. Then in section \ref{OVWellposednessoftheCauchyproblem} we use this decomposition to develop a well-posedness theory for the equations of linearised gravity. Then in section \ref{OVGaugenormalisationofinitialdata} we discuss initial-data normalised solutions to the equations of linearised gravity. A decay statement for these solutions will follow from a decay statement for solutions to the Regge--Wheeler and Zerilli equations. In section \ref{OVAside:thescalarwaveequationontheSchwarzschildexteriorspacetime} we thus make an aside to discuss the techniques developed for establishing a decay statement for solutions to the scalar wave equation on Schwarzschild. Finally in section \ref{OVThemaintheorems} we give rough statements and outlines of the proofs of the main two theorems of this paper, the first of which concerns a decay statement for solutions to the Regge--Wheeler and Zerilli equations and the second of which concerns a decay statement for
initial-data normalised solutions to the equations of linearised gravity. 

\subsection{The equations of linearised gravity around Schwarzschild}\label{OVTheequationsoflinearisedgravityaroundSchwarzschild}

In order to present the linearised equations we first define the Schwarzschild exterior family of spacetimes in section \ref{OVTheSchwarzschildexteriorfamily}. Then in section \ref{OVTheEinsteinvacuumequationsasexpressedinageneralisedwavegauge} we introduce the generalised wave gauge and present how the Einstein vacuum equations appear in such a gauge. Finally in section \ref{OVTheequationsoflinearisedgravity} we present the equations of linearised gravity.

This section of the overview corresponds to section \ref{TheequationsoflinearisedgravityaroundSchwarzschild} in the main body of the paper.

\subsubsection{The Schwarzschild exterior family}\label{OVTheSchwarzschildexteriorfamily}

Let $M>0$ and let $\mcalm$ be the manifold with boundary
\begin{align*}
\mcalm:=\reals_{t^*}\times [1, \infty)_x\times\twosphere
\end{align*}
which we equip with the 1-parameter family of Lorentzian metrics defined by
\begin{align*}
g_M:=-\bigg(1-\frac{1}{x}\bigg)\exd {t^*}^2+\frac{4M}{x}\exd t^*\exd x+4M^2\bigg(1+\frac{1}{x}\bigg)\exd x^2+4Mx^2\,\roundmetric
\end{align*}
with $\roundmetric$ the unit metric on the round sphere. Then the 1-parameter family of Lorentzian manifolds with boundary $\Mg$ define the Schwarzschild exterior family of spacetimes. They each satisfy the Einstein vacuum equations,
\begin{align}\label{OVeineqn}
\textnormal{Ric}[g_M]=0
\end{align}
and arise as the maximal Cauchy development under \eqref{OVeineqn} of the \emph{asymptotically flat} geometric data
\begin{align*}
\big(\Sigma, h_M, k_M\big):=\Big([1, \infty)_x\times\twosphere, 4M^2\big(1+\tfrac{1}{x}\big)\exd x^2+4Mx^2\,\roundmetric, \big({1+\tfrac{1}{x}}\big)^{-\tfrac{1}{2}}\big(\tfrac{M}{x^2}\big({2+\tfrac{1}{x}}\big)\exd x^2-2M\roundmetric\big)\Big)
\end{align*}
subject to the embedding $i(\Sigma)=\Sigma_0:=\{0\}\times\Sigma.$ In particular, observe that the boundary
\begin{align*}
\eh:=\big\{p\in\mcalm | x(p)=1\big\}
\end{align*}
is a \emph{null hypersurface}. Moreover as the causal vector field $\pt$ is manifestly Killing it follows that the family $g_M$ are both static and spherically symmetric.

\subsubsection{The Einstein vacuum equations as expressed in a generalised wave gauge}\label{OVTheEinsteinvacuumequationsasexpressedinageneralisedwavegauge}

Let now $\bg$ be a smooth Lorentzian metric on $\mcalm$ and let $\fgaumap:\smtwocov\rightarrow \smonecov$ be the $\reals$-linear map defined as in section \ref{Thegeneralisedwavegaugewithrespecttothepair}. We define the connection tensor $C_{\bg, g_M}\in\smthreecov$ between $\bg$ and $\g$ according to
\begin{align*}
(C_{\bg, \g})_{abc}=\frac{1}{2}\Big(2\nabla_{(b}\bg_{c)a}-\nabla_a\bg_{bc}\Big)
\end{align*}
with $\nabla$ the Levi-Civita connection of $\g$, noting therefore that $C_{\g, \g}=0$. Then we say that $\boldsymbol{g}$ is in a generalised wave gauge with respect to the pair $(\g,\fgaumap)$ iff 
\begin{align}\label{OVwavegauge}
(\boldsymbol{g^{-1}})^{bc}(C_{\bg, \g})_{abc}=\fgaumap(\bg)_a.
\end{align}
Assuming this to be the case then the Einstein vacuum equations on $\bg$ reduce to a quasilinear wave equation having the schematic representation
\begin{align}
\widetilde{\Box}_{\bg, g_M}\bg+C_{\bg, g_M}\cdot C_{\bg, g_M}+\textnormal{Riem}\cdot\boldsymbol{g}\cdot\bg&=\bg\cdot{\nabla}\fgaumap(\bg)\label{OVeineqningenwavgaufull},\\
(\boldsymbol{g^{-1}})\cdot C_{\bg, g_M}&=\fgaumap(\bg)\label{OVwavegaugefull}
\end{align}
where Riem is the Riemann tensor of $g_M$.

\subsubsection{The equations of linearised gravity}\label{OVTheequationsoflinearisedgravity}

Observing from section \ref{Thegeneralisedwavegaugewithrespecttothepair} that by construction $\fgaumap(g_M)=0$ we have that $g_M$ defines a solution to the system \eqref{OVeineqningenwavgaufull}-\eqref{OVwavegaugefull}. Pursuing the formal linearisation procedure developed in section \ref{Theformallinearisationoftheequationsofsection} one then finds that the linearisation of \eqref{OVeineqningenwavgaufull}-\eqref{OVwavegaugefull} about the solution $g_M$ is
\begin{align}
\big(\Box_{g_M}\glin\big)_{ab}-2\textnormal{Riem}\indices{^c_{ab}^d}\glin_{cd}&=2\nabla_{(a}\flin_{b)},\label{OVeqnlingrav1}\\
\nabla^b\glin_{ab}-\frac{1}{2}\nabla_a\trglin&=\flin_a\label{OVeqnlingrav2}.
\end{align}
 Here $\Box_{g_M}=\nabla^a\nabla_a$, $\trglin=g_M^{ab}\glin_{ab}$ and we have defined
\begin{align*}
\flin:=\fgaumap(\glin).
\end{align*}
Note in particular that one exploits the linearity of $\fgaumap$ over $\reals$ to derive the above.

The equations of linearised gravity on $\Mg$ thus describe a \emph{tensorial system of linear wave equations \eqref{OVeqnlingrav1} coupled with the divergence relation \eqref{OVeqnlingrav2}}\footnote{We remark that these are nothing but the linearised Einstein equations on $\Mg$ as expressed in a \emph{generalised} Lorentz gauge. See the book of Wald \cite{Waldbook}.}.\newline

\subsection{Special solutions to the equations of linearised gravity}\label{OVSpecialsolutionstothelinearisedequations}

One has the aim of establishing a \emph{decay} statement for solutions to the equations of linearised gravity. This is complicated however by the existence of both geometrically spurious and stationary solutions to the linearised system which we discuss now.

This section of the overview corresponds to section \ref{Specialsolutionstotheequationsoflinearisedgravity} in the main body of the paper.\newline

Let $\mfm, \mfa\in\reals$ be fixed and define the 1-parameter family of functions $M_\epsilon:=M+\epsilon\cdot\mfm, a_\epsilon:=a+\epsilon\cdot\mfa$. We subsequently consider the following two-parameter family of Kerr exterior metrics \cite{DRwaveeqnkerrII} on $\mcalm$, neglecting to write down higher than linear terms in $a_\epsilon$\footnote{Here we identify the $t^*, x$ coordinates on $\mcalm$ with the rescaled Kerr-star coordinates $t^*, \tfrac{r}{M+\sqrt{M^2-a^2}}$ of (21) in \cite{DRwaveeqnkerrII}.}:
\begin{align}\label{OVeqnkerrmetric}
g_{M_\epsilon, a_\epsilon}:=g_{M_\epsilon}-2a_\epsilon Y{\otimes}\Big(\tfrac{1}{x}\exd t^*+2M_\epsilon\big(1+\tfrac{1}{x}\big)\exd x\Big)+o\big(a_\epsilon^2\big)
\end{align}
where $Y\in\smonecov$ is such that $Y=\sin^2\theta\exd\varphi$ in spherical coordinates $(\theta, \varphi)$ on $\twosphere$. We then assume the following:
\begin{enumerate}[i)]
	\item there exists a 1-parameter family of diffeomorphims $\phi_\epsilon:\mcalm\rightarrow\mcalm$, with $\phi_0=\Id$, such that $\phi_\epsilon^*g_{M_0}$ is in a generalised wave gauge with respect to the pair $(g_M, \fgaumap)$
	\item for each $\epsilon$ there exists a diffeomorphism $\phi_\epsilon:\mcalm\rightarrow\mcalm$ such that $\phi_\epsilon^*g_{M_\epsilon, a_\epsilon}$ is in a generalised wave gauge with respect to the pair $(g_M, \fgaumap)$
\end{enumerate}
Diffeomorphism invariance of \eqref{OVeineqn} thus yields that $\phi_\epsilon^*g_{M_0}$ and $\phi_\epsilon^*g_{M_\epsilon, a_\epsilon}$ each comprise a 1-parameter family of solutions to the system of equations that result from expressing \eqref{OVeineqn} in a generalised wave gauge with respect to the pair $(g_M, \fgaumap)$. This leads to the following:

\begin{proposition*}
Let $\mfm, \mfa\in\reals$ and let $V\in\Gamma(T^*M)$ satisfy
\begin{align*}
\Box_{g_M}V=\fgaumap(\mcalL_{V^\sharp }g_M).
\end{align*}
Then the following are smooth solutions to the equations of linearised gravity \eqref{OVeqnlingrav1}-\eqref{OVeqnlingrav2}:
\begin{align*}
\glin_{\mfm,\mfa}:&=\tfrac{4\mfm}{ x}\exd t^*\exd x+8M\mfm\big(1+\tfrac{1}{x}\big)\exd x^2-2\mfa Y\otimes\Big(\tfrac{1}{x}\exd t^*+2M\big(1+\tfrac{1}{x}\big)\exd x\Big)+8M\mfm x^2\roundmetric,\\
\glin_V:&=\mcalL_Vg_M.
\end{align*}
\end{proposition*}
The above can be verified explicitly from the equations of linearised gravity. Note in particular that the 1-parameter family of metrics $g_{M_\epsilon, a_\epsilon}$ is in a generalised wave gauge with respect to the pair $(g_M, \fgaumap)$ to first order in $\epsilon$ -- this is a consequence of how the map $\fgaumap$ was defined.

We call the first class of solutions \emph{residual pure gauge solutions} due to the fact that they arise from the potential for residual gauge freedom in the nonlinear theory. Conversely, we call the second class of (stationary) solutions \emph{linearised Kerr solutions}, the nomenclature in this instance being clear. Note this latter class of solutions actually extends to a \emph{four-parameter} family of solutions to the equations of linearised gravity as a consequence of the spherical symmetry of $\Mg$ -- see \cite{DHRlinstabschwarz} for further discussion. It is this extended family that we refer to when referencing the linearised Kerr family in the remainder of the overview.

A first version of our main theorem is then the following: \textbf{we prove that all sufficiently regular solutions to the equations of linearised gravity decay towards the future on $\Mg$ to the sum of a residual pure gauge and linearised Kerr solution.} Included in this statement is a well-posedness theory both for the equations of linearised gravity and the class of residual pure gauge solutions to the former.

Note this statement is consistent with the statement that the maximal Cauchy development under \eqref{OVeineqningenwavgaufull}-\eqref{OVwavegaugefull} of suitably small perturbations of the geometric data $(\Sigma, h_M, k_M)$ dynamically asymptotes to a member of the subextremal Kerr exterior family.

\subsection{Decoupling the equations of linearised gravity}\label{OVDecouplingtheequationsoflinearisedgravity}

One expects that establishing the above decay statement is sensitive to ``gauge''. Further complications are provided by the tensorial structure of the equations of linearised. We shall now discuss how both these issues are naturally coupled and can be simultaneously resolved by exploiting classical insights into the linearised equations.

This section of the overview corresponds to section \ref{DecouplingtheequationsoflinearisedgravityuptoresidualpuregaugeandlinearisedKerrsolutions:theRegge--WheelerandZerilliequations} in the main body of the paper.\newline

It is natural to search for linearised quantities which vanish for both the special solutions of the previous section. Indeed it is necessary that such {invariant quantities} decay if a decay statement for the equations of linearised gravity is to hold in some ``gauge''. It is moreover natural to look for scalar versions of these quantities as a means of mitigating the tensorial structure of the linearised system. Remarkably two such invariant scalars $\Philin$ and $\Psilin$ exist which actually decouple from the full system into the Regge--Wheeler and Zerilli equations respectively (with $r=2M$):
\begin{align}\label{OVeqnRW}
\Box_{g_M}\big(r^{-1}\Philin\big)=-\frac{8}{r^3}\frac{M}{r}\Philin,
\end{align}
\begin{align}\label{OVeqnZ}
\Box_{g_M}\big(r^{-1}\Psilin\big)=-\frac{8}{r^3}\frac{M}{r}\Psilin+\frac{24}{r^4}\frac{M}{r}(r-3M)\zslapinv{1}\Psilin+\frac{72}{r^4}\frac{M}{r}\frac{M}{r}(r-2M)\zslapinv{2}\Psilin.
\end{align}
Here, $\zslapinv{p}$ is the inverse of the operator $r^2\slap+2-\frac{6M}{r}$
applied $p$-times, with $r^2\slap$ the spherical Laplacian. Note that $\zslapinv{p}$ is well defined over the space of smooth functions on $\mcalm$ supported on the $l\geq 2$ spherical harmonics (see the bulk of the paper for the definition)-- that $\Psilin$ (and indeed $\Philin$) are supported on the $l\geq 2$ spherical harmonics is a consequence of the fact they were constructed so as to vanish for all linearised Kerr solutions. To see that $\Philin$ and $\Psilin$ do indeed decouple in such a manner, see Theorem \ref{thmgaugeinvariantquantintermsofRWandZ}.

The decoupling of $\Philin$ and $\Psilin$ will play a fundamental role in our work. Indeed, the key analytical point is that a decay statement for the the two equations \eqref{OVeqnRW} and \eqref{OVeqnZ} can be established using the techniques developed for establishing a decay statement for the scalar wave equation $\Box_{g_M}\psi=0$ on $\mcalm$. Moreover it turns out that these decay statements will actually provide all one needs to establish a decay statement for the equations of linearised gravity due to the following proposition proved in section \ref{DecomposingageneralsolutiontotheequationsoflinearisedgravityintothesumofaresidualpuregaugeandlinearisedKerrsolutionandasolutiondeterminedbytheRegge--WheelerandZerilliequations}.

\begin{proposition}\label{OVpropdecoupling}
Let $\glin$ be a smooth solution to the equations of linearised gravity. Then there exists a linear map $\gamma$, a residual pure gauge solution $\glin_V$ and a linearised Kerr solution $\glin_{\textnormal{Kerr}}$ such that
\begin{align*}
\glin-\glin_V-\glin_{\textnormal{Kerr}}=\gamma(\Psilin, \Philin).
\end{align*}
Moreover, $\gamma$ satisfies the bound 
\begin{align*}
|r\gamma(\Psilin, \Philin)|\lesssim |\partial^2\Psilin|+|\partial^2\Philin|.
\end{align*}
\end{proposition}
It therefore follows that decay bounds on $\Psilin$ and $\Philin$ immediately yield decay bounds on the normalised solution $\glin-\glin_V-\glin_{\textnormal{Kerr}}$!\newline

We emphasize that Proposition \ref{OVpropdecoupling} only holds as a consequence of the fact that the gauge-map $\fgaumap$ appears in the definition of the equations of linearised gravity -- this is in fact the sole reason for its presence. Consequently, to explain how we identified such a gauge-map our procedure was as follows: first one studies the linearised system \eqref{OVeqnlingrav1}-\eqref{OVeqnlingrav2} defined with respect to a general map $f:\smtwocov\rightarrow\smonecov$. One then identifies a general decomposition of solutions as given above but with $\glin_V$ replaced with $\mcalL_Xg_M$ for $X\in\smonecon$ -- this is in fact easy to see from how the map $\gamma$ is defined (see section \ref{Theconnectionwiththeequationsoflinearisedgravity}). The desired gauge-map $f$ is then constructed by demanding that the linearised system associated to $f$ imposes on $\mcalL_Xg_M$ the equation
\begin{align*}
\Box_{g_M}X^\flat=f(\mcalL_Xg_M).
\end{align*}
Key to the above procedure is the fact that the remarkable decoupling of $\Philin$ and $\Psilin$ into the Regge--Wheeler and Zerilli equations holds for \emph{any} solution of the linearisation of \eqref{eineqn} around $g_M$, a fact which was originally discovered by Regge--Wheeler \cite{RWrweqn} and Zerilli \cite{Zzeqn} in the context of a full ``mode'' and spherical harmonic decomposition of the linearised Einstein equations around Schwarzschild\footnote{See \cite{COSschwarz} for the non-modal, covariant derivation.}.

\subsection{Well-posedness of the Cauchy problem}\label{OVWellposednessoftheCauchyproblem}

The insights of the previous section allows in particular for a well-posedness theory for the equations of linearised gravity to be developed.

This section of the overview corresponds to section \ref{Initialdatandwellposednessfortheequationsoflinearisedgravity} in the main body of the paper.\newline

In view of the existence of constraints in the linear theory, in particular those inherited from the Gauss--Codazzi equation, the appropriate Cauchy problem for the equations of linearised gravity is to construct unique solutions from freely prescribed \emph{seed} data on the initial Cauchy hypersurface $\Sigma_{0}$. Consequently, a suitable notion of seed data is provided by a collection of Cauchy data for the Regge--Wheeler and Zerilli equations, Cauchy data for the residual pure gauge equation and parameters of a linearised Kerr solution with Proposition \ref{OVpropdecoupling} then determining a canonical solution map. This yields:

\begin{theorem*}
Let $\seedspace^{af}$ denote the vector space of smooth, asymptotically flat seed data and let $\solnspace$ denote the vector space of smooth solutions to the equations of linearised gravity. Then there exists an isomorphism $\solnmap:\seedspace^{af}\rightarrow\solnmap(\seedspace^{af})\subset\solnspace$.
\end{theorem*}

See section \ref{Thewell-posednesstheorem} for details and section \ref{Pointwiseasymptoticallyflatseeddata} for the definition of asymptotically flat seed. In particular, we note that any smooth solution to the equations of linearised for which the associated quantities $\Philin$ and $\Psilin$ satisfy the necessary regularity for the decay bounds we establish for the Regge--Wheeler and Zerilli equations in section \ref{Theorem1:BoundednessanddecayforsolutionstotheRegge--WheelerandZerilliequations} to hold must lie in the space  $\solnmap(\seedspace^{af})$. Our notion of seed thus parametrises the space of ``admissible'' smooth solutions to the equations of linearised gravity -- in fact, Theorem \ref{thmwellposedness} shows that our smooth seed data actually parametrises the full solution space $\solnspace$ albeit non-uniquely.

\subsection{Gauge-normalisation of initial data}\label{OVGaugenormalisationofinitialdata}

We now identify the class of solutions to the equations of linearised gravity that will be subject to our decay statement.

This section of the overview corresponds to section \ref{Residualpuregaugesolutionstotheequationsoflinearisedgravity} in the main body of the paper.\newline

It is clear that solutions $\gidnlin$ to the equations of linearised gravity arising under our well-posedness theorem from the subset of smooth asymptotically flat seed data consisting only of Cauchy data for the Regge--Wheeler and Zerilli equations will satisfy
\begin{align*}
\gidnlin=\gamma(\Psilin, \Philin).
\end{align*}
A decay statement for this family of solutions thus follows immediately from a decay statement for solutions to the Regge--Wheeler and Zerilli equations in view of the properties of the map $\gamma$. We shall call this class of solutions \emph{initial-data normalised} since whether an ``admissible'' solution to the equations of linearised gravity is initial-data normalised is manifestly a condition on the seed data from which it arises.

Moreover Proposition \ref{OVpropdecoupling} shows that given any solution to the equations of linearised gravity lying in the space $\solnmap(\seedspace^{af})$ then one can normalise it by residual gauge and linearised Kerr solutions so as to make it initial-data normalised. In fact the solutions one has to add are now unique and can be explicitly identified from the seed from which the original solution arises -- see section \ref{Achievingtheinitialdatanormalisationforageneralsolution} for verification. Establishing a decay statement for initial-data normalised thus suffices to establish a decay statement for all ``admissible'' solutions to the equations of linearised gravity.

\subsection{Aside: The scalar wave equation on the Schwarzschild exterior spacetime}\label{OVAside:thescalarwaveequationontheSchwarzschildexteriorspacetime}

In this section of the overview we make an aside to discuss the scalar wave equation $\Box_{g_M}\psi=0$ on $\Mg$ and the methods by which one establishes a decay statement for solutions thereof. Insights gained for this simpler problem will prove fundamental in establishing a decay statement for solutions to the Regge--Wheeler and Zerilli equations on $\Mg$ and hence for the equations of linearised gravity by virtue of the initial-data normalised solutions identified in the previous section.

\subsubsection{Boundedness and decay for solutions to the scalar wave equation on $\Mg$}\label{OVBoundednessanddecayforsolutionstothescalarwaveeqnonMg}

Let $\psi$ be a smooth solution to the scalar wave equation on Schwarzschild:
\begin{align}\label{OVwaveeqn}
\Box_{g_M}\psi=0.
\end{align}
Then the boundedness and decay statement for such solutions is most naturally formulated in terms of certain $r$-weighted energy norms on hypersurfaces which penetrate both the horizon and future null infinity. Indeed, we define the function $\taus$ on $\mcalm$ according to (recalling $r=2Mx$)
\begin{align*}
\tau^\star(t^*, r, \mathbbm{p}):=\begin{cases} 
t^* & r\leq R \\
u(t^*, r, \mathbbm{p}) +R+4M\log(R-2M)& r\geq R,
\end{cases}
\end{align*}
for $\mathbbm{p}\in\twosphere$ and where $u$ is the optical function of section \ref{OVTheSchwarzschildexteriorfamily}. Consequently, denoting by $\Xi_{\taus}$ the level sets of the function $\taus$, we associate to $\psi$ the following flux norms (for $R>>10M$ and with definitions to follows):
\begin{align}
\mathbb{F}[\psi]:=&\sup_{\taus\geq\taus_0}\int_{\Xi_{\taus}\cap\{p\in\mcalm | r(p)\leq R\}}\Big(|\pt  \psi|^2+|\pr  \psi|^2+|\sn  \psi|^2\Big)\exd r\wedge\mathring{\epsilon}\nonumber\\
+&\sup_{\taus\geq\taus_0}\int_{\Xi_{\taus}\cap\{p\in\mcalm | r(p)\geq R\}}\Big(r^2|D (r\psi)|^2+|\sn  (r\psi)|^2\Big)\exd r\wedge\mathring{\epsilon},\label{OVflux}\\
\mathbb{D}[\psi]:&=\int_{\Sigma_0}r^{2}\Big(|\pt  (r\psi)|^2+|\pr  (r\psi)|^2+|\sn  (r\psi)|^2\Big)\exd r\wedge\mathring{\epsilon}\label{OVinitialflux}
\end{align}
along with the integrated decay norms (for $1>\beta_0>>0$)\footnote{We interpret $\exd\taus$ in the sense of measure.}:
\begin{align}
\mathbb{I}[\psi]:=&\int_{\taus_0}^\infty\int_{\Xi_{\taus}\cap\{p\in\mcalm | r(p)\geq R\}}\Big(r|D (r\psi)|^2+r^{\beta_0}|\sn  (r\psi)|^2\Big)\exd\taus\wedge\exd r\wedge\mathring{\epsilon},\label{OViled}\\
\mathbb{M}[\psi]:=&\int_{\taus_0}^\infty\int_{\Xi_{\taus}}r^{-3}\Big(|\pt  (r\psi)|^2+|\pr  (r\psi)|^2+|r\sn  (r\psi)|^2+|  (r\psi)|^2\Big)\exd\taus\wedge\exd r\wedge\mathring{\epsilon}\label{OVied}.
\end{align}
Here, $\sn$ is the standard ``spherical gradient'' whereas $\depsilon$ is the standard ``unit spherical volume form''\footnote{In particular, in the coordinates $(t^*, r, \theta, \varphi)$ we have for $f\in\smfun$ $|r\sn f|^2=|\partial_\theta f|^2+\tfrac{1}{\sin^2\theta}|\partial_\varphi f|^2$  and $\depsilon=\sin^2\theta\exd\theta\wedge\exd\varphi$. Note we give more geometric interpretations to these objects in the bulk of the paper.}. Moreover $D$ is the derivative operator
\begin{align*}
D:&=\frac{1+\frac{2M}{r}}{1-\frac{2M}{r}}\pt+\pr,\\
&=-(\exd u)^\sharp
\end{align*}
and we recall the definition of $\Sigma_0$ from section \ref{OVTheSchwarzschildexteriorfamily}. Thus, the flux norms \eqref{OVflux} and \eqref{OVinitialflux} denote energy norms containing all tangential and normal derivatives to the hypersurfaces $\Xi_{\taus}$ and $\Sigma_0$, the former of which foliate $\mcalm$ -- see the Penrose diagram of Figure 2.

\begin{figure}[!h]
	\centering
	\begin{tikzpicture}
	\node (I)    at ( 4,0)   {};
	\path 
	(I) +(90:4)  coordinate (Itop)
	+(-90:4) coordinate (Ibot)
	+(180:4) coordinate (Ileft)
	+(0:4)   coordinate (Iright)
	;
	\draw  (Ileft) -- (Itop)node[midway, above, sloped] {$\eh$} -- (Iright) node[midway, above right]    {$\cal{I}^+$} -- (Ibot)node[midway, below right]    {} -- (Ileft)node[midway, below, sloped] {} -- cycle;
	
	\draw[dashed]  
	(Ibot) to[out=50, in=-50, looseness=0.75] node[pos=0.475, above, sloped] {$r=R$        }($(Itop)!.5!(Itop)$) ;
	
	\draw 
	($(Itop)!.8!(Ileft)$) to[out=0, in=0, looseness=0.05] ($(Iright)!.5!(Iright)$);
	
	\draw 
	($(Itop)!.3!(Ileft)$) to[out=0, in=0, looseness=0.05] ($(Iright)!.5!(Iright)$);
	
	\draw 
	($(Itop)!.6!(Iright)$) --node[midway, below, sloped]{$\Xi _{\tau_1 ^*}$} (5.12, 0.32);
	
	\draw 
	($(Itop)!.42!(Iright)$) --node[right, below, sloped]{$\Xi _{\tau_2 ^*} $} (4.975, 1.66);
	
	\node [below] at (3,0.5) {$\Sigma_1$};
	
	\node [below] at (3,2.5) {$\Sigma_2$};
	
	\end{tikzpicture}
	\caption{A Penrose diagram of $\Mg$ depicting the hypersurfaces $\Xi_{\taus}$ which penetrate both $\eh$ and $\nullinf$. Here, the hypersurfaces $\Sigma_{t^*}$ are level sets of the time function $t^*$.} 
	\label{Figure4}
\end{figure}
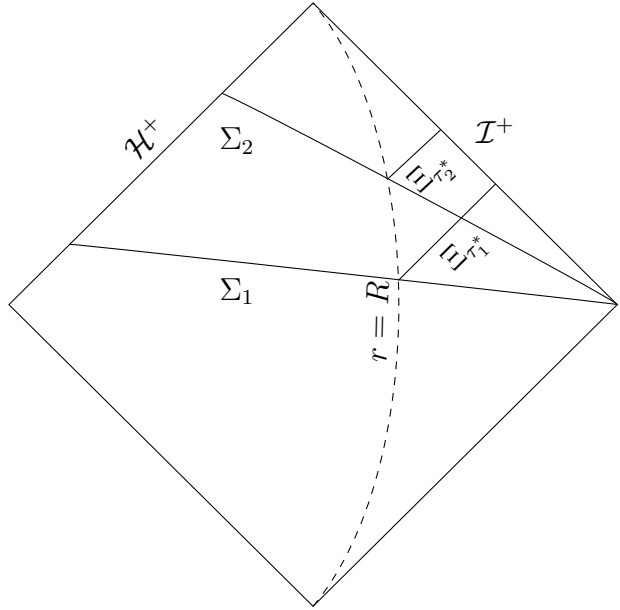

We then have the following definite statement due to Dafermos--Rodnianski.

\begin{theorem*}[Dafermos--Rodnianski -- \cite{DRlecturenotes}--\cite{DRredshift}]
	Let $\psi$ be a smooth solution to \eqref{OVwaveeqn}. Then for any $n\geq0$ the following estimates hold, provided that the fluxes on the right hand side are finite.
	\begin{enumerate}[i)]
		\item the higher order flux and weighted bulk estimates
		\begin{align}\label{OVfluxestimate}
		\mathbb{F}^n[\psi]+\mathbb{I}^n[\psi]\lesssim\mathbb{D}^n[\psi].
		\end{align}
		\item the higher order integrated decay estimate
		\begin{align}\label{OViedestimate}
		\mathbb{M}^{n}[\psi]\lesssim\mathbb{D}^{n+1}[\psi].
		\end{align}
		\item the higher order pointwise decay bounds for $i+j+k+l\leq n$
		\begin{align*}
		|\pt^i\pr^j(r\sn)^k\big((r-2M)D\big)^l\big(r\psi\big)|_{\taus,r}\,\lesssim\,\frac{1}{\sqrt{\taus}}\,\mathbb{D}^{n+4}[\psi].
		\end{align*}
	\end{enumerate}
\end{theorem*}
Here, the above are natural higher order norms defined by replacing $\psi$ in \eqref{OVflux}-\eqref{OVied} with the appropriate derivatives -- see the bulk of the paper for the precise definition. Moreover, the pointwise norm is defined according to
\begin{align*}
|\phi|_{\taus,r}:=\sup_{\twosphere_{\taus,r}}|\phi|,\qquad \phi\in\smfun
\end{align*}
for $\twosphere_{\taus,r}\subset\mcalm$ the 2-sphere given by the intersection of the level sets of the functions $\taus$ and $r$.

The proof by Dafermos--Rodnianski of the $n=0$ case of the above relies on the following two key estimates for solutions to \eqref{OVwaveeqn} and for any $\taus_2\geq\taus_1$:
\begin{align}
&\int_{\Xi_{\taus_2}\cap\{p\in\mcalm | r(p)\leq R\}}\Big(|\pt  \psi|^2+|\pr  \psi|^2+|\sn  \psi|^2\Big)r^2\exd r\wedge{\depsilon}\\+
&\int_{\Xi_{\taus_2}\cap\{p\in\mcalm | r(p)\geq R\}}\Big(|D \psi|^2+|\sn  \psi|^2\Big)r^2\exd r\wedge{\depsilon},\nonumber\\
\lesssim\,&\int_{\Xi_{\taus_1}\cap\{p\in\mcalm | r(p)\leq R\}}\Big(|\pt  \psi|^2+|\pr  \psi|^2+|\sn  \psi|^2\Big)\exd r\wedge{\depsilon}\\
+&\int_{\Xi_{\taus_1}\cap\{p\in\mcalm | r(p)\geq R\}}\Big(|D \psi|^2+|\sn  \psi|^2\Big)r^2\exd r\wedge{\depsilon}\label{OVenergyestimate}
\end{align}
and
\begin{align}
\int_{\taus_1}^\infty&\int_{\Xi_{\taus}}r^{-3}\Big(|\pt  (r\psi)|^2+|\pr  (r\psi)|^2+|r\sn  (r\psi)|^2+|  (r\psi)|^2\Big)\exd\taus\wedge\exd r\wedge{\depsilon}\nonumber\\
\lesssim\,\sum_{i=0}^1\bigg(&\int_{\Xi_{\taus_1}\cap\{p\in\mcalm | r(p)\leq R\}}\Big(|\pt  \pt^i\psi|^2+|\pr  \pt^i\psi|^2+|\sn  \pt^i\psi|^2\Big)\exd r\wedge{\depsilon}\nonumber\\
+&\int_{\Xi_{\taus_1}\cap\{p\in\mcalm | r(p)\geq R\}}\Big(|D \pt^i\psi|^2+|\sn  \pt^i\psi|^2\Big)r^2\exd r\wedge{\depsilon}\bigg)\label{OViledestimate}.
\end{align}
Indeed, in \cite{DRrpmethod} Dafermos--Rodnianski developed a very robust method which takes as input the estimates \eqref{OVenergyestimate}-\eqref{OViledestimate} and returns, via a hierarchy of $r$-weighted estimates on the $r$-weighted quantity $r\psi$, the estimates $i)-iii)$ of the ($n=0$ case of the) theorem statement, where for pointwise bounds one to a Sobolev embedding. Consequently, we shall see in the following two sections that establishing the estimates \eqref{OVenergyestimate} and \eqref{OViledestimate} requires an intimate understanding of the geometry of $\Mg$, in particular how the celebrated red-shift effect and the presence of trapped null geodesics (see \cite{Waldbook}) effects the propagation of waves.

The higher order estimates will then be discussed in section \ref{OVHigherorderestimates}.

\subsubsection{The degenerate energy and red-shift estimates}\label{OVThedegenerateenergyandredshiftestimates}

To investigate how one proves such estimates it is expedient to introduce the stress-energy tensor
\begin{align*}
\mathbb{T}[\psi]:=2\exd\psi{\otimes}\exd\psi-g_M\,g_M^{-1}(\exd\psi, \exd\psi)
\end{align*}
where $\exd$ is the exterior derivative on $\mcalm$. Then one has the following positivity properties at any $p\in\mcalm$ and for vector fields $X,Y$ on $\mcalm$:
\begin{enumerate}[1)]
	\item if $g_M(X,Y)\big|_p<0$ and $X,Y$ are future-directed\footnote{Recall the time-orientation is provided by $\pt$.} then $\mathbb{T}[\psi](X,Y)\big|_p$ bounds all derivatives of $\psi$ at $p$
	\item if $g_M(X,Y)\big|_p\leq0$ $X,Y$ are future-directed then $0\leq\mathbb{T}[\psi](X,Y)\big|_p$
\end{enumerate}
Moreover, if $\psi$ in addition satisfies \eqref{OVwaveeqn} then
\begin{align*}
\nabla^b\big(\mathbb{T}[\psi]\big)_{ba}=0.
\end{align*}
Defining the 1-form $\mathbb{J}^X[\psi]:=\mathbb{T}[\psi](X,\cdot)$ where $X$ is a causal vector field on $\mcalm$ then Stokes Theorem (on a manifold with corners) therefore yields, for $\psi$ a solution to $\eqref{OVwaveeqn}$, the inequality
\begin{align}
&\int_{\Xi_{\taus_2}}\mathbb{J}^X[\psi](n_{\Xi_{\taus_2}})\,\exd \text{Vol}(\Xi_{\tau_2 ^*})+\int_{\taus_1}^{\taus_2}\int_{\Xi_{\taus}}\frac{1}{2}\mcalL_Xg_M\cdot\mathbb{T}[\psi]\exd \text{Vol}(\Xi_{\taus})\nonumber\\
\lesssim\,&\int_{\Xi_{\taus_1}}\mathbb{J}^X[\psi](n_{\Xi_{\taus_1}})\exd \text{Vol}(\Xi_{\taus_1})\label{OVStokes}
\end{align}
where $n_{\Xi_{\taus}}$ is a suitably interpreted normal to the hypersurface $\Xi_{\taus}$ and we have discarded the flux term on $\eh$ as this is positive-definite by property 2).

In particular, as the vector field $\pt$ is causal and Killing we have from \eqref{OVStokes} the degenerate energy estimate
\begin{align}
&\int_{\Xi_{\taus_2}\cap\{p\in\mcalm | r(p)\leq R\}}\Big(|\pt  \psi|^2+\big(1-\tfrac{2M}{r}\big)|\pr  \psi|^2+|\sn  \psi|^2\Big)\exd r\wedge{\depsilon}\nonumber\\
+&\int_{\Xi_{\taus_2}\cap\{p\in\mcalm | r(p)\geq R\}}\Big(|D \psi|^2+|\sn  \psi|^2\Big)r^2\exd r\wedge{\depsilon},\nonumber\\
\lesssim\,&\int_{\Xi_{\taus_1}\cap\{p\in\mcalm | r(p)\leq R\}}\Big(|\pt  \psi|^2+\big(1-\tfrac{2M}{r}\big)|\pr  \psi|^2+|\sn  \psi|^2\Big)\exd r\wedge{\depsilon}\nonumber\\
+&\int_{\Xi_{\taus_1}\cap\{p\in\mcalm | r(p)\geq R\}}\Big(|D \psi|^2+|\sn \psi|^2\Big)r^2\exd r\wedge{\depsilon}\label{OVdegenergyestimate}
\end{align}
where the degeneration in the transversal derivative $\pr$ to $\eh$ occurs due to the fact that $\pt$ is null there (cf. property 2)). The first estimate \eqref{OVenergyestimate} thus follows from the estimate \eqref{OVdegenergyestimate} if the weights at $\eh$ can be improved and this improvement was achieved by Dafermos--Rodnianski in \cite{DRlecturenotes} where they established the existence of a time-like vector field $N$ which satisfies the following so-called red-shift estimate in a neighbourhood of $\eh$:
\begin{align}
\mathbb{J}^N[\psi](n_{\Xi_{\taus}})\,\lesssim\,\mcalL_Ng_M\cdot\mathbb{T}[\psi]\label{OVredshiftestimate}.
\end{align}
Note that the existence of such a vector field $N$ is intimately related to the celebrated red-shift effect on $\Mg$ -- see \cite{DRlecturenotes}. Moreover, noting that the left hand side of \eqref{OVredshiftestimate} controls all derivatives of $\psi$ by property 2) (and the fact that $\Xi_{\taus}$ and the spacelike hypersurface $\Sigma_{\taus}$ given by the level sets of $t^*$ coincide near $\eh$), the estimate \eqref{OVredshiftestimate}, when combined with the degenerate estimate \eqref{OVdegenergyestimate} and the integral inequality \eqref{OVStokes}, ultimately suffices to establish the desired estimate \eqref{OVenergyestimate}.

\subsubsection{Integrated local energy decay and the role of trapping}\label{OVIntegratedlocalenergydecayandtheroleoftrapping}

In order to establish the estimate \eqref{OViledestimate} it is convenient to exploit once more the formalism of the previous section. Indeed, revisiting the integral inequality \eqref{OVStokes} one has the aim of choosing a vector field $X$ so as to generate a bulk term which controls all derivatives of $\psi$. 

Now it turns out (see \cite{DRlecturenotes}) that one can use estimate \eqref{OVStokes} with $X_g$ a space-like\footnote{The additional flux term this generates at $\eh$ in the application of Stokes theorem is in fact controllable by the flux term of $\eh$ associated to the $X=\pt$ estimate of section \ref{OVThedegenerateenergyandredshiftestimates} that was discarded in the estimate \eqref{OVStokes}.} vector field of the form
\begin{align}\label{OVmorawetzvectorfield}
X_g=\big(1-\tfrac{3M}{r}\big)\big(1+g(r)\big)\Big(\tfrac{2M}{r}\pt+\big(1-\tfrac{2M}{r}\big)\pr\Big),\qquad g(r)=o(r^{-1}) 
\end{align}
in conjunction with both the estimate \eqref{OVenergyestimate} and the red-shift estimate \eqref{OVredshiftestimate}, to establish the bound
\begin{align}
&\int_{\taus_1}^\infty\int_{\Xi_{\taus}}r^{-3}\bigg(\big(1-\tfrac{3M}{r}\big)^2\Big(|\pt  (r\psi)|^2+|\pr  (r\psi)|^2+|\sn  (r\psi)|^2\Big)+|  (r\psi)|^2\bigg)\exd\taus\wedge\exd r\wedge{\depsilon}\nonumber\\
\lesssim\,&\int_{\Xi_{\taus_1}\cap\{p\in\mcalm | r(p)\leq R\}}\Big(|\pt  \psi|^2+|\pr  \psi|^2+|\sn  \psi|^2\Big)\exd r\wedge{\depsilon}\nonumber\\
+&\int_{\Xi_{\taus_1}\cap\{p\in\mcalm | r(p)\geq R\}}\Big(|D \psi|^2+|\sn  \psi|^2\Big)r^2\exd r\wedge{\depsilon}\label{OVdegiledestimate}.
\end{align}
Note that the degeneration at $r=3M$ is a necessary consequence of the existence of trapped\footnote{Null geodesics which remain tangent to the hypersurface $r=3M$ for all $t^*$.} null geodesics at $r=3M$ on $\Mg$ and a general result due to Sbierski \cite{Sbierskigaussbeam}. However, to provide a bulk estimate that does not degenerate at $r=3M$ it in fact suffices to obtain the estimate
\begin{align*}
\int_{\taus_1}^\infty&\int_{\Xi_{\taus}}r^{-3}|(r\pt\psi)|^2
\exd\taus\wedge\exd r\wedge{\depsilon}\nonumber\\
\lesssim\,&\int_{\Xi_{\taus_1}\cap\{p\in\mcalm | r(p)\leq R\}}\Big(|\pt ^2\psi|^2+|\pr  \pt\psi|^2+|\sn  \pt\psi|^2\Big)\exd r\wedge{\depsilon}\nonumber\\
+&\int_{\Xi_{\taus_1}\cap\{p\in\mcalm | r(p)\geq R\}}\Big(|D \pt\psi|^2+|\sn  \pt\psi|^2\Big)r^2\exd r\wedge{\depsilon}
\end{align*}
which follows easily from \eqref{OVdegiledestimate} and the fact that $\pt$ is Killing and therefore commutes with the wave operator $\Box_{g_M}$. This consequently yields the estimate \eqref{OViledestimate} and moreover explains the derivative loss in the statement of the Theorem\footnote{Although this derivative loss can be improved \cite{MMTTstrichschwarz} some degree of loss is required due to the result of Sbierski.}.

\subsubsection{Higher order estimates}\label{OVHigherorderestimates}

With the key ingredients for the proving the $n=0$ case of the Theorem understood we turn now to the higher order cases.

Indeed, we first observe that since $\pt$ and $\Omega_i$ for $i=1,2,3$ are Killing fields of $\Mg$, where $\Omega_i$ denote a basis of $SO(3)$, then $\pt$ and each of the $\Omega_i$ commute trivially with the wave operator $\Box_{g_M}$ and thus the $n=0$ case of the Theorem holds replacing $\psi$ by\footnote{Here we use that $\sum_{i=0}^3|\Omega_i(\psi)|^2\,\lesssim\, |r\sn\psi|^2\,\lesssim\sum_{i=0}^3|\Omega_i(\psi)|^2$.} $\pt^i(r\sn)^j\psi$ for any positive integers $i,j$. In addition, by writing the wave equation for $\psi$ as an ODE in $r$ with inhomogeneities given by derivatives of $\psi$ containing at least one $t^*$ or angular derivative, then the previously derived bounds on $\pt^i(r\sn)^j\psi$ in fact allows one to replace $\psi$ by $\pt^i\Big(\big(1-\tfrac{2M}{r}\big)\pr\Big)^j(r\sn)^k\Big(\big(1-\tfrac{2M}{r}\big)D\Big)^l\psi$ in the $n=0$ case of the Theorem statement, for any positive integers $i,j,k,l$. It thus remains to remove the degeneration at $\eh$ for the derivative $\pr$ and the degeneration towards $\nullinf$ for the derivative operator $rD$ (cf. the pointwise bounds in part $iii)$ of the Theorem statement).

Consequently, to remove the degeneration at $\eh$ one proceeds by first commuting the wave equation \eqref{OVwaveeqn} with the (time-like) vector field $-\pr^m$ for any positive integer $m$, thus generating additional lower order terms
as the vector field $-\pr$ is not Killing. However, these lower order terms are such that they are either controlled by the bounds derived on the quantities $\pt^i\Big(\big(1-\tfrac{2M}{r}\big)\pr\Big)^j(r\sn)^k\Big(\big(1-\tfrac{2M}{r}\big)D\Big)^l\psi$ in the previous step for sufficiently many $i,j,k,l$ or they come with the correct sign. In particular, for any positive integer $j$, one has the higher order red-shift estimate\footnote{This presentation of the higher order red-shift estimate is borrowed from the overview of \cite{DHRlinstabschwarz}.}
\begin{align*}
\mathbb{J}^N[\pr^j\psi](n_{\Xi_{\taus}})\,\lesssim\,\mcalL_Ng_M\cdot\mathbb{T}[\pr^j\psi]-\{\text{controllable terms}\}.
\end{align*}
Proceeding as in section \ref{OVThedegenerateenergyandredshiftestimates} one thus removes the degeneration at $\eh$ for the derivative $\pr$ -- see \cite{DRlecturenotes} for further details.

Similarly, to remove the degeneration towards $\nullinf$ one proceeds by now considering the wave equation satisfied by the $r$-weighted commuted quantity $\big((r-2M)D\big)^m(r\psi)$ for any positive integer $m$. The error terms this generates are lower order in the sense that they are either controllable by the estimates derived on the quantities $\pt^i\pr^j(r\sn)^k\Big(\big(1-\tfrac{2M}{r}\big)D\Big)^l\psi$ in the previous two steps for sufficiently many $i,j,k,l$ or they come with favourable weights in the sense that the hierarchy of $r$-weighted estimates established by Dafermos and Rodnianski for the scalar wave $r\psi$ hold with equal validity for the commuted quantity $\big((r-2M)D\big)^m(r\psi)$ -- see \cite{DRlecturenotes} (and also \cite{Moschidisrpmethod}) for further details.

\subsection{Statement of main theorems and outline of proofs}\label{OVThemaintheorems}

In this final part of the overview we provide rough statements of the main theorems of this paper and then give an outline of proofs.

We begin in section \ref{OVTheorem1} with a rough version of our first theorem which concerns a boundedness and decay statement for solutions to the Regge--Wheeler and Zerilli equations on $\Mg$ -- this will have applications to the equations of linearised gravity in view of the gauge-normalised solutions of section \ref{OVGaugenormalisationofinitialdata}. Then in section \ref{OVOutlineoftheproofofTheorem1} we give an outline of the proof, noting already that all insights from section \ref{OVAside:thescalarwaveequationontheSchwarzschildexteriorspacetime} enter. Then in section \ref{OVTheorem2} we provide a rough version of our second theorem which concerns a boundedness and decay statement for the previously mentioned gauge-normalised solutions to the equations of linearised gravity. Finally in section \ref{OVOutlineoftheproofofTheorem2} an outline of the proof of Theorem 2 is given.

This section of the overview corresponds to sections \ref{Precisestatementsofthemaintheorems}-\ref{Proofoftheorem2} in the main body of the paper.

\subsubsection{Theorem 1: Boundedness and decay for solutions to the Regge--Wheeler and Zerilli equations}\label{OVTheorem1}

A rough formulation of Theorem 1 is as follows -- see the main body of the paper for the precise version. Note in what follows we remove the superscript $(1)$ from the quantities under consideration as the stated results hold independently of the equations of linearised gravity.

\begin{customthm}{1}\label{OVmainthmdecayRWandZ}
	Let $\Phi$ be a smooth solution to the Regge--Wheeler equation \eqref{OVeqnRW} on $\Mg$ supported on the $l\geq 2$ spherical harmonics.
	Then for any integer $n\geq 0$ the following estimates hold provided that the fluxes on the right hand side are finite:
	\begin{enumerate}[i)]
		\item the higher order flux and weighted bulk estimates
		\begin{align*}
		\norm{F}{r^{-1}\Phi}[n]+\norm{I}{r^{-1}\Phi}[n]&\lesssim\norm{D}{r^{-1}\Phi}[n]
		\end{align*}
		\item the higher order integrated decay estimate
		\begin{align*}
		\norm{M}{r^{-1}\Phi}[n]&\lesssim\norm{D}{r^{-1}\Phi}[n+1]
		\end{align*}	
		\item the higher order pointwise decay bounds for $i+j+k+l\leq n$
		\begin{align*}
		|\pt^i\pr^j(r\sn)^k\big((r-2M)D\big)^l\Phi|\,\lesssim\,\frac{1}{\sqrt{\taus}}\,\mathbb{D}^{n+4}[r^{-1}\Phi].
		\end{align*}
	\end{enumerate}
	
	Let now $\Psi$ be a smooth solution to the Zerilli equation \eqref{OVeqnZ} on $\Mg$ supported on the $l\geq 2$ spherical harmonics. Then for any integer $n\geq 0$ the following estimates hold provided that the fluxes on the right hand side are finite:
	\begin{enumerate}[i)]
		\item the higher order flux and weighted bulk estimates
		\begin{align*}
		\norm{F}{r^{-1}\Psi}[n]+\norm{I}{r^{-1}\Psi}[n]&\lesssim\norm{D}{r^{-1}\Psi}[n]
		\end{align*}
		\item the higher order integrated decay estimate
		\begin{align*}
		\norm{M}{r^{-1}\Psi}[n]&\lesssim\norm{D}{r^{-1}\Psi}[n+1]
		\end{align*}	
		\item the higher order pointwise decay bounds for $i+j+k+l\leq n$
		\begin{align*}
		|\pt^i\pr^j(r\sn)^k\big((r-2M)D\big)^l\Psi|\,\lesssim\,\frac{1}{\sqrt{\taus}}\,\mathbb{D}^{n+4}[r^{-1}\Psi].
		\end{align*}
	\end{enumerate}
\end{customthm}
Note that the norms in the theorem statement concern $r^{-1}\Phi$ and $r^{-1}\Psi$ as it is these quantities which satisfy the wave equation up to first and second order derivatives (cf. equations \eqref{OVeqnRW}-\eqref{OVeqnZ} and the theorem of section \ref{OVBoundednessanddecayforsolutionstothescalarwaveeqnonMg}).

We make the following remarks regarding Theorem 1.

The flux estimate associated to the norm $\mathbb{F}$ in both parts $i)$ of Theorem 1 should be considered as a boundedness statement that does not lose derivatives. Conversely, the derivative loss in parts $ii)$ is unavoidable and and is a consequence of the trapping effect on $\Mg$ (cf. section \ref{OVIntegratedlocalenergydecayandtheroleoftrapping}).

In addition, observe from both parts $iii)$ of Theorem 1 that commuting with $\sn$ and $\big(1-\tfrac{2M}{r}\big)D$ improves the pointwise bounds in $r$ -- this will prove important in the proof of Theorem 2 to be discussed in section \ref{OVOutlineoftheproofofTheorem2}. Note that the former is a consequence of how the angular operator $\sn$ is defined whereas the latter is a consequence of the improved $r$-weights on the operator $D$ in the norms of section \ref{OVBoundednessanddecayforsolutionstothescalarwaveeqnonMg}.

Finally, we note a version of Theorem 1 regarding solutions to the Regge--Wheeler equation was originally given by Holzegel in \cite{Holzegelultschwarz} (see also \cite{DHRlinstabschwarz}). Conversely, a version of Theorem 1 regarding solutions to the Zerilli equation was originally given in the independent works of the author \cite{Johnsondecayz} and Hung--Keller--Wang \cite{HKWlinstabschwarz}.

\subsubsection{Outline of the proof of Theorem 1}\label{OVOutlineoftheproofofTheorem1}

We now discuss the proof of Theorem 1. We discuss only the Zerilli equation as this is the more complicated case.

We first recall the definition of the Zerilli equation \eqref{OVeqnZ}:
\begin{align}\label{OVZwave}
\Box_{g_M}(r^{-1}\Psi)=-\frac{8}{r^2}\frac{M}{r}r^{-1}\Psi+\frac{24}{r^3}\frac{M}{r}(r-3M)\zslapinv{1}r^{-1}\Psi+\frac{72}{r^3}\frac{M}{r}\frac{M}{r}(r-2M)\zslapinv{2}r^{-1}\Psi.
\end{align}
Thus, the Zerilli equation differ from the scalar wave equation \eqref{OVwaveeqn} by an $r$-weight and the presence of the lower order terms on the right hand side of \eqref{OVZwave}. Consequently, all insights gained for the scalar wave equation in section \ref{OVAside:thescalarwaveequationontheSchwarzschildexteriorspacetime} of the overview enter and it remains to understand the complications, if any, provided by these additional lower order terms. These complications can in fact be understood at the level of the $\pt$-flux estimate of section \ref{OVThedegenerateenergyandredshiftestimates} in the overview and the integrated local energy decay estimates of section \ref{OVIntegratedlocalenergydecayandtheroleoftrapping}. Indeed, with the associated issues resolved, the proof of Theorem 1 proceeds analogously to that detailed in section \ref{OVAside:thescalarwaveequationontheSchwarzschildexteriorspacetime} and shall not be discussed further in this overview. In particular, we emphasize that the techniques developed by Dafermos--Rodnianski in \cite{DRrpmethod} to derive the hierarchy of $r$-weighted estimates mentioned in section \ref{OVBoundednessanddecayforsolutionstothescalarwaveeqnonMg} are indeed robust enough to allow for the lower order terms appearing in the Regge--Wheeler and Zerilli equations respectively.

We begin with the $\pt$-flux estimate of section \ref{OVThedegenerateenergyandredshiftestimates}. Introducing the stress-energy tensors associated to $r^{-1}\Psi$ as
\begin{align*}
\mathbb{T}[r^{-1}\Psi]=2\exd(r^{-1}\Psi)\otimes\exd(r^{-1}\Psi)-g_M\,g_M^{-1}(\exd(r^{-1}\Psi), \exd(r^{-1}\Psi))
\end{align*}
then it follows from  \eqref{OVZwave} that
\begin{align}\label{OVZdiv}
\nabla^b\big(\mathbb{T}[r^{-1}\Psi]\big)_{ab}=\bigg(-\frac{8}{r^2}\frac{M}{r}r^{-1}\Psi+\frac{24}{r^3}\frac{M}{r}(r-3M)&\zslapinv{1}r^{-1}\Psi\nonumber\\
+\frac{72}{r^3}\frac{M}{r}\frac{M}{r}(r-2M)&\zslapinv{2}r^{-1}\Psi\bigg)\,\exd_a(r^{-1}\Psi).
\end{align}
Applying Stokes theorem as in section \ref{OVThedegenerateenergyandredshiftestimates} of the overview with $X=\pt$ (noting that the positivity properties 1) and 2) hold for $\mathbb{T}[r^{-1}\Psi]$ so that the flux term along $\eh$ can be ignored) we therefore find
\begin{align}
&\int_{\Xi_{\taus_2}\cap\{p\in\mcalm | r(p)\leq R\}}\Big(|\pt  \Psi|^2+\big(1-\tfrac{2M}{r}\big)|\pr  \Psi|^2+|\sn_{\mfZ}  \Psi|^2\Big)\exd r\wedge{\depsilon}\nonumber\\
+&\int_{\Xi_{\taus_2}\cap\{p\in\mcalm | r(p)\geq R\}}\Big(|D \Psi|^2+|\sn_{\mfZ}  \Psi|^2\Big)\exd r\wedge{\depsilon},\nonumber\\
\lesssim\,&\int_{\Xi_{\taus_1}\cap\{p\in\mcalm | r(p)\leq R\}}\Big(|\pt  \Psi|^2+\big(1-\tfrac{2M}{r}\big)|\pr  \Psi|^2+|\sn \Psi|^2\Big)\exd r\wedge{\depsilon}\nonumber\\
+&\int_{\Xi_{\taus_1}\cap\{p\in\mcalm | r(p)\geq R\}}\Big(|D \Psi|^2+|\sn \Psi|^2\Big)\exd r\wedge{\depsilon}\label{OVZdegenergyestimate}.
\end{align}
Here, we have defined
\begin{align*}
|\sn_{\mfZ}\Psi|^2:=|\sn\Psi|^2-\frac{1}{r^2}|\Psi|^2-\frac{6}{r^2}\frac{M}{r}|\Psi|^2-&\frac{24}{r^3}\frac{M}{r}(r-3M)|r\sn\zslapinv{1}\Psi|^2\\
+&\frac{24}{r^4}\frac{M}{r}\Big(2(r-3M)^2+3M(r-2M)\Big)|\zslapinv{1}\Psi|^2.
\end{align*}
In addition, we have integrated by parts to write the additional bulk terms generated from \eqref{OVZdiv} as a flux term. In particular, we note the integration by parts formulae associated to the operator $\zslapinv{p}$ derived in section \ref{Commutationformulaeandusefulidentities}.
Positivity of the left hand side of the $\pt$-flux \\eqref{OVZdegenergyestimate}, which we recall was immediate for the case of the scalar wave equation (cf. estimate \eqref{OVdegenergyestimate}), thus rests upon whether one can ensure positivity of the terms $|\sn_{\mfZ}\Psi|^2$.

To see that this is indeed the case we invoke the fact that $\Psi$ is supported on the $l\geq 2$ spherical harmonics. One thus has on any 2-sphere $\twosphere_{\taus,r}\subset\mcalm$ the Poincar\`e inequality
\begin{align}\label{OVZpoin}
\frac{6}{r^2}\inttwosphere{}{r}{|\Psi|^2}\,\lesssim\, \inttwosphere{}{r}{|\sn\Psi|^2}
\end{align}
with $\sepsilon=r^2\depsilon$. 
For the positivity of \eqref{OVZdegenergyestimate} it thus suffices to note the following refined estimate from section \ref{Thefamilyofoperatorsslap} which exploits the ellipticity of the operator $\zslapinv{p}$:
\begin{align}\label{OVzerefinedestim}
\frac{2}{r^2}\inttwosphere{}{r}{\Big(2(r-3M)^2|\zslapinv{1}\Psi|^2+r(r+9M)|r\sn\zslapinv{1}\Psi|^2\Big)}\,\lesssim\, \inttwosphere{}{r}{|\Psi|^2}.
\end{align}
This estimate combined with the estimate \eqref{OVZpoin} thus ultimately yields positivity of the left hand side of \eqref{OVZdegenergyestimate}. In particular, one has the bound
\begin{align*}
\inttwosphere{}{r}{|\sn\Psi|^2}\,\lesssim\, \inttwosphere{}{r}{|\sn_{\mfZ}\Psi|^2}.
\end{align*}
The degenerate energy estimate of section \ref{OVThedegenerateenergyandredshiftestimates} for the scalar wave $\psi$ thus holds for the $r$-weighted solutions to the Zerilli $r^{-1}\Psi$ respectively.

We turn now to the integrated local energy decay estimates, namely the estimate \eqref{OVdegiledestimate} with $r\psi$ replaced by $\Psi$. We recall from section \ref{OVIntegratedlocalenergydecayandtheroleoftrapping} that one is able to derive the analogous estimate for solutions to the wave equation by utilising a multiplier of the form given in \eqref{OVmorawetzvectorfield}. It is therefore natural to work with such a multiplier again. However, we now note the interesting fact that the equations satisfied by $\Psi$ are actually more symmetric in the first order derivatives than the wave equation satisfied by $r^{-1}\Psi$. Thus, the overarching appeal of symmetry suggests that it is fact more natural to analyse the Zerilli equations directly. Indeed, although in this case one loses the Lagrangian structure associated to the wave equation exploited in deriving the estimates of section \ref{OVAside:thescalarwaveequationontheSchwarzschildexteriorspacetime}, this can easily be replicated by the more pedestrian method of integrating by parts over the spacetime region under consideration the integrand given by the Zerilli equation multiplied by the ``multiplier'' $X_g(\Psi)$ for a vector field $X_g$ as in \eqref{OVmorawetzvectorfield}.

Proceeding in this manner first for the Regge--Wheeler equation it is in fact quite simple to show using both the associated $\pt$-energy estimate and the Poincar\'e inequality derived previously that the multiplier $X_g$ as in \eqref{OVmorawetzvectorfield} with $g(r)=1+\tfrac{3M}{r}$ yields the desired integrated local energy decay estimate. Moreover, the exact same multiplier yields the desired integrated local enery decay estimate for the Zerilli equation if one exploits in addition the integration by parts formulae and the refined estimate \eqref{OVzerefinedestim} associated to the operator $\zslapinv{p}$ to control the additional lower order terms.

This concludes our overview of the proof of Theorem \ref{OVmainthmdecayRWandZ}.

\subsubsection{Theorem \ref{OVthmboundedness}: Boundedness, decay and propagation of asymptotic flatness for gauge-normalised solutions to the equations of linearised gravity}\label{OVTheorem2}

We give now a rough formulation our of second theorem which concerns a boundedness and decay statement for the initial-data normalised solutions discussed in section \ref{OVGaugenormalisationofinitialdata}. The statement in question involves the norms introduced in section \ref{OVAside:thescalarwaveequationontheSchwarzschildexteriorspacetime} but now generalised to smooth, symmetric 2-covariant tensor fields on $\mcalm$ -- see section \ref{Flux,integrateddecayandpointwisenorms} for the precise definition. 

\begin{customthm}{2}\label{OVthmboundedness}
	Let $\glin$ be a smooth solution to the equations of linearised gravity arising from the smooth, asymptotically flat seed data and let $\Philin$ and $\Psilin$ be the associated invariant quantities. 
	
	We consider the initial-data normalised solution $\gidnlin$ associated to $\glin$ as in section \ref{OVGaugenormalisationofinitialdata}.
	
	Then for any $n\geq 0$ the following estimates hold, with the fluxes on the right hand side finite by the assumption of asymptotic flatness:
	\begin{enumerate}[i)]
		\item the higher order flux and weighted bulk estimates
		\begin{align*}
		\mathbb{F}^{n}[\gidnlin]+\mathbb{I}^{n}[\gidnlin]\lesssim\mathbb{D}^{n+2}[r^{-1}\Philin] + \mathbb{D}^{n+2}[r^{-1}\Psilin].
		\end{align*}
		\item the higher order integrated decay estimate
		\begin{align*}
		\mathbb{M}^{n}[\gidnlin]\lesssim\mathbb{D}^{n+2}[r^{-1}\Philin] + \mathbb{D}^{n+2}[r^{-1}\Psilin].
		\end{align*}
		\item the pointwise decay bound
		\begin{align*}
		|\gidnlin|\lesssim\frac{1}{\sqrt{\taus}r}\Big(\mathbb{D}^{6}[r^{-1}\Philin] + \mathbb{D}^{6}[r^{-1}\Psilin]\Big).
		\end{align*}
	\end{enumerate}
In particular, the solution $\glin$ decays inverse polynomially to the sum $\gkerr+\glin_V$.
\end{customthm}
Theorem 2 thus represents the appropriate analogue for the equations of linearised gravity of the decay statement for the linear wave equation discussed in section \ref{OVAside:thescalarwaveequationontheSchwarzschildexteriorspacetime} -- the definition of the above pointwise norm is to be found in the bulk of the paper. 

We make the following additional remarks regarding Theorem 2.

We emphasize as discussed in section \ref{OVGaugenormalisationofinitialdata} that the solutions $\glin_{\textnormal{Kerr}}$ and $\glin_V$ are determined explicitly from the seed of $\glin$.

The flux estimate associated to the norm $\mathbb{F}$ in part $i)$ of Theorem 2 should be considered as a boundedness statement that does not lose derivatives. Conversely, the derivative loss in parts $ii)$ is unavoidable and relates once more to the trapping effect on $\Mg$. Lastly, one can obtain higher pointwise bounds courtesy of Theorem 2 although we shall not state them explicitly in this paper.

In addition, we note that the finitness of the initial data norms in the theorem statement can be verfified explicitly from the asymptotic flatness criterion of the seed data (see the bulk of the paper for verifcation).

Finally, we note a version of Theorem 2 regarding solutions to the system of equations that result from linearising the Einstein vacuum equations, as they are expressed in a generalised wave gauge with respect to a different gauge-map $f$, around $\Mg$ was given by the author in \cite{Johnsonlinstabschwarzold}. Here however the solutions possessed weaker asymptotics -- in particular, the pointwise decay in $r$ of part $iii)$ was absent.

\subsubsection{Outline of the proof of Theorem 2}\label{OVOutlineoftheproofofTheorem2}

We finish the overview by discussing the proof of Theorem 2.

This in fact follows immediately from Theorem 1 combined with the properties of the map $\gamma$ from section \ref{OVDecouplingtheequationsoflinearisedgravity}.\newline

So ends our overview of this paper.

\section{The Schwarzschild exterior background }\label{TheSchwarzschildexteriorbackground}

In this section we introduce the Schwarzschild exterior family as well as various objects and operations defined on these spacetimes that shall prove useful throughout the remainder of the paper.

An outline of this section is as follows. We begin in section \ref{TheSchwarzschildexteriorspacetime} by defining the Schwarzschild exterior spacetime. Then in section \ref{Ageometricfoliationby2-spheres} we define a foliation of $\mcalm$ by 2-spheres. Then in section \ref{Theprojectionofsmoncovandsmsymtwocovontoandawayfromthel=0,1sphericalharmonics} we define the projection of smooth one forms and smooth symmetric, 2-covariant tensor fields onto and away from the $l=0,1$ spherical harmonics.  Then in section \ref{Ellipticoperatorson2-spheres} we derive various elliptic estimates on spheres. We then finish this section by presenting various commutation relations and indentities that will be useful throughout the text. 

Note we advise the reader to skip sections until the relevant machinery developed therein is required.

\subsection{The Schwarzschild exterior spacetime}\label{TheSchwarzschildexteriorspacetime}

We begin in section \ref{ThedifferentialstructureandmetricoftheSchwarzschildexteriorspacetime} by defining the differential structure and metric of the Schwarzschild exterior spacetime. Then in section \ref{KillingfieldsoftheSchwarzschildmetric} we consider the Killing fields. Then in section \ref{Tensoralgebra} we introduce several tensor spaces. Then in section \ref{Tensoranalysis} we develop a calculus on these spaces. Finally in section \ref{TheCauchyhypersurfaceSigma} we introduce a Cauchy hypersurface for $\mcalm$.

\subsubsection{The differential structure and metric of the Schwarzschild exterior spacetime}\label{ThedifferentialstructureandmetricoftheSchwarzschildexteriorspacetime}

Let $M>0$ be a fixed parameter.

We define the smooth manifold with boundary 
\begin{align*}
\mcalm:=\reals_{t^*}\times [1, \infty)_x\times \twosphere
\end{align*}
which we equip with the smooth Ricci-flat Lorentzian metric
\begin{align}\label{schwarzschildmetric}
g_M=-\bigg(1-\frac{1}{x}\bigg)\exd {t^*}^2+\frac{4M}{x}\exd t^*\exd x+4M^2\bigg(1+\frac{1}{x}\bigg)\exd x^2+4Mx^2\,\pi_{\twosphere}^*\roundmetric
\end{align}
where $\roundmetric$ is the metric on the unit round sphere and $\pi_{\twosphere}:\mcalm\rightarrow\twosphere$ is the projection map.  The Lorenztian manifold with boundary $\Mg$ thus defines the Schwarzschild exterior spacetime.

We define a time orientation on $\mcalm$ via the causal vector field $\pt$ and denote by $\eh$ the boundary, which we note is a null hypersurface.

\subsubsection{Killing fields of the Schwarzschild metric}\label{KillingfieldsoftheSchwarzschildmetric}

It is manifest from the form of \eqref{schwarzschildmetric} that the (causal) vector field
\begin{align*}
T=\pt
\end{align*}
is Killing. It is moreover clear that $g_M$ possesses the same symmetries as $\roundmetric$. In particular, the following vectors fields expressed in a spherical coordinate chart $(\theta, \varphi)$ of $\twosphere$ are Killing fields of $g_M$:
\begin{align*}
\partial_\varphi, \qquad -\sin\varphi\partial_\theta-\cot\theta\cos\varphi\partial_\varphi, \qquad\cos\varphi\partial_\theta-\cot\theta\sin\varphi\partial_\varphi.
\end{align*}

\subsubsection{Tensor algebra}\label{Tensoralgebra}

We introduce the spaces for $n\in\mathbb{N}$
\begin{align*}
\Gamma(T^nT\mcalm)&=\{\text{space of smooth n-contravariant tensor fields on $\mcalm$}\},\\
\Gamma(T^nT^*\mcalm)&=\{\text{space of smooth n-covariant tensor fields on $\mcalm$}\}
\end{align*}
along with
\begin{align*}
\Gamma(S^nT\mcalm)&=\{\text{space of smooth, symmetric n-contravariant tensor fields on $\mcalm$}\},\\
\Gamma(S^nT^*\mcalm)&=\{\text{space of smooth, symmetric n-covariant tensor fields on $\mcalm$}\}.
\end{align*}
We also set
\begin{align*}
\smfun=\{\text{space of smooth functions on $\mcalm$}\}.
\end{align*}
Finally, we introduce the notation $\astrosun$ for the symmetric tensor product.

\subsubsection{Tensor analysis}\label{Tensoranalysis}

We now introduce several operations on $\mcalm$ that act naturally on tensor fields.

\begin{itemize}
	\item for $n=1$ the sharp operator $\sharp:\smonecov\rightarrow\smonecon$ defined according to
	\begin{align*}
	(\tau^{\sharp})^a=\g^{ab}\tau_b
	\end{align*}
	\item for $n=2$ the trace operator $\tr:\smtwocov\rightarrow\Gamma(\mcalm)$ defined according to
	\begin{align*}
	\tr\tau=\big(\g^{-1}\big)^{ab}\tau_{ab}
	\end{align*}
	\item for $n\geq 1$ the divergence operator  $\text{div}:\Gamma(T^{n}T^*\mcalm)\rightarrow\Gamma(T^{n-1}T^*\mcalm)$ defined according to
	\begin{align*}
	\text{div}\tau_{a_1...a_{n-1}}=(\g^{-1})^{ab}(\nabla\tau)_{aba_1...a_{n-1}}
	\end{align*}
	with $\nabla$ the Levi-Civita connection of $\g$
	\item the wave operator $\Box_{g_M}:\Gamma(T^{n}T^*\mcalm)\rightarrow\Gamma(T^{n}T^*\mcalm)$ defined according to
	\begin{align*}
	\Box_{g_M}\tau_{a_1...a_n}=(\g^{-1})^{ab}(\nabla\nabla\tau)_{aba_1...a_n}
	\end{align*}
	\item for $n=1$ the Lie derivative $\nabla{\astrosun}:\smonecov\rightarrow\smsymtwocov$ defined according to
	\begin{align*}
	\nabla{\astrosun}\tau_{ab}=2(\nabla\tau)_{(ab)}
	\end{align*}
\end{itemize}

\subsubsection{The Cauchy hypersurface $\Sigma_0$}\label{TheCauchyhypersurfaceSigma}

We define the manifold with boundary
\begin{align*}
\Sigma_{0}:=\{0\}\times[1,\infty)\times\twosphere\subset\mcalm
\end{align*}
which we note determines a Cauchy hypersurface for $\Mg$. Defining by $i^*_0:\Sigma_{0}\rightarrow\mcalm$ the inclusion map we set
\begin{align*}
\overline{g}_M:=i^*_0\g
\end{align*}
which determines a Riemannian metric on $\Sigma_{0}$. 
\subsection{A geometric foliation by 2-spheres}\label{Ageometricfoliationby2-spheres}

We begin in section \ref{Thegeometric2-spherestwosphere} by defining the 2-spheres that shall foliate $\mcalm$. Then in section \ref{QandStensoralgebra} we decompose tensor fields on $\mcalm$ relative to this projection. Then in sections \ref{QandStensoranalysis}-\ref{MixedQandStensoranalysis} we define a natural calculus on these decomposed tensor fields.

\subsubsection{The geometric 2-spheres $\twosphere_{\taus,r}$}\label{Thegeometric2-spherestwosphere}

Let $\mathbbm{p}\in\twosphere$ and $R>>10M$. 

\begin{definition}
Let $\taus:\mcalm\rightarrow\reals$ and $r:\reals_{t^*}\times[1, \infty)_x\rightarrow\reals$ be the functions
\begin{align*}
\tau^\star(t^*,x,\mathbbm{p})&=\begin{cases} 
t^* & Mx\leq R, \\
t^*-2Mx-4M\log(2Mx-2M)+R+4M\log(R-2M)& Mx\geq R,
\end{cases}\\
r(t^*, x, \mathbbm{p})&=2Mx.
\end{align*}
Then we define the 2-spheres $\twosphere_{\taus,r}\subset\mcalm$ as the intersection of the level sets of $\taus$ and $r$:
\begin{align*}
\twosphere_{\taus,r}:=\big\{\{t^*\}\times\{x\}\times\twosphere\subset\mcalm \,|\, \taus(t^*, x, \mathbbm{p})=\taus, 2Mx=r\text{ for all }\mathbbm{p}\in\twosphere\big\}.
\end{align*}
\end{definition}
\begin{remark}\label{rmkeikonal}
	For $x>\tfrac{R}{M}$ the (smooth) function $\taus$ solves the eikonal equation:
	\begin{align*}
	g_M^{-1}(\exd \taus, \exd\taus)\big|_{p}=0\qquad \text{for } p\in\mcalm \text{ such that } x(p)>\tfrac{R}{M} .
	\end{align*}
	In addition, $r$ is the area radius function of the 2-spheres $\{t^*\}\times\{x\}\times\twosphere\subset\mcalm$ given as the intersection of the level sets of $t^*$ and $x$.
\end{remark}

A Penrose diagram depicitng this foliation follows.

\begin{figure}[!h]
	\centering
	\begin{tikzpicture}
	\node (I)    at ( 4,0)   {};
	\path 
	(I) +(90:4)  coordinate (Itop)
	+(-90:4) coordinate (Ibot)
	+(180:4) coordinate (Ileft)
	+(0:4)   coordinate (Iright)
	;
	\draw  (Ileft) -- (Itop)node[midway, above, sloped] {$\eh$} -- (Iright) node[midway, above right]    {$\cal{I}^+$} -- (Ibot)node[midway, below right]    {} -- (Ileft)node[midway, below, sloped] {} -- cycle;
	
	\draw[dashed]  
	(Ibot) to[out=50, in=-50, looseness=0.75] node[pos=0.475, above, sloped] {$r=R$        }($(Itop)!.5!(Itop)$) ;
	
	\draw 
	($(Itop)!.8!(Ileft)$) to[out=0, in=0, looseness=0.05] ($(Iright)!.5!(Iright)$);
	
	\draw 
	($(Itop)!.3!(Ileft)$) to[out=0, in=0, looseness=0.05] ($(Iright)!.5!(Iright)$);
	
	\draw 
	($(Itop)!.6!(Iright)$) --node[midway, below, sloped]{$\Xi _{\tau_1 ^*}$} (5.12, 0.32);
	
	\draw 
	($(Itop)!.42!(Iright)$) --node[right, below, sloped]{$\Xi _{\tau_2 ^*} $} (4.975, 1.66);
	
	\node [below] at (3,0.5) {$\Sigma_1$};
	
	\node [below] at (3,2.5) {$\Sigma_2$};
	
	\end{tikzpicture}
	\caption{A Penrose diagram of $\Mg$ depicting the hypersurfaces $\Sigma_{t^*}$ and $\Xi_{\taus}$ given as the level sets of $t^*$ and $\taus$.} 
	\label{Figure1}
\end{figure}
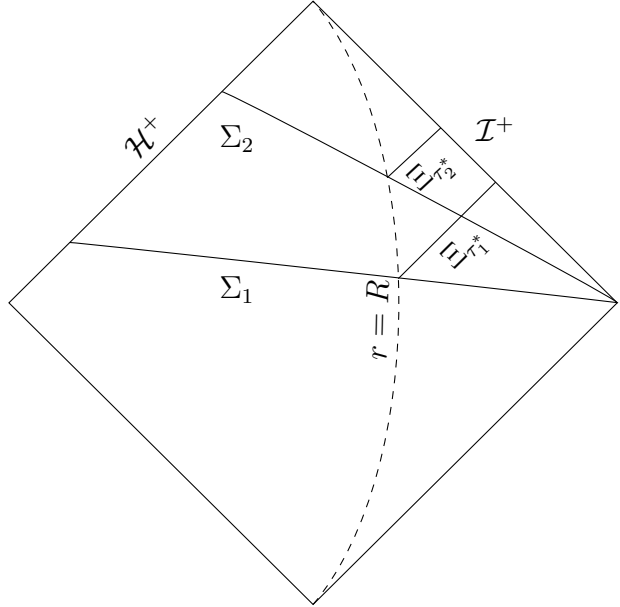

\subsubsection{$Q$ and $S$ tensor algebra}\label{QandStensoralgebra}

First we define the notion of $Q$ and $S$ vector fields on $\mcalm$. 

In what follows we let $p\in\mcalm, X\in\smonecon$ and set $\pr:=\tfrac{1}{2M}\partial_x$.
\begin{definition}\label{defnQandSvectorfields}
Let $\taus_{x_0}:\mcalm\rightarrow\reals$ be the smooth function $\taus_{x_0}(t^*, x,\mathbbm{p})=\taus(t^*, x_0, \mathbbm{p})$ for $x_0\in[1, \infty)$.

Then we say that $X$ is a smooth $Q$ vector field if $X\in\Gamma(TQ)$ where $TQ$ is the subbundle of $T\mcalm$ with fibres
\begin{align*}
T_pQ=\Big\{Y\in T_p\mcalm \,\big|\, Y\in\textnormal{span}\big\{\pt\big|_p,\px\big|_p\big\}\Big\}.
\end{align*}
Conversely, we say that $X$ is a smooth $S$ vector field if $X\in\Gamma(TS)$ where $TS$ is the subbundle of $T\mcalm$ with fibres
\begin{align*}
T_pS=\Big\{Y\in T_p\mcalm \,\big|\, Y(\taus_{x(p)})=Y(r)=0\Big\}.
\end{align*}
\end{definition}
\begin{remark}\label{rmkfibres}
The fibres of $TQ$ and $TS$ have dimension two.
\end{remark}
We denote by $T_pN_Q$ and $T_pN_S$ the normal subspaces to $T_pQ$ and $T_pS$ under $g_M$ respectively\footnote{Note that $\pt, \pr \not\in TS$.}:
\begin{align*}
T_pN_Q&=\big\{Y\in T_p\mcalm \,\big|\, g_M\big|_p(Y, Z)=0\text{ for all } Z\in T_pQ\big\}\setminus\big\{L\in T_pQ\,\big|\,g_M\big|_p(L,L)=0\big\},\\
T_pN_S&=\big\{Y\in T_p\mcalm \,\big|\, g_M\big|_p(Y, Z)=0\text{ for all } Z\in T_pS\big\}.
\end{align*}
We thus have the orthogonal decompositions\footnote{In fact, it follows from the spherical symmetry of $g_M$ that $T_pN_Q=T_pS$ and $T_pN_S=T_pQ$.}
\begin{align*}
T_p\mcalm&=T_pQ\oplus T_pN_Q,\\
T_p\mcalm&=T_pS\oplus T_pN_S
\end{align*}
along with the associated projection maps
\begin{align*}
\tilde{\pi}_p&:T_p\mcalm\rightarrow T_pQ,\\
\slashed{\pi}_p&:T_p\mcalm\rightarrow T_pS.
\end{align*} 
This leads to the following decomposition of smooth vector fields.
\begin{definition}\label{defndecompvectorfields}
Let $X\in\smonecon$.

Then we define the projection of $X$ onto $\Gamma(TQ)$ to be the smooth $Q$ vector field $\qX$ defined by
\begin{align*}
\qX\big|_p=\tilde{\pi}_p\big(X\big|_p\big).
\end{align*}
Conversely, we define the projection of $X$ onto $\Gamma(SQ)$ to be the smooth $S$ vector field $\sX$ defined by
\begin{align*}
\sX\big|_p=\slashed{\pi}_p\big(X\big|_p\big).
\end{align*}
\end{definition}

This subsequently determines a decomposition of smooth $n$-covariant tensor fields into $n$-covariant $Q$ and $S$ tensor fields respectively.
\begin{definition}\label{defnQandStensors}
Let $T\in\Gamma(T^nT^*\mcalm)$ for $n\in\mathbb{N}$.

Then we say that $T$ is a smooth $n$-covariant $Q$ tensor field if $T\in\Gamma(T^nT^*Q)$ where $T^nT^*Q$ is the subbundle of $T^nT^*\mcalm$ with fibres
\begin{align*}
T^nT^*_pQ=\Big\{U\in T^nT^*_p\mcalm \,\big|\, U(...,Y,...)=0 \text{ for any } Y\in T_pS\Big\}.
\end{align*}
Otherwise, we define the projection of $T$ onto $\Gamma(T^nT^*Q)$ to be the smooth n-covariant $Q$ tensor field $\widetilde{T}$ defined by
\begin{align*}
\widetilde{T}(X_1,...,X_n)=T(\qX_1,..., \qX_1)\text{ for } X_1,...,X_n\in\smonecon.
\end{align*}

Conversely, we say that $T$ is a smooth $n$-covariant $S$ tensor field if $T\in\Gamma(T^nT^*S)$ where $T^nT^*S$ is the subbundle of $T^nT^*\mcalm$ with fibres
\begin{align*}
T^nT^*_pS=\Big\{U\in T^nT^*_p\mcalm \,\big|\, U(...,Y,...)=0 \text{ for any } Y\in T_pQ\Big\}.
\end{align*}
Otherwise, we define the projection of $T$ onto $\Gamma(T^nT^*S)$ to be the smooth n-covariant $S$ tensor field $\slashed{T}$ defined by
\begin{align*}
\slashed{T}(X_1,...,X_n)=T(\sX_1,..., \sX_1)\text{ for } X_1,...,X_n\in\smonecon.
\end{align*}
\end{definition}
Note we will use the convention that $\Gamma(\mcalm)=\Gamma(T^0T^*Q)=\Gamma(T^0T^*S)$.

In this paper we are only interested in the case where the above projection is applied to 1-forms and symmetric 2-covariant tensors. Consequently, if $T\in\smonecov$ then the above decomposition completely determines $T$:
\begin{align*}
T=\widetilde{T}+\slashed{T}.
\end{align*}
If however $T\in\smsymtwocov$ then one must supplement the above with a further decomposition.
\begin{definition}\label{defnqmsmonesforms}
Let $T\in\Gamma(S^2T^*\mcalm)$.

Then we say that $T$ is a smooth $\qmsm$ 1-form if $T\in\Gamma(T^*Q\astrosun T^*S)$.

Otherwise, we define the projection of $T$ onto $\Gamma(T^*Q\astrosun T^*S)$ to be the smooth $\qmsm$ 1-form $\stkout{T}$ defined by
\begin{equation*}
\stkout{T}(X_1,X_2)=T(X_1, X_2)-\widetilde{T}(X_1, X_2)-\slashed{T}(X_1, X_2),\qquad X_1,X_2\in\smonecon.
\end{equation*}
\end{definition}

\subsubsection{$Q$ and $S$ tensor analysis}\label{QandStensoranalysis}

We now introduce several operations on $\mcalm$ that act naturally on symmetric $Q$ and $S$ tensor fields.\newline

Let $\qg$ be the projection of $g_M$ onto $\smsymtwocovQ$. Then we have the following natural operations on $\smncovQ$:
\begin{itemize}
	\item for $n=1$ the sharp operator $\qsharp:\smonecovQ\rightarrow\smoneconQ$ defined according to
	\begin{align*}
	(\qtau^{\qsharp})^a=\qg^{ab}\qtau_b
	\end{align*}
	\item for $n=1$ the Hodge-star operator $\qhd:\smonecovQ\rightarrow\smonecovQ$ defined according to
	\begin{align*}
	(\qhd\qtau)_a=\qepsilon_{ba}(\qtau^{\qsharp})^b
	\end{align*}
where $\qepsilon$ is the unique 2-form on $\mcalm$ such that $\big(\qg^{-1}\big)^{ac}\big(\qg^{-1}\big)^{bd}\qepsilon_{cd}\qepsilon_{ab}=2$ and $\qepsilon(\pt, \px)>0$
\item for $n=2$ the trace operator $\qtr:\smcovQ{2}\rightarrow\Gamma(\mcalm)$ defined according to
\begin{align*}
\qtr\qtau=\big(\qg^{-1}\big)^{ab}\qtau_{ab}
\end{align*}
\item for $n=2$ the traceless operator  $\,\widehat{}:\smsymcovQ{2}\rightarrow\smsymtfcovQ{2}$ defined according to
\begin{align*}
\qhattau=\qtau-\tfrac{1}{2}\qg\,\qtrtau
\end{align*}
with $\smsymtfcovQ{2}:=\smsymcovQ{2}\cap\textnormal{image}(\,\,\widehat{}\,\,)$
\item for $n=1$ the traceless symmetric product $\qastrosunhat:\smonecovQ{\times}\smonecovQ\rightarrow\smsymtfcovQ{2}$ defined according to
\begin{align*}
\big(\qtau_1\qastrosunhat\qtau_2\big)_{ab}=\widehat{\big(\qtau_1\astrosun\qtau_2\big)}_{ab}
\end{align*}
\item for $n=1$ the exterior derivative $\qexd:\smfun\rightarrow\smonecovQ$ defined according to
\begin{align*}
\qexd f(X)=\exd f(\qX),\qquad X\in\smonecon
\end{align*}
with $\exd$ the exterior derivative on $\mcalm$
\end{itemize}

Observe now that $\qg$ defines a Lorentzian metric on $TQ$. We thus denote by $\{\tilde{e}_0, \tilde{e}_1\}$ an associated orthonormal frame of $TQ$ and by $\qn$ the associated Levi-Civita connection which we extend to act on $\smncovQ$ for $n\geq 0$ in the standard fashion. Then we have the following natural differential operators acting on $\smncovQ$:
\begin{itemize}
	\item for $n\geq 1$ the divergence operator  $-\qdiv:\smncovQ\rightarrow\Gamma(T^{n-1}T^*Q)$ defined according to
	\begin{align*}
	-\qdiv \widetilde{\tau}(X_1,..., X_{n-1})=\sum_{i=0,1}(-1)^{i+1}\qn_{\tilde{e}_i}\widetilde{\tau}(\tilde{e}_i, \qX_1,..., \qX_{n-1}),\qquad X_1,..., X_{n-1}\in \Gamma(TQ)
	\end{align*}
	\item the d'Alembertian $\qbox:\smncovQ\rightarrow\smncovQ$ defined according to
	\begin{align*}
	\qbox \widetilde{\tau}=\sum_{i=0,1}(-1)^{i+1}\Big(\qn^2_{\tilde{e}_i}\widetilde{\tau}-\qn_{\qn_{\tilde{e}_i}\tilde{e}_i}\widetilde{\tau}\Big)
	\end{align*}
\item for $n=1$ the Lie derivative $\qn{\astrosun}:\smonecovQ\rightarrow\smsymcovQ{2}$ defined according to
\begin{align*}
\qn{\astrosun}\stau(X_1, X_2)=\qn_{\sX_1}\qtau(\sX_2)+\qn_{\sX_2}\qtau(\sX_1),\qquad X_1,X_2\in\smonecon
\end{align*}
\item for $n=1$ the traceless Lie derivative $\qn\astrosunhat:\smonecovQ\rightarrow\smsymtfcovQ{2}$ defined according to
\begin{align*}
\qn\astrosunhat\qtau=\widehat{\qn{\astrosun}\qtau}
\end{align*}
\end{itemize}

Let now $\sg$ be the projection of $g_M$ onto $\smsymtwocovS$. Then we have the following natural operations acting on $\smncovS$:
\begin{itemize}
	\item for $n=1$ the sharp operator $\ssharp:\smonecovS\rightarrow\smoneconS$ defined according to
	\begin{align*}
	(\stau^{\ssharp})^a=\sg^{ab}\qtau_b
	\end{align*}
	\item for $n=1$ the Hodge-star operator $\shd:\smonecovS\rightarrow\smonecovS$ defined according to
	\begin{align*}
	(\shd\stau)_a=\sepsilon_{ba}(\stau^{\ssharp})^b
	\end{align*}
	where $\sepsilon$ is the unique 2-form on $\mcalm$ such that $\big(\sg^{-1}\big)^{ac}\big(\sg^{-1}\big)^{bd}\sepsilon_{cd}\sepsilon_{ab}=2$ and $\qexd t^*\wedge\qexd x\wedge\sepsilon$ determines the same orientation class as $\epsilon:=8M^3x^2\exd t^*\wedge\exd x\wedge\pi_{\twosphere}^*\epsilon_{\twosphere}$ with $\epsilon_{\twosphere}$ the volume form on $\twosphere$ associated to $\roundmetric$
	\item for $n=2$ the trace operator $\str:\smcovS{2}\rightarrow\Gamma(\mcalm)$ defined according to
	\begin{align*}
	\str\slashed{\tau}=\big(\sg^{-1}\big)^{ab}\slashed{\tau}_{ab}
	\end{align*}
	\item for $n=2$ the traceless operator  $\,\widehat{}:\smsymcovS{2}\rightarrow\smsymtfcovS{2}$ defined according to
	\begin{align*}
	\shattau=\stau-\tfrac{1}{2}\sg\,\strtau
	\end{align*}
	with $\smsymtfcovS{2}:=\smsymcovS{2}\cap\textnormal{image}(\,\,\widehat{}\,\,)$
\item for $n=1$ the traceless symmetric product $\sastrosunhat:\smonecovS{\times}\smonecovS\rightarrow\smsymtfcovS{2}$ defined according to
\begin{align*}
\big(\stau_1\sastrosunhat\stau_2\big)_{ab}=\widehat{\big(\stau_1\astrosun\stau_2\big)}_{ab}
\end{align*}
	\item for $n=0$ the exterior derivative $\sexd:\smfun\rightarrow\smonecovS$ defined according to
	\begin{align*}
	\sexd f(X)=\exd f(\sX),\qquad X\in\smonecon
	\end{align*}
	\item the pointwise norm 
	\begin{align*}
	|\stau|_{\sg}:=(\sg^{-1})^{a_1b_1}...(\sg^{-1})^{a_nb_n}\stau_{a_1...a_n}\stau_{b_1...b_n}
	\end{align*}
\end{itemize}
Observe now that $\sg$ defines a Riemannian metric on $TS$. We thus denote by $\{\slashed{e}_2, \slashed{e}_3\}$ an associated orthonormal frame of $TS$ and by $\sn$ the associated Levi-Civita connection which we extend to act on $\smncovS$ for $n\geq 0$ in the standard fashion. Then we have the following natural differential operators acting on $\smncovS$:
\begin{itemize}
	\item for $n\geq 1$ the divergence operator  $\sdiv:\smncovS\rightarrow\Gamma(T^{n-1}T^*S)$ defined according to
	\begin{align*}
	\sdiv \stau(X_1,..., X_{n-1})=\sum_{i=2,3}\qn_{\tilde{e}_i}\stau(\tilde{e}_i, \sX_1,..., \sX_{n-1}),\qquad X_1,..., X_{n-1}\in \Gamma(TS)
	\end{align*}
	\item the Laplacian $\slap:\smncovS\rightarrow\smncovS$ defined according to
	\begin{align*}
	\slap\slashed{\tau}=\sum_{i=1,2}\Big(\sn^2_{\slashed{e}_i}\slashed{\tau}-\sn_{\sn_{\slashed{e}_i}\slashed{e}_i}\slashed{\tau}\Big)
	\end{align*}
	\item for $n=1$ the Lie derivative $\sn{\astrosun}:\smonecovS\rightarrow\smsymcovS{2}$ defined according to
	\begin{align*}
	\sn{\astrosun}\stau(X_1, X_2)=\sn_{\sX_1}\stau(\sX_2)+\sn_{\sX_2}\stau(\sX_1),\qquad X_1,X_2\in\smonecon
	\end{align*}
	\item for $n=1$ the traceless Lie derivative $\sn\astrosunhat:\smonecovS\rightarrow\smsymtfcovS{2}$ defined according to
	\begin{align*}
	\sn\astrosunhat\stau=\widehat{\sn{\astrosun}\stau}
	\end{align*}
\end{itemize}

\subsubsection{Mixed $Q$ and $S$ tensor analysis}\label{MixedQandStensoranalysis}

To define the action of $\qn$ on $S$ tensors and the action of $\sn$ on $Q$ tensors we simply define them via the action of the spacetime operator $\nabla$:
\begin{align*}
\qn_{\qX}\stau:&=\nabla_{\qX}\stau, \qquad\qX\in\smoneconQ, \stau\in\smncovS,\\
\sn_{\sX}\qtau:&=\nabla_{\sX}\qtau, \qquad\sX\in\smoneconS, \qtau\in\smncovQ.
\end{align*}
This moreover allows the generalisation of all the differential operators of the previous section. In addition for $\tau\in\smqmsm$ we define
\begin{align*}
\qn_{\qX}\mtau:&=\nabla_{\qX}\mtau,\\
\sn_{\sX}\mtau:&=\nabla_{\sX}\mtau.
\end{align*}
which moreover allows the generalisation of all the differential operators of the previous section to smooth $\qmsm$ 1-forms.

\subsection{The projection of $\smonecov$ and $\smsymtwocov$ onto and away from the $l=0,1$ spherical harmonics}\label{Theprojectionofsmoncovandsmsymtwocovontoandawayfromthel=0,1sphericalharmonics}

We define now the projection of smooth 1-forms and smooth, symmetric 2-covariant tensor fields onto and away from the $l=0,1$ spherical harmonics. We begin in section \ref{Theprojectionofsmoothfunctionson twosphereontoandawayfromthel=0,1sphericalharmonics} by recalling the classical projection of smooth functions on $\twosphere$ onto and away from the $l=0,1$ spherical harmonics. Then in section \ref{Theprojectionofsmfunontoandawayfromthel=0,1sphericalharmonics} we upgrade this projection to smooth functions on $\mcalm$. Then in section \ref{TheprojectionofsmoothQtensors,smoothsymmetrictracelessStensorsandsmoothqmsm1-formsontoandawayfromthel=0,1sphericalharmonics} we upgrade this to smooth to smooth $Q$ tensors,smooth, symmetric, traceless, $S$ tensors and smooth $\qmsm$ 1-forms. Finally in section \ref{Theprojectionofsmoncovandsmsymtwocovontoandawayfromthel=0,1sphericalharmonics} we upgrade this projection further to act on smooth 1-forms and smooth symmetric 2-covariant tensor fields on $\mcalm$.

\subsubsection{The projection of smooth functions on $\twosphere$ onto and away from the $l=0,1$ spherical harmonics}\label{Theprojectionofsmoothfunctionson twosphereontoandawayfromthel=0,1sphericalharmonics}

We recall the classical spherical harmonics $Y^l_m$ with $l\in\mathbb{N}$ and $m\in\{-l,...,0,...l\}$ defined as the set of orthogonal eigenfunctions for the Laplacian $\Delta_{\twosphere}$ associated to the metric $\twosphere$ on the unit round sphere:
\begin{align*}
\Delta_{\twosphere}Y^l_m=-l(l+1)Y^l_m
\end{align*}
and
\begin{align*}
\int_{\twosphere}Y^l_m\,Y^{l'}_{m'}\,\epsilon_{\twosphere}=\delta^{ll'}\delta_{mm'}.
\end{align*}
Here, $\delta$ is the Kronecker delta symbol. 

We explicitly note the form of the $l=0$ and $l=1$ modes
\begin{align}
Y^{l=0}_{m=0}&=\frac{1}{\sqrt{4\pi}}\label{l=0sphericalharmonics}\\
Y^{l=1}_{m=-1}=\sqrt{\frac{3}{4\pi}}\sin\theta\cos\varphi,\qquad Y^{l=1}_{m=0}&=\sqrt{\frac{3}{8\pi}}\cos\theta,\qquad Y^{l=1}_{m=1}=\sqrt{\frac{3}{4\pi}}\sin\theta\sin\varphi.\label{l=1sphericalharmonics}
\end{align}
The classical spherical harmonic decomposition of smooth functions $f$ on $\twosphere$ then says that 
\begin{align*}
f=\sum_{l,m}f_i^jY^i_j, \qquad f_i^j\in\reals.
\end{align*}
this leads to the following defintion

\begin{definition}\label{defnfunctiontwospheresupportedl=0,1}
	Let $f$ be a smooth function on $\twosphere$ and let $\epsilon_{\twosphere}$ be a volume form associated to $\roundmetric$. 
	
	Then we say that $f$ is supported only on $l=0,1$ iff for every $l\geq 2$
	\begin{align*}
	\int_{\twosphere}f\cdot Y^l_m\,\epsilon_{\twosphere}=0.
	\end{align*}
	
	Conversely, we say that $f$ has vanishing projection to $l=0,1$ iff
	\begin{align*}
	\int_{\twosphere}f\cdot Y^0_m\,\epsilon_{\twosphere}=	\int_{\twosphere}f\cdot Y^1_m\,\epsilon_{\twosphere}=0.
	\end{align*}
\end{definition}
Combining this with the above decomposition gives:
\begin{proposition}\label{propsphericalharmonicdecompfunc}
	Let $f$ be a smooth function on $\twosphere$. Then one has the unique decomposition
	\begin{align*}
	f=\mdpart{f}+f_{l\geq 2}
	\end{align*}
	where $\mdpart{f}$ is a smooth function on $\twosphere$ supported only on $l=0,1$, $\rpart{f}$ a smooth function on $\twosphere$ with vanishing projection to $l=0,1$ and
	\begin{align*}
	\int_{\twosphere}\mdpart{f}\rpart{f}\,\epsilon_{\twosphere}=0.
	\end{align*}
\end{proposition}

\subsubsection{The projection of $\smfun$ onto and away from the $l=0,1$ spherical harmonics}\label{Theprojectionofsmfunontoandawayfromthel=0,1sphericalharmonics}

Observe that
\begin{align}\label{Smetric}
\sg=r^2\pi^*_{\twosphere}\roundmetric
\end{align}

Therefore defining
\begin{align*}
\slashed{Y}^l_m=\pi_{\twosphere}^*Y^l_m
\end{align*}
one easily computes that
\begin{align*}
\slap\slashed{Y}^l_m =-\frac{l(l+1)}{r^2}\slashed{Y}^l_m
\end{align*}
and
\begin{align*}
\frac{1}{r^2}\int_{\twosphere_{\taus,r}}\slashed{Y}^l_m\,\slashed{Y}^{l'}_{m'}\,\sepsilon=\delta^{ll'}\delta_{mm'}.
\end{align*}
Given then $f\in\smfun$ one readily shows
\begin{align*}
f=\sum_{l,m}f_i^j\slashed{Y}^i_j,
\end{align*}
where $f_i^j$ are smooth functions of $t^*$ and $x$.

\begin{definition}\label{defnfunctionsupportedl=0,1}
	Let $f\in\Gamma(\mcalm)$. 
	
	Then we say that $f$ is supported only on $l=0,1$ iff for every $l\geq 2$ and for every 2-sphere $\twosphere_{\taus,r}$
	\begin{align*}
	\int_{\twosphere_{\taus,r}}f\cdot \slashed{Y}^l_m\,\sepsilon=0.
	\end{align*}
	
	Conversely, we say that $f$ has vanishing projection to $l=0,1$ iff for every 2-sphere $\twosphere_{\taus,r}$
	\begin{align*}
	\int_{\twosphere_{\taus,r}}f\cdot \slashed{Y}^0_m\,\sepsilon=\int_{\twosphere_{\taus,r}}f\cdot \slashed{Y}^1_m\,\sepsilon=0.
	\end{align*}
	In this latter case we say $f\in\smfunrad$.
\end{definition}

\begin{proposition}\label{propsphericalharmonicdecomp}
	Let $f\in\Gamma(\mcalm)$. Then one has the unique decomposition
	\begin{align*}
	f=\mdpart{f}+f_{l\geq 2}
	\end{align*}
	where $\mdpart{f}\in \Gamma(\mcalm)$ is supported only on $l=0,1$, $\rpart{f}\in \smfunrad$ and for every 2-sphere $\twosphere_{\taus,r}$
	\begin{align*}
	\int_{\twosphere_{\taus, r}}\mdpart{f}\rpart{f}\,\sepsilon=0.
	\end{align*}
\end{proposition}

\subsubsection{The projection of smooth $Q$ tensors, smooth symmetric traceless $S$ tensors and smooth $\qmsm$ 1-forms onto and away from the $l=0,1$ spherical harmonics}\label{TheprojectionofsmoothQtensors,smoothsymmetrictracelessStensorsandsmoothqmsm1-formsontoandawayfromthel=0,1sphericalharmonics}

In order to extend Proposition \ref{propsphericalharmonicdecomp} to tensor fields we first need the following Hodge decomposition.

In what follows we extend the connections $\qn$ and $\sn$ to $\smonecon$ by $\qn_X:=\qn_{\qX}$ and $\sn_X:=\sn_{\sX}$ thus yielding covariant derivative operators on $\mcalm$.

\begin{proposition}\label{prophodgedecomp}
	Let $\sxi\in\smonecovS$. Then there exist two unique functions $\even{\sxi}, \odd{\sxi}\in\smfun\setminus\ker\slap$ such that $\sxi$ has the (unique) representation
	\begin{align*}
	\sxi&=\sdso\big(\even{\sxi}, \odd{\sxi}\big)\\
	:&=\sn\even{\sxi}+\shd\sn\odd{\sxi}
	\end{align*}
	with $\ker\slap$ spanned by functions depending only on $t^*$ and $x$.
	
	Let now $\shatxi\in\smsymtfcovS{2}$. Then there exist two unique functions $\even{\shatxi}, \odd{\shatxi}\in\smfun\setminus\ker\slap\big(\slap+\tfrac{2}{r^2}\Id\big)$ such that $\shatxi$ has the (unique) representation
	\begin{align*}
	\shatxi&=\sn\astrosunhat\sdso\big(\even{\shatxi}, \odd{\shatxi}\big).
	\end{align*}
	with $\smfunrad=\smfun\setminus\ker\slap\big(\slap+\tfrac{2}{r^2}\Id\big)$.
\end{proposition}
\begin{proof}
The general decomposition is a special case of \cite{CKstabmink}. The statement about the kernel follows from \cite{DHRlinstabschwarz}.
\end{proof}

One interpretation of the above is that smooth functions provide a global basis for smooth $S$ 1-forms and smooth, symmetric, traceless 2-covariant $S$ tensors. The tensor products of $Q$ 1-forms $\qexd t^*$ and $\qexd x$ also clearly span the space of smooth $Q$ tensor fields. We can then combine these facts to determine a basis of smooth $\qmsm$ 1-forms.

\begin{proposition}\label{prophodgedecompqmsm1form}
	Let $\mtau\in\smqmsm$. Then there exist two unique $Q$ 1-forms $\even{\mtau}, \odd{\mtau}\in\smonecovQ\setminus\ker\slap$ such that $\mtau$ has the (unique) representation
	\begin{align*}
	{\mtau}&=\sdso\astrosun\big(\even{\mtau}, \odd{\mtau}\big)\\
	:&=\sn\astrosun\even{\mtau}+\shd\sn\astrosun\odd{\mtau}.
	\end{align*}
\end{proposition}
with $\ker\slap$ spanned by $f \qexd t^*+g\qexd x$ with $f,g$ functions depending only on $t^*$ and $x$.
\begin{proof}
	Since $\big\{\qexd t^*, \qexd x\big\}$ is a global frame for $T^*Q$ one has the decomposition
	\begin{align*}
	{\mtau}=\ssigma\astrosun\qexd t^*+\srho\astrosun\qexd r,\qquad \ssigma, \srho\in\smonecovS.
	\end{align*}
	Proposition \ref{prophodgedecomp} then yields the further decomposition
	\begin{align*}
	{\mtau}=\sdso\big(\even{\ssigma}, \odd{\ssigma}\big)\astrosun\qexd t^*+\sdso\big(\even{\srho}, \odd{\srho}\big)\astrosun\qexd r,\qquad \even{\ssigma}, \odd{\ssigma}, \even{\srho}, \odd{\srho}\in\smfun\setminus\ker\slap.
	\end{align*}
	Defining thus $\even{\mtau}, \odd{\mtau}\in\smonecovQ\setminus\ker\slap$ according to
	\begin{align*}
	\even{\mtau}=\even{\ssigma}\qexd t^*+\even{\srho}\qexd x,\qquad\odd{\mtau}=\odd{\ssigma}\qexd t^*+\odd{\srho}\qexd x
	\end{align*}
	the proposition follows.
\end{proof}

\begin{definition}\label{defnqmtensorssupportedonlgeq2}
	Let $\qtau\in\smncovQ$ for $n\in\mathbb{N}$. Then we say that $\qtau$ is supported only on $l=0,1$ iff the components of $\qtau$ in the frame $\{\qexd t^*, \qexd x\}$ are supported only on $l=0,1$ (cf. section \ref{MixedQandStensoranalysis}). Conversely, we say that $\qtau$ has vanishing projection to $l=0,1$ iff the components of $\qtau$ in the frame $\{\qexd t^*, \qexd x\}$ have vanishing projection to $l=0,1$.

Let now $\mtau\in\smqmsm$. Then we say that $\qtau$ is supported only on $l=0,1$ iff $\even{\mtau}$ and $\odd{\mtau}$ are supported only on $l=0,1$ (cf. Proposition \ref{prophodgedecompqmsm1form}). Conversely, we say that $\mtau$ has vanishing projection to $l=0,1$ iff $\even{\mtau}$ and $\odd{\mtau}$ have vanishing projection to $l=0,1$.

Finally, let $\shattau\in\smsymtratwocovS$. Then we say that $\shattau$ is supported only on $l=0,1$ iff $\even{\shattau}$ and $\odd{\shattau}$ are supported only on $l=0,1$ (cf. Proposition \ref{prophodgedecomp}). Conversely, we say that $\shattau$ has vanishing projection to $l=0,1$ iff $\even{\shattau}$ and $\odd{\shattau}$ have vanishing projection to $l=0,1$.
\end{definition}

Note Proposition \ref{prophodgedecomp} thus yields that smooth, symmetric, traceless 2-covariant $S$ tensors automatically have vanishing projection to $l=0,1$.

\begin{proposition}\label{propsphericalharmonicdecompQStensors}
	Let $\qtau\in\smncovQ$ for $n\in\mathbb{N}$. Then one has the unique decomposition
	\begin{align*}
	\qtau=\mdpart{\qtau}+\qtau_{l\geq 2}
	\end{align*}
	where $\mdpart{\qtau}\in \smncovQ$ is supported only on $l=0,1$, $\rpart{\qtau}\in \smncovQ$ has vanishing projection to $l=0,1$ and for every 2-sphere $\twosphere_{\taus,r}$
	\begin{align*}
	\int_{\twosphere_{\taus, r}}\big(\mdpart{\qtau}\big)_{ij}\big(\rpart{\qtau}\big)_{kl}\,\sepsilon=0
	\end{align*}
	with $\big(\mdpart{\qtau}\big)_{ij}$ and $\big(\rpart{\qtau}\big)_{kl}$ the coefficients of $\mdpart{\qtau}$ and $\rpart{\qtau}$ in the frame $\{\qexd t^*, \qexd x\}$ respectively.
	
		Let now $\mtau\in\smqmsm$. Then one has the unique decomposition
		\begin{align*}
		{\mtau}=\mdpart{\mtau}+\mtau_{l\geq 2}
		\end{align*}
		where $\mdpart{\mtau}\in \smqmsm$ is supported only on $l=0,1$, $\rpart{\mtau}\in \smqmsm$ has vanishing projection to $l=0,1$ and for every 2-sphere $\twosphere_{\taus,r}$
		\begin{align*}
		\int_{\twosphere_{\taus, r}}\even{\big(\mdpart{\mtau}\big)}\even{\big(\rpart{\mtau}\big)}\,\sepsilon=\int_{\twosphere_{\taus, r}}\odd{\big(\mdpart{\mtau}\big)}\odd{\big(\rpart{\mtau}\big)}\,\sepsilon=0.
		\end{align*}
		
\end{proposition}

\subsubsection{The projection of $\smonecov$ and $\smsymtwocov$ onto and away from the $l=0,1$ spherical harmonics}\label{Theprojectionofsmsymtwocovontoandawayfromthel=0,1sphericalharmonics}

The projection of smooth 1-forms and smooth symmetric, 2-covariant tensor fields on $\mcalm$ is then by first decomposing them as in section \ref{QandStensoralgebra} and then using the projections of the previous section.

\subsection{Elliptic operators on 2-spheres}\label{Ellipticoperatorson2-spheres}

In this section we introduce an $L^2$ norm on the 2-spheres of section \ref{Ageometricfoliationby2-spheres}. A family of elliptic operators acting on tensor fields on $\mcalm$ are then defined in section \ref{ThefamilyofoperatorssA}-\ref{Thefamilyofoperatorsslap} for which elliptic estimates will be derived with respect to these norms.

\subsubsection{Norms on spheres}\label{Normsonspheres}

First we define the norms through which the higher order angular derivatives are to be measured.

\begin{definition}\label{defnL2norm}
Let $Q$ be a symmetric $2$-covariant $\qm$-tensor, let $\momega$ be an $\qmsm$ 1-form and let $\Theta$ be an $n$-covariant $S$-tensor for $n\geq 0$ an integer. Then on any 2-sphere $\geomtwosphere{}{r}$ we define the $L^2$ norm $\normtwosphere{\cdot}{}{r}$ according to 
	\begin{align*}
	\normtwosphere{Q}{}{r}:&=\int_{\twosphere_{\taus,r}}\Big(|Q_{\pt\pt}|^2+|Q_{\pt\pr}|^2+|Q_{\pr\pr}|^2\Big)r^{-2}\sepsilon,\\
	\normtwosphere{\momega}{}{r}:&=\int_{\twosphere_{\taus,r}}\Big(|\momega_{\pt}|_{\sg}^2+|\momega_{\pr}|_{\sg}^2\Big)r^{-2}\sepsilon,\\
	\normtwosphere{\Theta}{}{r}:&=\int_{\twosphere_{\taus,r}}|\Theta|_{\sg}^2r^{-2}\sepsilon.
	\end{align*}
	Moreover, the higher order norms $\normtwosphere{(r\sn)^k\cdot}{}{r}$ for $k\geq 1$ are defined according to
	\begin{align*}
	\normtwosphere{(r\sn)^kQ}{}{r}:&=\int_{S^2_{\taus,r}}\Big(|(r\sn)^k Q_{\pt\pt}|_{\sg}^2+|(r\sn)^k Q_{\pt\pr}|_{\sg}^2+|(r\sn)^k Q_{\pr\pr}|_{\sg}^2\Big)r^{-2}\sepsilon,\\
	\normtwosphere{(r\sn)^k\momega}{}{r}:&=\int_{S^2_{\taus,r}}\Big(|(r\sn)^k\momega_{\pt}|_{\sg}^2+|(r\sn)^k\momega_{\pr}|_{\sg}^2\Big)r^{-2}\sepsilon,\\
	\normtwosphere{(r\sn)^k\Theta}{}{r}:&=\int_{S^2_{\taus,r}}|(r\sn)^k\Theta|_{\sg}^2r^{-2}\sepsilon.
	\end{align*}

\end{definition}

\subsubsection{Elliptic estimates on 2-spheres}\label{ThefamilyofoperatorssA}

We continue by introducing a family of operators on $\mcalm$ which shall ultimately serve as a shorthand notation for controlling higher order angular derivatives of tensor fields on $\mcalm$  measured in the norms of the previous section. Indeed, proceeding similarly as in \cite{DHRlinstabschwarz}, we define
\begin{itemize}
	\item the operators $\smcA{i}$ are defined inductively as
	\begin{align*}
	\smcA{2i+1}:=r\sn\smcA{2i},\qquad\smcA{2i+2}:=-r\sdiv\smcA{2i+1}
	\end{align*} 
	with $\smcA{1}=r\sn$
	\item the operators $\vmcA{i}$ are defined inductively as
	\begin{align*}
	\vmcA{2i+1}:=r\sdo\vmcA{2i},\qquad\vmcA{2i+2}:=r\sdso\vmcA{2i+1}
	\end{align*} 
	with $\vmcA{1}=r\sdo$
	\item the operators $\tmcA{i}$ are defined inductively as
	\begin{align*}
	\tmcA{2i+1}:=r\sdiv\tmcA{2i},\qquad\tmcA{2i+2}:=-r\sn\astrosunhat\tmcA{2i+1}
	\end{align*} 
	with $\tmcA{1}=r\sdiv$
\end{itemize}
Before we derive the elliptic estimates we note the following lemma (we recall by convention that a 0-covariant $\qm$-tensor is a function on $\mcalm$).
\begin{lemma}\label{lemmapoincare}
	Let $Q\in\smncovQ$ for $n\geq 0$ an integer have vanishing projection to $l=0,1$. Then for $i=0,...,5$ and any 2-sphere $\geomtwosphere{}{r}$ one has the estimate
	\begin{align*}
	(6-i)\|(r\sn)^iQ\|^2_{\twosphere_{\taus, r}}&\leq\|(r\sn)^{i+1}Q\|^2_{\twosphere_{\taus, r}}.
	\end{align*}
\end{lemma}
\begin{proof}
	Recalling that $\sg=r^2\roundmetric$ where $\roundmetric$ is the metric on the unit round sphere, the above estimate thus follows from applying the classical Poincar\'e inequality on the 2-spheres $\geomtwosphere{}{r}$ to the components of $Q'$ in the frame $\{\pt, \pr\}$.
\end{proof}
Lastly, we again follow \cite{DHRlinstabschwarz} in introducing a family of angular on $\mcalm$\footnote{Note however the difference in sign convention for the operators $\sdo$ and $\sdso$.}.
\begin{itemize}
\item the operator $\sdo$ is defined by
\begin{align*}
\sdo\tau=-\big(\sdiv\tau, \scurl\tau\big)
\end{align*}
\item the operator $\sdso$ is defined by
\begin{align*}
\sdso(\tau, \tau')=\sn\tau+\slashed{\star}\sn\tau'\newline
\end{align*}
\end{itemize}

This leads to the subsequent elliptic estimates.
\begin{proposition}\label{propellipticestimatesonA}
	Let $Q'$ be a smooth, symmetric $2$-covariant $\qm$-tensor with vanishing projection to $l=0,1$, let $\momega'$ be a smooth $\qmsm$ 1-form with vanishing projection to $l=0,1$ and let $\Theta$ be a smooth, symmetric, traceless 2-covariant $S$-tensor respectively. Then for any 2-sphere $\twosphere_{\taus,r}$ and any integer $m\geq 0$
	\begin{align*}
	\sum_{i=0}^{m}\|(r\sn)^iQ'\|^2_{S^2_{\taus, r}}&\lesssim\normtwosphere{\mcala_f^{[m]}Q'}{}{r},\\
	\sum_{i=0}^{m}\|(r\sn)^i\momega'\|^2_{S^2_{\taus, r}}&\lesssim\normtwosphere{\mcala_\xi^{[m]}\momega'}{}{r},\\
	\sum_{i=0}^{m}\|(r\sn)^i\Theta\|^2_{S^2_{\taus, r}}&\lesssim\normtwosphere{\mcala_\theta^{[m]}\Theta}{}{r}.
	\end{align*}
\end{proposition}
\begin{proof}
	We first note the identities
	\begin{align*}
	\sdso\sdo&=-\slap+\frac{1}{r^2},\\
	-\sn\otimeshat\sdt&=-\slap+\frac{2}{r^2}.
	\end{align*}
	Computing thus in the frame $\{\pt,\pr\}$ one finds that on every 2-sphere $\twosphere_{\taus,r}$
	\begin{align}
	\normtwosphere{\smcA{1}Q'}{}{r}&=\normtwosphere{r\sn Q'}{}{r},\label{identityAf1}\\
	\normtwosphere{\vmcA{1}\momega'}{}{r}&=\normtwosphere{r\sn \momega'}{}{r}+\normtwosphere{\momega'}{}{r},\label{identityAv1}\\
	\normtwosphere{\tmcA{1}\theta'}{}{r}&=\normtwosphere{r\sn \theta'}{}{r}+2\normtwosphere{\theta'}{}{r}\label{identityAt1}
	\end{align}
	and
	\begin{align*}
	\normtwosphere{\smcA{2}Q'}{}{r}&=\normtwosphere{r^2\slap Q'}{}{r},\\
	\normtwosphere{\vmcA{2}\momega'}{}{r}&=\normtwosphere{r^2\slap \momega'}{}{r}+2\normtwosphere{r\sn \momega'}{}{r}+\normtwosphere{\momega'}{}{r},\\
	\normtwosphere{\tmcA{2}\theta'}{}{r}&=\normtwosphere{r^2\slap \theta'}{}{r}+4\normtwosphere{r\sn \theta'}{}{r}+4\normtwosphere{\theta'}{}{r}.
	\end{align*}
	The former along with Lemma \ref{lemmapoincare} immediately yields the $m=1$ case of the proposition whereas the latter combined with elliptic estimates on $\slap$ and Lemma \ref{lemmapoincare} once more yields the $m=2$ case.
	
	The higher order cases then follow by an inductive procedure and Lemma \ref{lemmapoincare}, noting that commuting with higher order derivatives generates positively signed lower order terms.
\end{proof}

\subsubsection{The family of operators $\slap_{a,b}^{-1}$}\label{Thefamilyofoperatorsslap}

The following family of operators will appear later in the paper:
\begin{align*}
\slap_{a,b}:=\slap+\frac{2}{r^2}\Big(a-\tfrac{bM}{r}\Big)\Id
\end{align*}
for $a,b\in\reals$. A particularly important member which will play an important role is
\begin{align*}
\slap_{\mathfrak{Z}}:=\slap_{1,3}.
\end{align*}
We have that these operators are invertible which follows easily from Lemma \ref{lemmapoincare} and the spherical harmonic decomposition.
\begin{proposition}\label{propinverseoperators}
Let $a,b\in\reals$ be such that $0\leq a\leq 3$ and $0\leq b\leq 3$. Then the map $\slap_{a,b}:\Gamma_{l\geq 2}(\mcalm)\rightarrow\Gamma_{l\geq 2}(\mcalm)$ is a bijection. In particular, the inverse $\slap^{-1}_{a,b}:\Gamma_{l\geq 2}(\mcalm)\rightarrow\Gamma_{l\geq 2}(\mcalm)$ is well defined.
\end{proposition}

We have the following proposition.
\begin{proposition}\label{propellipticestimateszlsap}
	Let $f\in\smfunrad$. Then for any integer $p\geq 1$ and any 2-sphere $S^2_{\taus,r}$ one has the elliptic estimates
	\begin{align*}
	\sum_{i=0}^{2p}||(r\sn)^i\zslapinv{p}f||^2_{\twosphere_{\taus,r}}\lesssim||f||^2_{\twosphere_{\taus,r}}.
	\end{align*}
\end{proposition}
\begin{proof}
	Integrating by parts on any 2-sphere $\geomtwosphere{}{r}$ one finds
	\begin{align*}
	\normtwosphere{\slap_{\mathfrak{Z}} f}{}{r}=\normtwosphere{\sn\sn f}{}{r}-\frac{3}{r^3}(r-6M)\normtwosphere
	{\sn f}{}{r}+\frac{4}{r^6}(r-3M)^2\normtwosphere{f}{}{r}.
	\end{align*}
	Successively applying Lemma \ref{lemmapoincare} therefore yields
	\begin{align*}
	\frac{4}{r^6}\Big((r-3M)^2+r(r+18M)\Big)\normtwosphere{f}{}{r}+\frac{1}{3}\frac{1}{r^3}(r+18M)\normtwosphere{\sn f}{}{r}+&\normtwosphere{\sn^2f}{}{r}\\
	\lesssim&\normtwosphere{\slap_{\mathfrak{Z}} f}{}{r}
	\end{align*}
	from which we conclude
	\begin{align*}
	\sum_{i=0}^2||(r\sn)^if||^2_{\twosphere_{\taus,r}}\lesssim \normtwosphere{r^2\slap_{\mathfrak{Z}} f}{}{r}.
	\end{align*}
	Standard elliptic theory then yields
	\begin{align*}
	\sum_{i=0}^{2p}||(r\sn)^if||^2_{\twosphere_{\taus,r}}\lesssim \normtwosphere{r^2\slap_{\mathfrak{Z}}^p f}{}{r}.
	\end{align*}
	The proposition then follows from the above estimate coupled with the fact that $\zslapinv{p}$ is a bijection on the space $\smfunrad$.
\end{proof}

Finally, we note the subsequent estimate which follows from the proof of Proposition \ref{propellipticestimateszlsap} and will prove useful in the sequel.
\begin{corollary}\label{correstimatesonszeta}
	Let $f\in\smfunrad$. Then on any 2-sphere $\geomtwosphere{}{r}$ one has the estimate
	\begin{align*}
	\frac{4}{r^2}(r-3M)^2\normtwosphere{\zslapinv{1}f}{}{r}+\frac{2}{r}(r+9M)\normtwosphere{(r\sn)\zslapinv{1}f}{}{r}\lesssim\normtwosphere{f}{}{r}.
	\end{align*}
\end{corollary}

\subsection{Commutation formulae and useful identities}\label{Commutationformulaeandusefulidentities}

In this final section we collect certain commutation relations and identities that will be used throughout the paper.

\begin{lemma}\label{lemmacommrelationsandidentities}
	Let $k, p\geq 1$ be integers. We denote by $\slashed{\mcala}^{[k]}$ any of the operators $\smcA{k}, \vmcA{k}$ or $\tmcA{k}$. Then on smooth functions we have the commutation relations
	\begin{align*}
	\big[\qn, \sn\big]&=0,\\
	\big[\qn, \slashed{\mcala}^{[2k]}\big]&=0,\\
	\big[\qn, \slashed{\mcala}^{[2k-1]}\big]&=\frac{\qexd r}{r}\slashed{\mcala}^{[2k-1]},\\
	\big[\qn, \zslapinv{p}\big]&=-3k\frac{\mu}{r}\,\qexd r\,\zslapinv{p-1}
	\end{align*}
	and
	\begin{align*}
	\big[\slap, \sdso\big]&=\frac{1}{r^2}\sdso,\\
	\big[\slap, \sn\astrosunhat\sdso\big]&=\frac{4}{r^2}\sn\astrosunhat\sdso.
	\end{align*}
	Moreover, we have the identities
	\begin{align*}
	\sdiv\sdso&=\slap,\\
	\sdiv\sn\astrosunhat\sdso&=\sdso\slap+\frac{2}{r^2}\sdso
	\end{align*}
	and on any 2-sphere $\twosphere_{\taus,r}$ 
	\begin{align*}
	\int_{\geomtwosphere{}{r}}\zslapinv{2p-1}f\cdot f\,r^{-2}\sepsilon&=-\normtwosphere{(r\sn)\zslapinv{p}f}{}{r}+(2-3\mu)\normtwosphere{\zslapinv{p}f}{}{r},\\
	\int_{\geomtwosphere{}{r}}\zslapinv{2p}f\cdot f\,r^{-2}\sepsilon&=\normtwosphere{\zslapinv{p}f}{}{r}
	\end{align*}
	with $f\in\smfunrad$.
\end{lemma}

\begin{proof}
	The first three commutation relations follow from the definitions of the operators in question, in particular noting the presence of the $r$-weights in the definitions of the $\slashed{\mcala}^{[k]}$. 
	
	For the fourth one we have
	\begin{align*}
	\big[\qn,r^2\slap_{\mathfrak{Z}}\big]=\frac{3\mu}{r}\qexd r\,\textnormal{Id}.
	\end{align*}
	and therefore the $p=1$ case follows from the formula
	\begin{align*}
	[\qn, \zslapinv{1}]=-\zslapinv{1}\big[\qn, r^2\slap_{\mathfrak{Z}}\big]\zslapinv{1}.
	\end{align*}
	For general $p$, one applies the induction formulae
	\begin{align*}
	\big[\qn, \zslapinv{n}\big]=\big[\qn, \zslapinv{n-1}\big]\zslapinv{1}+\zslapinv{n-1}\big[\qn, \zslapinv{1}\big].
	\end{align*}
	
For the final two we note the commutation relations
\begin{align*}
\big[\slap, \sn\big]f&=\frac{1}{r^2}f,\\
\big[\slap, \sn\big]\xi&=\frac{1}{r^2}\xi
\end{align*}
for a smooth function $f$ on $\mcalm$ and a smooth  1-form $\xi$.
	
	Turning now to the identities the first follows from the definition of $\sdso$ whereas for the second we note the identity 
	\begin{align*}
	\sdiv\sn\astrosunhat\xi=\slap\xi+\frac{1}{r^2}\xi
	\end{align*}
	on smooth $S$ 1-forms $\xi$.
	
	For the final two we perform an integration by parts on any 2-sphere $\twosphere_{\taus,r}$ to find
	\begin{align*}
	\inttwosphere{}{r}{\slap_{\mathfrak{Z}}f\cdot f}&=-\normtwosphere{\sn f}{}{r}+(2-3\mu)\normtwosphere{f}{}{r},\\
	\inttwosphere{}{r}{\slap_{\mathfrak{Z}}^2 f\cdot f}&=\inttwosphere{}{r}{\slap_{\mathfrak{Z}} f\cdot\slap_{\mathfrak{Z}} f}.
	\end{align*}
	This yields the $p=1$ case of the desired identities after recalling that $\zslapinv{1}$ is a bijection on the space $\smfunrad$. The remaining cases then follow by induction.
	\end{proof}

\section{The equations of linearised gravity around Schwarzschild}\label{TheequationsoflinearisedgravityaroundSchwarzschild}

In this section we derive the equations of interest in this paper, namely the system of equations that result from expressing the Einstein vacuum equations in a generalised wave gauge on $\mcalm$ and then linearising about $g_M$. The remainder of the paper then concerns the analysis of these ``equations of linearised gravity".

An outline of this section is as follows. We begin in section \ref{TheEinsteinvacuumequationsinageneralisedwavegauge} by first defining the generalised wave gauge and then presenting the Einstein vacuum equations on $\mcalm$ as they appear in this gauge. Finally in section \ref{Theformallinearisationoftheequationsofsection} we formally linearise the equations of section \ref{TheEinsteinvacuumequationsinageneralisedwavegauge} about $g_M$ to arrive at the equations of linearised gravity.

\subsection{The Einstein vacuum equations in a generalised wave gauge}\label{TheEinsteinvacuumequationsinageneralisedwavegauge}

In order to present the Einstein vacuum equations as they appear when expressed in a generalised wave gauge we must first define said gauge. This is the content of section \ref{Thegeneralisedwavegauge}. The generalised wave gauge actually defines a family of gauges and so the generalised wave gauge of interest in this paper is defined in section \ref{Thegeneralisedwavegaugewithrespecttothepair}. The Einstein vacuum equations as they appear in this particular generalised wave gauge are then presented in section \ref{TheEinsteinvacuumequationsasexpressedinageneralisedwavegaugewithrespecttothepair}.

\subsubsection{The generalised wave gauge}\label{Thegeneralisedwavegauge}

Let $\bg$ and $\bog$ be two smooth Lorentzian metrics on $\mcalm$ and let $\boldsymbol{f}:\smsymtwocov\rightarrow \smonecov$ be a smooth map. We define the connection tensor $C_{\bg, \bog}\in\smthreecov$ between $\bg$ and $\bog$ according to
\begin{align*}
(C_{\bg, \bog})_{abc}=\frac{1}{2}\Big(2\onabla_{(b}\bg_{c)a}-\onabla_a\bg_{bc}\Big)
\end{align*}
with $\onabla$ the Levi-Civita connection of $\bog$. The generalised wave gauge is then defined as follows.

\begin{definition}\label{defngenwavegauge}
We say that $\boldsymbol{g}$ is in a generalised wave gauge with respect to the pair $(\boldsymbol{\overline{g}},\boldsymbol{f})$ iff 
\begin{align}\label{wavegauge}
(\boldsymbol{g^{-1}})^{bc}(C_{\bg, \bog})_{abc}=\boldsymbol{f}(\bg)_a.
\end{align}
We will refer to the map $\boldsymbol{f}$ as a gauge-map.
\end{definition}

\subsubsection{The generalised wave gauge with respect to the pair $(g_M, \fgaumap)$}\label{Thegeneralisedwavegaugewithrespecttothepair}

The generalised wave gauge of interest in this paper is the one that is defined with respect to the pair $(g_M, \fgaumap)$ where the gauge-map $\fgaumap$ is to be defined below. 

First however we introduce two auxiliary maps that appear in the definition of $\fgaumap$ and which shall moreover appear again in section \ref{DecouplingtheequationsoflinearisedgravityuptoresidualpuregaugeandlinearisedKerrsolutions:theRegge--WheelerandZerilliequations}. In what follows we employ the machinery of section \ref{Ageometricfoliationby2-spheres} and the spherical harmonic projections of section \ref{Theprojectionofsmsymtwocovontoandawayfromthel=0,1sphericalharmonics}.\newline

We consider the $\reals$-linear maps $\Phi,\Psi:\smsymtwocov\rightarrow\smfunrad$ given by
\begin{align*}
\Phi(\tau):&=-r\slap_{1,0}^{-1}\bigg(\qhd\qexd\Big(r^{-2}\rpart{(\odd{\mtau})}-\qexd\big(r^{-2}\odd{\shattau}\big)\Big)\bigg),\\
\Psi(\tau):&=\slap^{-1}\Bigg[\rpart{\strtau}-\frac{4}{r}\Big(\rpart{(\even{\mtau})}-r^2\qexd\big(r^{-2}\even{\shattau}\big)\Big)_{\qP}-2\slap\even{\shattau}\\
&\qquad\qquad-\slap^{-1}_{\mathfrak{Z}}\Bigg(\qn_{\widetilde{P}}\bigg(\rpart{\strtau}-\frac{4}{r}\Big(\rpart{(\even{\mtau})}-r^2\qexd\big(r^{-2}\even{\shattau}\big)\Big)_{\qP}-2\slap\even{\shattau}\bigg)\\
&\qquad\qquad\qquad\quad-\frac{2}{r}\bigg(\rpart{\qhattau}-\qn{\astrosunhat}\Big(\rpart{(\even{\mtau})}-r^2\qexd\big(r^{-2}\even{\shattau}\big)\Big)\bigg)_{\widetilde{P}\widetilde{P}}\Bigg)\Bigg]
\end{align*}
where $\slap^{-1}_{1,0}$ and $\zslapinv{1}$ are defined as in section \ref{Thefamilyofoperatorsslap} and $\qP:=\big(\qexd r\big)^{\qsharp}$ -- note that the inverse operators are well defined by Proposition \ref{propinverseoperators}. Moreover we employ in the above and throughout the remainder of the paper the notation
\begin{align*}
\tau_{X_1...X_n}:=\tau(X_1,...,X_n)\text{ for }\tau\in\Gamma(S^nT^*\mcalm)\text{ and }X_1,...,X_n\in\smonecon.
\end{align*}

The desired gauge-map $\fgaumap$ is then defined as follows.

\begin{definition}\label{defngaugemap}
Let
\begin{align*}
\mfZ=\frac{24}{r^3}\frac{M}{r}(r-3M)\zslapinv{1}+\frac{72}{r^3}\frac{M}{r}\frac{M}{r}(r-2M)\slap^{-2}_{\mathfrak{Z}}.
\end{align*}
Then we define the $\reals$-linear map $\fgaumap:\smsymtwocov\rightarrow\smonecov$ according to
\begin{align*}
\fgaumap(\tau):=&\frac{2}{r}\qtau_{\qP}+\frac{2}{r}{\mtau}_{\qP}-\frac{1}{r}\qexd r\,\strtau-\frac{1}{r}\big(\qhattau_{\qP}\big)_{l=0,1}\\
+&\frac{2}{r}\big(1-\tfrac{4M}{r}\big)\sdso\big(\Psi(\tau), \Phi(\tau)\big)-\frac{1}{r^2}\qhd\qn\Big(r^3\mfZ\big(\Psi(\tau)\big)\Big)+r\sn\mfZ\big(\Psi(\tau)\big).
\end{align*}
\end{definition}
Note that $\fgaumap$ is well defined by Proposition \ref{propinverseoperators}.

Since $g_M$ is supported only on $l=0,1$ we have by explicit computation that $\fgaumap(g_M)=0$. We therefore have the following.
\begin{lemma}\label{lemmagaumap}
The Schwarzschild exterior metric $g_M$ is in a generalised wave gauge with respect to the pair $(g_M, \fgaumap)$.
\end{lemma}

\subsubsection{The Einstein vacuum equations as expressed in a generalised wave gauge with respect to the pair $(g_M, \fgaumap)$}\label{TheEinsteinvacuumequationsasexpressedinageneralisedwavegaugewithrespecttothepair}

We assume now that the Lorentzian metric $\bg$ of section \ref{Thegeneralisedwavegauge} is in a generalised wave gauge with respect to the pair $(g_M, \fgaumap)$. Assuming in addition that $\bg$ solves the Einstein vacuum equations,
\begin{align}\label{eineqn}
\ric[\bg]=0,
\end{align}
then \eqref{eineqn} reduces under the gauge condition \eqref{wavegauge} to the following system of equations on $\bg$:
\begin{align}
\big(\boldsymbol{g^{-1}}\big)^{cd}\nabla_c\nabla_d\bg_{ab}+&2(C_{\bg, g_M})^{c}_{de}\bg_{c(a}\nabla_{b)}\big(\boldsymbol{g^{-1}}\big)^{de}-4\bg_{de}(C_{\bg, g_M})_{db[a}\nabla_{c]}\big(\boldsymbol{g^{-1}}\big)^{cd}\nonumber\\
&-4(C_{\bg, g_M})^d_{b[a}(C_{\bg, g_M})^c_{c]d}+2\bg_{cd}\bg_{e(a}\tensor{\textnormal{Riem}}{_{b)}^{cd}_e}\nonumber\\
&=2\bg_{c(a}\nabla_{b)}\Big(\big(g_M^{-1}\big)^{ce}\fgaumap(\bg)_\epsilon\Big)\label{einsteinequationsinageneralisedwavegauge},\\
(\boldsymbol{g^{-1}})^{bc}(C_{\bg, g_M})_{abc}&=\fgaumap(\bg)_a\label{einsteinequationwavegauge}.
\end{align}
Here, $\textnormal{Riem}$ is the Riemann tensor of $g_M$ and $\big(C_{\bg, g_M}\big)^a_{bc}=\big(\boldsymbol{g^{-1}}\big)^{ad}\big(C_{\bg, g_M}\big)_{dbc}$.

Since $\ric[g_M]=0$ Lemma \ref{lemmagaumap} immediately yields the following.
\begin{lemma}\label{lemmaschwarzschildsolves}
	The Schwarzschild exterior metric $g_M$ is a solution to \eqref{einsteinequationsinageneralisedwavegauge}-\eqref{einsteinequationwavegauge}.
\end{lemma}

This paper is thus concerned with the system of equations that result from linearising the system  \eqref{einsteinequationsinageneralisedwavegauge}-\eqref{einsteinequationwavegauge} about the solution $g_M$.

\subsection{The formal linearisation of the equations of section \ref{TheEinsteinvacuumequationsinageneralisedwavegauge}: the equations of linearised gravity}\label{Theformallinearisationoftheequationsofsection}

In order to linearise the system of equations \eqref{einsteinequationsinageneralisedwavegauge}-\eqref{einsteinequationwavegauge} about the solution $g_M$ one must first develop a formal linearisation theory. This is the content of section \ref{Thelinearisationprocedure}. The linearised equations and the corresponding solution space are then presented in section \ref{Theequationsoflinearisedgravity}.

\subsubsection{The linearisation procedure}\label{Thelinearisationprocedure}

To formally linearise the system of equations \eqref{einsteinequationsinageneralisedwavegauge}-\eqref{einsteinequationwavegauge} about $g_M$ we assume the existence of a smooth 1-parameter family of smooth Lorentzian metrics $\bg(\epsilon)$ solving \eqref{einsteinequationsinageneralisedwavegauge}-\eqref{einsteinequationwavegauge} with $\bg(0)=g_M$.
We then define the linearised metric $\glin\in\smsymtwocov$ as the first order term in a formal power series expansion of $\bg(\epsilon)$ in powers of $\epsilon$ about $g_M$:
\begin{align}\label{powerseriesmetric}
\boldsymbol{g}(\epsilon)=g_M +\epsilon\, \glin+o(\epsilon^2).
\end{align}
Thus, in keeping with the notation of \cite{DHRlinstabschwarz}, quantities with a superscript ``(1)'' denote linear perturbations of bolded quantities about their background Schwarzschild value. 

Moreover, the $\reals$-linearity of the gauge-map $\fgaumap$ allows one to power series expand the 1-form $\fgaumap(\bg(\epsilon))$ according to
\begin{align}\label{powerseriesmap}
\fgaumap(\bg(\epsilon))=\epsilon\, \fgaumap(\glin)+o(\epsilon^2).
\end{align}
We thus define $\flin\in\smonecov$ as the first order term in the above expansion:
\begin{align*}
\flin:=\fgaumap(\glin).
\end{align*}
Consequently, to derive the linearisation of the system \eqref{einsteinequationsinageneralisedwavegauge}-\eqref{einsteinequationwavegauge} about $g_M$ one simply inserts the power series expansions \eqref{powerseriesmetric}-\eqref{powerseriesmap} into \eqref{einsteinequationsinageneralisedwavegauge}-\eqref{einsteinequationwavegauge} and discards higher than linear terms in $\epsilon$.

\subsubsection{The equations of linearised gravity}\label{Theequationsoflinearisedgravity}

Proceeding in this manner one arrives at the following \textbf{equations of linearised gravity} for the linearised metric $\glin$\footnote{In particular, note that $C_{g_M,g_M}=0$.}:
\begin{align}
\big(\Box_{\g} \lin{g}\big)_{ab}-2\tensor{\text{Riem}}{^c_{ab}^d}\lin{g}_{cd}&=2\nabla_{(a}\lin{f}_{b)},\label{eqnlinearisedeinsteinequations}\\
\big(\text{div}\glin\big)_a-\frac{1}{2}\nabla_a\trglin&=\lin{f}_{a}\label{eqnlorentzgauge}.
\end{align}
Here, $\Box$ and div are defined as in section \ref{Tensoranalysis}.

The equations of linearised gravity thus describe a coupled wave-divergence equation acting on smooth, symmetric 2-covariant tensor fields on $\mcalm$. This paper is concerned with establishing a \emph{decay} statement for this system of equations.

\section{Special solutions to the equations of linearised gravity}\label{Specialsolutionstotheequationsoflinearisedgravity}

In this section we introduce two special classes of solutions to the equations of linearised gravity. The decay statement of section \ref{Precisestatementsofthemaintheorems} that we establish for said equations involves solutions that have been normalised by the addition of particular members of each of these two classes. 

An outline of this section is as follows. We begin in section \ref{LinearisedKerrsolutionstotheequationsoflinearisedgravity} by introducing the first special class of \emph{linearised Kerr solutions} to the equations of linearised gravity. Finally in section \ref{Residualpuregaugesolutionstotheequationsoflinearisedgravity} we introduce the second special class of \emph{residual pure gauge solutions} to the equations of linearised gravity.

\subsection{Linearised Kerr solutions to the equations of linearised gravity}\label{LinearisedKerrsolutionstotheequationsoflinearisedgravity}

That the following do indeed solve the equations of linearised gravity can be verified by explicit computation. See also section \ref{OVSpecialsolutionstothelinearisedequations} for a more geometric derivation.

\begin{proposition}\label{proplinkerr}
	Let $\mfm, \mfa_{-1}, \mfa_0, \mfa_{1}\in\reals$. Then for $i=-1,0,1$ the following is a smooth solutions to the equations of linearised gravity:
	\begin{align}\label{eqnlinkerr}
	\glin_{\mfm, \mfa_{i}}:=\frac{4\mfm}{ x}\exd t^*\exd x+8M\mfm\bigg(1+\frac{1}{x}\bigg)\exd x^2-2\mfa_i\bigg[\frac{1}{x}\exd t^*+2M\bigg(1+\frac{1}{x}\bigg)\exd x\bigg]\astrosun\shd\sexd\slashed{Y}^1_i+8M\mfm x^2\pi_{\twosphere}^*\roundmetric.
	\end{align}
	
\end{proposition}

We will call such solutions linearised Kerr solutions to the equations of linearised gravity. In particular, observe that these solutions are parametrised by four real numbers and are moreover supported only on $l=0,1$.

\subsection{Residual pure gauge solutions to the equations of linearised gravity}\label{Residualpuregaugesolutionstotheequationsoflinearisedgravity}

That the following do indeed solve the equations of linearised gravity can be verified by explicit computation. See also section \ref{OVSpecialsolutionstothelinearisedequations} for a more geometric derivation.
\begin{proposition}\label{proppuregauge}
Let $V\in\smonecov$ solve
\begin{align}
\Box_{g_M}V=\fgaumap(\nabla\astrosun V).\label{eqnpuregauge}
\end{align}
Then the following is a smooth solution to the equations of linearised gravity:
\begin{align*}
\glin_V:=\nabla\astrosun V.
\end{align*} 
\end{proposition}

We will call such a solution a residual pure gauge solution to the equations of linearised gravity. In particular, these solutions are parametrised by smooth 1-forms on $\mcalm$ solving the equation \eqref{eqnpuregauge} which we henceforth refer to as the \emph{residual pure gauge equation}. Consequently, we will show that such 1-forms exist in section \ref{TheCauchyinitialvalueproblemfortheresidualpuregaugeequation} and thus there exist non-trivial residual pure gauge solution to the equations of linearised gravity.

\begin{remark}\label{rmkpuregaugedistinctfromkerr}
It is clear from their derivation in section \ref{OVSpecialsolutionstothelinearisedequations} that linearised Kerr and residual pure gauge solutions comprise two distinct classes of solutions to the equations of linearised gravity. This can also be shown purely in the context of the linear theory -- see \cite{SarbachPHD}.
\end{remark}

\section{Decoupling the equations of linearised gravity up to residual pure gauge and linearised Kerr solutions: the Regge--Wheeler and Zerilli equations}\label{DecouplingtheequationsoflinearisedgravityuptoresidualpuregaugeandlinearisedKerrsolutions:theRegge--WheelerandZerilliequations}

In this section we show that any smooth solution to the equations of linearised gravity can be decomposed as the sum of a residual pure gauge solution, a linearised Kerr solution and a solution to the equations of linearised gravity that effectively decouples into the Regge--Wheeler and Zerilli equations respectively. The decay statement we establish for the Regge--Wheeler and Zerilli equations in section \ref{Theorem1:BoundednessanddecayforsolutionstotheRegge--WheelerandZerilliequations} will then ultimately yield the decay statement of section \ref{Theorem 2:Boundedness,decayandasymptoticflatnessofinitial-data-normalisedsolutionstotheequationsoflinearisedgravity} for those solutions to the equations of linearised gravity that have been normalised with respect to this decomposition.

An outline of this section is as follows. We begin in section \ref{TheRegge--WheelerandZerilliequationsandtheirconnectiontotheequationsoflinearisedgravity} by revealing the connection between the Regge--Wheeler and Zerilli equations and the equations of linearised gravity, in particular showing that said equations actually generate a class of solutions to the linearised system. Then in section \ref{Extractingaresidualpuregaugesolutionfromageneralsolutiontotheequationsoflinearisedgravity} we show that one can always extract a residual pure gauge solution from any given smooth solution to the equations of linearised gravity. Finally in section \ref{DecomposingageneralsolutiontotheequationsoflinearisedgravityintothesumofaresidualpuregaugeandlinearisedKerrsolutionandasolutiondeterminedbytheRegge--WheelerandZerilliequations} we combine the insights of the previous two sections to construct the desired decoupling of the equations of linearised gravity.

\subsection{The Regge--Wheeler and Zerilli equations and their connection to the equations of linearised gravity}\label{TheRegge--WheelerandZerilliequationsandtheirconnectiontotheequationsoflinearisedgravity}

Since the decay statement we aim to establish for the equations of linearised gravity is sensitive to the addition of both residual pure gauge and linearised Kerr solutions it is natural as a first step to establish a decay statement for those linearised quantities which vanish for all such solutions. It is moreover natural to try and isolate scalar versions of these \emph{invariant} quantities as a means of mitigating potential complications arising from the tensorial structure of the equations of linearised gravity. Consequently, we will show in section \ref{Theconnectionwiththeequationsoflinearisedgravity} that there exist two invariant scalars satisfying the property that they actually decouple from the full linearised system into the wave equations described by the Regge--Wheeler and Zerilli equations respectively. The key point is that solutions to these equations can be analysed using the techniques developed for establishing a decay statement for the scalar wave equation on $\Mg$ as we shall prove in section \ref{Proofoftheorem1}. We then further demonstrate in section \ref{SolutionstotheequationsoflinearisedgravitygeneratedbysolutionstotheRegge--WheelerandZerilliequations} that solutions to the Regge--Wheeler and Zerilli equations actually generate a class of solutions to the equations of linearised gravity. It then follows that such a class of solutions should be susceptible to a decay statement -- we shall exploit this fact later. First however let us define the Regge--Wheeler and Zerilli equations.

\subsubsection{The Regge--Wheeler and Zerilli equations}\label{TheRegge--WheelerandZerilliequations}

The Regge--Wheeler and Zerilli equations describe two scalar \emph{wave} equations on $\mcalm$. They are defined as follows.
\begin{definition}\label{defnRWandZ}
	Let $\psi\in\smfunrad$. 
	
	Then we say that $\psi$ is a smooth solution to the Regge--Wheeler equation on $\mcalm$ iff
	\begin{align}\label{eqnRWeqn}
	\qbox\psi+\slap\psi=-\frac{6}{r^2}\frac{M}{r}\psi.
	\end{align}
	Conversely, we say that $\psi$ is a smooth solution to the Zerilli equation on $\mcalm$ iff
	\begin{align}\label{eqnZereqn}
	\qbox\psi+\slap\psi=-\frac{6}{r^2}\frac{M}{r}\psi+\frac{24}{r^3}\frac{M}{r}(r-3M){\slap^{-1}_{\mathfrak{Z}}}\psi+\frac{72}{r^3}\frac{M}{r}\frac{M}{r}(r-2M)\zslapinv{2}\psi.
	\end{align}
\end{definition}
Note that one can write the Zerilli equation as 
\begin{align*}
\qbox\psi+\slap\psi=-\frac{6}{r^2}\frac{M}{r}\psi+\mfZ\psi
\end{align*}
where the operator $\mfZ$ is defined as in section \ref{Thegeneralisedwavegaugewithrespecttothepair}.

To see that the above do indeed define wave equations on $\mcalm$ it suffices to note from section \ref{Decomposingtheequationsoflinearisedgravity} that
\begin{align*}
r\Box_{g_M}(r^{-1}f)=r\qbox f+\slap f, \qquad f\in\smfun.
\end{align*}

\subsubsection{The connection with the equations of linearised gravity}\label{Theconnectionwiththeequationsoflinearisedgravity}

We now show that certain linearised scalar quantities which vanish for residual pure gauge and linearised Kerr solutions decouple from the equations of linearised gravity into the Regge--Wheeler and Zerilli equations respectively. 

We recall for this section the $\reals$-linear maps $\Phi$ and $\Psi$ defined in section \ref{Thegeneralisedwavegaugewithrespecttothepair}.\newline

An analogue of the theorem stated below for smooth solutions to the \emph{linearised Einstein equations around Schwarzschild} was originally established by Regge--Wheeler and Zerilli in \cite{RWrweqn} and \cite{Zzeqn}. Since however solutions to this system of equations differ from solutions to the equations of linearised gravity only by the addition of a \emph{pure gauge solution} to the linearised equations (see for instance the book of Wald \cite{Waldbook}) one can immediately infer from \cite{RWrweqn} and \cite{Zzeqn} the following theorem. We shall however reprove it for reasons of completeness. Note then that we follow \cite{COSschwarz} very closely in the proof.

\begin{theorem}\label{thmgaugeinvariantquantintermsofRWandZ}
	Let $\glin$ be a smooth solution to the equations of linearised gravity. 
	
	Then the scalar $\Philin\in\smfunrad$ defined according to
	\begin{align*}
	\Philin:=\Phi\big(\glin\big)
	\end{align*}
	satisfies the Regge--Wheeler equation \eqref{eqnRWeqn}. Moreover $\Philin$ vanishes if $\glin$ is either a residual pure gauge or linearised Kerr solution to the equations of linearised gravity.
	
	Conversely the scalar $\Psilin\in\smfunrad$ defined according to
	\begin{align*}
	\Psilin:=\Psi\big(\glin\big)
	\end{align*}
	satisfies the Zerilli equation \eqref{eqnZereqn}. Moreover $\Psilin$ vanishes if $\glin$ is either a residual pure gauge or linearised Kerr solution to the equations of linearised gravity.
\end{theorem}

\begin{proof}
Given the solution $\glin$ we first construct the quantity $\iotalin\in\smsymtwocov$ defined by the projections
\begin{align*}
\qiotalin:&=\rpart{\qglin}-\qn\astrosun\even{\miotalin},\\
{\miotalin}:&=\rpart{\mglin}-r^2\sdso\astrosun\bigg(\qexd\Big(r^{-2}\even{\shatglin}\Big),\qexd\Big(r^{-2}\odd{\shatglin}\Big)\bigg),\\
\siotalin:&=\lin{\slashed{g}}-\sg\Big(\slap\even{\shatglin}+\tfrac{2}{r}\big(\even{\miotalin}\big)_{\qP}\Big).
\end{align*}
each of which manifestly have vanishing projection to $l=0,1$. Then from the decomposed equations of linearised gravity presented in Corollary \ref{corrdecomposedeqnslingrag} combined with the fact that the operators $\sdso$ and $\sn\astrosunhat\sdso$ have kernels spanned by the $l=0,1$ spherical harmonics (cf. Propositions \ref{prophodgedecomp} and \ref{prophodgedecompqmsm1form}) we find
\begin{align}
\qbox\qhatiotalin+\slap\qhatiotalin+\frac{2}{r}\qn_{\qP}\qhatiotalin-\frac{2}{r}\qn\astrosunhat\Big(\qhatiotalin_{\qP}\Big)-\frac{2}{r}\frac{\mu}{r}\,\qhatiotalin&=\frac{1}{r}\qexd r\qastrosunhat\qexd\qtriotalin-\frac{1}{r}\qexd r\qastrosunhat\qexd \striotalin,\label{eqnwaveeqnfortauhat}\\
\qbox\qtriotalin+\slap\qtriotalin&=0,\\
\qbox\odd{\miotalin}+\slap\odd{\miotalin}-\frac{2}{r}\qn_{\qP}\odd{\miotalin}+\frac{2}{r}\qn\Big(\big(\odd{\miotalin}\big)_{\qP}\Big)-\frac{1}{r}\frac{\mu}{r}\odd{\miotalin}+\frac{2}{r^2}\qexd r\,\big(\odd{\miotalin}\big)_{\qP}&=0,\label{eqnwaveeqnforeta}\\
\qbox \striotalin+\slap \striotalin+\frac{2}{r}\qn_{\qP}\striotalin+\frac{2}{r^2}\striotalin&=\frac{4}{r^2}\qhatiotalin_{\qP\qP}\label{eqnwaveeqnforsigma}
\end{align}
and
\begin{align}
\qdiv\qhatiotalin+\frac{1}{2}\qn \striotalin&=0,\label{eqndivtauhat}\\
\qtriotalin&=0\label{eqnfortrtau},\\
\qdiv\odd{\miotalin}&=0\label{eqnforeta}.
\end{align}
Here $\mu=\tfrac{2M}{r}$. Introducing now the quantity $\zetalin:=\qhatiotalin_{\qP}-\frac{r}{2}\,\qexd\striotalin$ we re-express the above system according to
\begin{align}
\qbox\odd{\miotalin}+\slap\odd{\miotalin}-\frac{2}{r}\qn_{\qP}\odd{\miotalin}+\frac{2}{r}\qexd\Big(\big(\odd{\miotalin}\big)_{\qP}\Big)-\frac{1}{r}\frac{\mu}{r}\odd{\miotalin}+\frac{2}{r^2}\qexd r\,\big(\odd{\miotalin}\big)_{\qP}&=0,\label{eqnwaveeqneta}\\
\qdiv\odd{\miotalin}&=0\label{eqndivrelationeta}
\end{align}
and
\begin{align}
-\frac{1}{r^2}\qexd\bigg(r^2\qdiv \zetalin+3M\striotalin\bigg)+\slap \zetalin&=0,\label{waveeqnzeta}\\
\frac{1}{r^2}\qexd\bigg(4r \zetalin_{\qP}-r^3\slap_{\mathfrak{Z}}\striotalin\bigg)+2\slap \zetalin&=0,\label{eqnzetasigma}\\
-2\qdiv \zetalin+r\qbox \striotalin&=0,\label{eqndivrelationzeta}\\
\qexd \zetalin&=0\label{eqnzetaclosed}.
\end{align}
Here, one arrives at the relations \eqref{eqndivrelationzeta} and \eqref{eqnzetaclosed} by contracting \eqref{eqndivtauhat} with $\qP$ and $\big(\qhd\qexd r\big)^{\qsharp}$ respectively whereas the equations \eqref{waveeqnzeta} and \eqref{eqnzetasigma} follow from contracting \eqref{eqnwaveeqnfortauhat} with $\qP$ and utilising  \eqref{eqnwaveeqnforsigma} in conjunction with \eqref{eqndivrelationzeta}.

Now, arguing as in the Poincar\'e lemma for the simply connected manifold $\mathbb{R}\times[1, \infty)$, it follows from relations \eqref{eqndivrelationeta} and \eqref{eqnzetaclosed} that there exists two unique functions $\philin,\varphilin\in\smfunrad$ such that
\begin{align}\label{eqnetaintermsofphi}
\odd{\miotalin}=-\qhd\qexd\big(r\philin+\pi_{\twosphere}^*f\big)
\end{align}
and
\begin{align}\label{eqnzetaintermsofvarphi}
\zetalin=\qexd \big(\varphilin+\pi_{\twosphere}^*g\big).
\end{align}
Here $f$ and $g$ are smooth functions on $\twosphere$ which are supported on the $l\geq 2$ spherical harmonics and we recall the projection map $\pi_{\twosphere}:\mcalm\rightarrow\twosphere$. Integrating the equations \eqref{eqnwaveeqneta} and \eqref{waveeqnzeta}-\eqref{eqnzetasigma} therefore yields
\begin{align*}
\qbox\philin+\slap\philin=-\frac{6}{r^2}\frac{M}{r}\philin
\end{align*}
and
\begin{align*}
\slap_{\mathfrak{Z}}\qbox \varphilin+\slap\slap\varphilin+\frac{2}{r^2}\slap\varphilin-\frac{6\mu}{r^3}\qn_{P}\varphilin=0
\end{align*}
where we have used the functions $f$ and $g$ to remove the constants of integration.
Thus, recalling that the operator $\zslapinv{1}$ is a bijection on the space $\smfunrad$ (cf. Proposition \ref{propinverseoperators}), the unique functions $\philin, \psilin\in\smfunrad$ with $\psilin$ defined according
\begin{align}\label{mum}
\zslapinv{1}\varphilin:=\psilin
\end{align}
satisfy the Regge--Wheeler and Zerilli equations of section \ref{TheRegge--WheelerandZerilliequations}:
\begin{align*}
\qbox\philin+\slap\philin&=-\frac{6}{r^2}\frac{M}{r}\philin,\\
\qbox\psilin+\slap\psilin&=-\frac{6}{r^2}\frac{M}{r}\psilin+\frac{24}{r^3}\frac{M}{r}(r-3M)\zslapinv{1}\psilin+\frac{72}{r^5}\frac{M}{r}\frac{M}{r}(r-2M)\zslapinv{2}\psilin.
\end{align*}
Here we use the commutation formulae of Lemma \ref{lemmacommrelationsandidentities}. We claim that $\philin=\Philin$ and $\psilin=\Psilin$. Indeed we have from the relations \eqref{eqnetaintermsofphi} and \eqref{eqnzetasigma} the identities
\begin{align}\label{ident}
\slap_{1,0}\philin=-r\qhd\qexd\Big(r^{-2}\odd{\miotalin}\Big),\qquad \slap\psilin=\frac{2}{r}\zslapinv{1}\qhatiotalin_{{\qP\qP}}-\frac{1}{2}\frac{1}{r}\striotalin-\zslapinv{1}\qn_{\qP}\striotalin
\end{align}
and thus since $\slap_{1,0}$ and $\slap$ are bijections over $\smfunrad$ (cf. Proposition \ref{propinverseoperators}) we must have $\philin=\Philin$ and $\psilin=\Psilin$ (recalling the definition of $\Philin$ and $\Psilin$ from section \ref{Thegeneralisedwavegaugewithrespecttothepair}).

It remains to show that $\Philin$ and $\Psilin$ vanish if constructed from either a residual pure gauge or linearised Kerr solution. Subsequently, in view of the bijective properties of the operators $\slap_{1,0}$ and $\slap$ and identities \eqref{ident} it suffices to show that $\qhatiotalin, \odd{\miotalin}$ and $\striotalin$ vanish if constructed from either a residual pure gauge or linearised Kerr solution. This is immediate however from Lemma \ref{lemmadecomposingderivatives1form} and the fact that $\qhatiotalin, \odd{\miotalin}$ and $\striotalin$ have vanishing projection to $l=0,1$.
\end{proof}

\subsubsection{Solutions to the equations of linearised gravity generated by solutions to the Regge--Wheeler and Zerilli equations}\label{SolutionstotheequationsoflinearisedgravitygeneratedbysolutionstotheRegge--WheelerandZerilliequations}

We end this section by showing that the Regge--Wheeler and Zerilli equations actually generate solutions to the equations of linearised gravity.\newline

We begin by noting the following corollary of Theorem \ref{thmgaugeinvariantquantintermsofRWandZ}.
\begin{corollary}\label{corrRWZ}
Let $\glin$ be a smooth solution to the equations of linearised gravity. Then the following identities hold:
\begin{align*}
\qiotalin&=\qn\astrosunhat\qexd\Big(r\Psilin\Big)+6\mu\qexd r\qastrosunhat\zslapinv{1}\qexd\Psilin,\\
\odd{\miotalin}&=-\qhd\qexd\Big(r\Philin\Big),\\
\striotalin&=-2r\slap\Psilin+4\qn_{\qP}\Psilin+12\mu r^{-1}(1-\mu)\zslapinv{1}\Psilin.
\end{align*}
\end{corollary}
\begin{proof}
The second is just \eqref{eqnetaintermsofphi} whereas the first and third follow from equation \eqref{eqnzetasigma} combined with the relations \eqref{eqnzetaintermsofvarphi}-\eqref{mum}.
\end{proof}

We shall use this corollary to demonstrate the fact that the Regge--Wheeler and Zerilli equations actually generate solutions to the equations of linearised gravity. In order to state the result succinctly however it is useful to first introduce the map $\gamma:\smfunrad\times\smfunrad\rightarrow\smsymtwocov$ defined according
\begin{align*}
\gamma(f,g)=\qgamma(f,g)+\mgamma(f,g)+\sgamma(f,g)
\end{align*}
where
\begin{align*}
\gamma(f,g)&=\qn\astrosunhat\qexdL\big(rf\big)+6\mu\qexd r\qastrosunhat\zslapinv{1}\qexd f,\\
{\mgamma}(f,g)&=\sdso\astrosun\Big(\qexdL\big(rf\big)-2\qexd r\,f, \qexdL\big(rg\big)-2\qexd r\,g\Big),\\
\sgamma(f,g)&=r\sn\astrosunhat\sdso\Big(f, g\Big)+2\sg\Big(\qexdL_{\qP}f+3\mu r^{-1}(1-\mu)\zslapinv{1}g\Big).
\end{align*}
Here we have in addition defined the operator
\begin{align*}
\qexdL:=\qexd-\qhd\qexd.
\end{align*}

\begin{proposition}\label{propsolnlingravRWZ}
Let $\Phi\in\smfunrad$ and $\Psi\in\smfunrad$ be solutions to the Regge--Wheeler equation \eqref{eqnRWeqn} and Zerilli equation \eqref{eqnZereqn} respectively. Then the following is a smooth solution to the equations of linearised gravity:
\begin{align*}
\gammalin:=\gamma(\Psi, \Phi).
\end{align*}
Moreover, $\gammalin$ does not lie in the space of residual pure gauge or linearised Kerr solutions to the equations of linearised gravity.
\end{proposition}
\begin{proof}
Let $\iota\in\smsymtwocov$ be defined according to
\begin{align*}
\iota:=\qn\astrosunhat\qexd\Big(r\Psi\Big)+6\mu\qexd r\qastrosunhat\zslapinv{1}\qexd\Psi-\sdso\astrosun\Big(0,\qhd\qexd\big(r\Phi\big)\Big)-\sg\Big(r\slap\Psi-2\qn_{\qP}\Psi-6\mu r^{-1}(1-\mu)\zslapinv{1}\Psi\Big).
\end{align*}
Then it follows from Corollary \ref{corrRWZ} and Theorem \ref{thmgaugeinvariantquantintermsofRWandZ} that $\iota$ satisfies the system \eqref{eqnwaveeqnfortauhat}-\eqref{eqnforeta}. Since moreover $\gammalin-\iota$ is equal to the expression
\begin{align*}
-\qn\astrosunhat\qhd\qexd\big(r\Psi\big)+\sdso\astrosun\Big(\qexdL\big(r\Psi\big)-2\qexd r\,\Psi, \qexd\big(r\Phi\big)-2\qexd r\,\Phi\Big)+r\sn\astrosunhat\sdso\Big(\Psi, \Phi\Big)+\sg\Big(r\slap\Psi-2\qhd\qexd_{\qP}\Psi\Big)
\end{align*}
then the decomposed equations of linearised gravity of Corollary \ref{corrdecomposedeqnslingrag} combined with the fact that $\Phi$ and $\Psi$ satisfy the Regge--Wheeler and Zerilli equation respectively yields the proposition after performing the required computation.
\end{proof}

Note we will need the following corollary of Proposition \ref{propsolnlingravRWZ} when it comes to proving the well-posedness theorem of section \ref{Thewell-posednesstheorem}.

\begin{corollary}\label{corrsolnlingravRWZ}
If $\gammalin=\gamma(\Psi, \Phi)=\glin_V+\glin_{\mfm, \mfa}$ for any residual pure gauge solution $\glin_V$ or linearised Kerr solution $\glin_{\mfm, \mfa}$ then $\Phi=\Psi=V=\mfm=\mfa=0$.
\end{corollary}
\begin{proof}
It follows by construction that the invariant quantities $\Philin$ and $\Psilin$ associated to $\gammalin$ are given by $\Phi$ and $\Psi$. Since the invariant quantities $\Philin$ and $\Psilin$ associated to $\glin_V$ and $\glin_{\mfm, \mfa}$ must necessarily vanish the result follows.
\end{proof}

\subsection{Extracting a residual pure gauge solution from a general solution to the equations of linearised gravity}\label{Extractingaresidualpuregaugesolutionfromageneralsolutiontotheequationsoflinearisedgravity}

We now extend the analysis of the previous section to show that one can actually extract from a smooth solution to the equations of linearised gravity a 1-form that satisfies the residual pure gauge equation. Importantly, this 1-form is \emph{not} gauge-invariant. Since in section \ref{DecomposingageneralsolutiontotheequationsoflinearisedgravityintothesumofaresidualpuregaugeandlinearisedKerrsolutionandasolutiondeterminedbytheRegge--WheelerandZerilliequations} we will show that one can decompose a general solution to the equations of linearised gravity into the sum of a solution of the class identified in section \ref{SolutionstotheequationsoflinearisedgravitygeneratedbysolutionstotheRegge--WheelerandZerilliequations} with a linearised Kerr solution and a residual pure gauge solution that is generated by this 1-form it follows that one can always normalise solutions to the equations of linearised gravity by the special solutions of section \ref{Specialsolutionstotheequationsoflinearisedgravity} into a solution the decay properties of which is determined by the decay properties of solutions to the Regge--Wheeler and Zerilli equations. We will make this analysis precise in section \ref{GaugeandlinearisedKerrnormalisedsolutionstotheequationsoflinearisedgravity}. First however we extract the desired 1-form.\newline

The main result of this section is as follows.

\begin{proposition}\label{propextracingrespurgau}
Let $\glin$ be a smooth solution to the equations of linearised gravity. Define the quantity
\begin{align*}
\Vlin:=\rpart{\big(\even{\mglin}\big)}-r^2\qexd\Big(r^{-2}\even{\shatglin}\Big)+\sdso\Big(\even{\shatglin}, \odd{\shatglin}\Big)-\qhd\qexd\Big(r\Psilin\Big)-\sdso\Big(r\Psilin, r\Philin\Big).
\end{align*}
Then $\Vlin$ satisfies the residual pure gauge equation:
\begin{align*}
\Box_{g_M}\Vlin=\fgaumap\big(\nabla\astrosun\Vlin\big).
\end{align*}
Moreover, if $\glin=\glin_V$ is a residual pure gauge solution to the equations of linearised gravity then $\Vlin=\rpart{V}$.
\end{proposition}
\begin{proof}
The latter statement follows from section \ref{Decomposingtheequationsoflinearisedgravity}.

For the former we compute from the decomposed equations of linearised gravity presented in Corollary \ref{corrdecomposedeqnslingrag} the system
\begin{align}
\qbox\qplin+\slap\qplin-\frac{2}{r}\qn\big(\qplin_{\qP}\big)+\frac{2}{r}\qn r\,\sdiv\slplin+\frac{2}{r^2}\qn r\,\qplin_{\qP}&=\frac{1}{r^2}\qhd\qn\Big(r^3\mfZ\Psilin\Big),\label{reducedmaxwell1}\\
\qbox\slplin+\slap\slplin+\frac{1}{r^2}\ommu\slplin&=-\frac{2}{r}(1-2\mu)\sdso\big(\Psilin, \Philin\big)+r\sn\mfZ\Psilin.\label{reducedmaxwell2}
\end{align}
where we have defined $\qplin:=\rpart{\big(\even{\mglin}\big)}-r^2\qexd\Big(r^{-2}\even{\shatglin}\Big)$ and $\slplin:=\sdso\Big(\even{\shatglin}, \odd{\shatglin}\Big)$. Using then Corollary \ref{corrdecomposedgaugeequation} combined with the fact that $\Philin$ and $\Psilin$ respectively satisfy the Regge--Wheeler and Zerilli equations by Theorem \ref{thmgaugeinvariantquantintermsofRWandZ}, the proposition follows.
\end{proof}
We then have the immediate corollary of the above relating to extracting residual pure gauge solutions from general smooth solutions to the equations of linearised gravity.

\begin{corollary}\label{corrpuregauge}
Let $\glin$ be a smooth solution to the equations of linearised gravity. Then $\glin_{\Vlin}=\nabla\astrosun\Vlin$ defines a residual pure gauge solution to the equations of linearised gravity.
\end{corollary}

\subsection{Decomposing a general solution to the equations of linearised gravity into the sum of a residual pure gauge and linearised Kerr solution and a solution determined by the Regge--Wheeler and Zerilli equations}\label{DecomposingageneralsolutiontotheequationsoflinearisedgravityintothesumofaresidualpuregaugeandlinearisedKerrsolutionandasolutiondeterminedbytheRegge--WheelerandZerilliequations}

We now finally present the decomposition of the equations of linearised gravity that was hinted at previously. We will exploit this decomposition in section \ref{Proofoftheorem2} to establish a decay statement for the equations of linearised gravity that takes into account the special solutions of section \ref{Specialsolutionstotheequationsoflinearisedgravity}.\newline

The main result of this section is as follows.
\begin{theorem}\label{thmdecoupling}
Let $\glin$ be a smooth solution to the equations of linearised gravity. Then there exists a quadruple $\mfm, \mfa_{-1},\mfa_0,\mfa_{1}\in\reals$ and a $V\in\smonecov$ that is supported only on $l=0,1$ such that $\glin$ can be decomposed as
\begin{align*}
\glin=\gamma\big(\Psilin,\Philin\big)+\glin_{\Vlin}+\glin_V+\sum_{i=-1,0,1}\glin_{\mfm,\mfa_{i}}.
\end{align*}
\end{theorem}
\begin{proof}
By proposition \ref{propsphericalharmonicdecomptensors} we have the decomposition
\begin{align*}
\glin=\rpart{\glin}+\glin_{l=0,1}
\end{align*}
where $\rpart{\glin}$ has vanishing projection to $l=0,1$ and $\glin_{l=0,1}$ is supported only on $l=0,1$. Subsequently, that we can decompose $\glin_{l=0,1}$ as 
\begin{align*}
\glin_{l=0,1}=\glin_V+\sum_{i=-1,0,1}\glin_{\mfm,\mfa_{i}}
\end{align*}
for $\mfm, \mfa_{-1},\mfa_0,\mfa_{1}\in\reals$ and $V\in\smonecov$ follows from classical work on the linearised Einstein equations around Schwarzschild which we shall not verify directly here -- see for instance \cite{SarbachPHD}.

It thus remains to show that 
\begin{align*}
\rpart{\glin}=\gamma\big(\Psilin,\Philin\big)+\nabla\astrosun\Vlin.
\end{align*}
Indeed, we observe that since $\Philin, \Psilin$ and $\Vlin$ have vanishing projection to $l=0,1$ we thus have
\begin{align*}
\rpart{\glin}&=\iotalin+\nabla\astrosun\bigg(\Vlin+\qhd\qexd\Big(r\Psilin\Big)+\sdso\Big(r\Psilin, r\Philin\Big)\bigg),\\
&=\gamma\big(\Psilin,\Philin\big)+\nabla\astrosun\Vlin
\end{align*}
where the last line follows by definition of the map $\gamma$. This yields the theorem.
\end{proof}

\section{Initial data and well-posedness for the equations of linearised gravity}\label{Initialdatandwellposednessfortheequationsoflinearisedgravity}

In this section we establish a well-posedness theorem for the equations of linearised gravity. This in particular shows that the space of solutions to the equations of linearised gravity is non-empty.

An outline of this section is as follows. We begin in section \ref{Seeddatafortheequationsoflinearisedgravity} by defining a notion of seed data for the equations of linearised gravity consisting of freely prescribed quantities on the initial hypersurface $\Sigma_0$. Then in section \ref{Thewell-posednesstheorem} we state and prove the well-posedness theorem which establishes a surjection between smooth solutions to the equations of linearised gravity and smooth seed data which is moreover injective over a subclass of suitably regular seed data. Note this latter notion of regularity shall in fact provide the correct notion of regularity required for our decay statement of section \ref{Theorem 2:Boundedness,decayandasymptoticflatnessofinitial-data-normalisedsolutionstotheequationsoflinearisedgravity} to hold.

\subsection{Seed data for the equations of linearised gravity}\label{Seeddatafortheequationsoflinearisedgravity}

It is a long but tedious computation (see for instance \cite{CBbook}) to show that Cauchy data for solutions to the equations of linearised gravity cannot be prescribed freely but must satisfy certain constraints. It is therefore more appropriate to interpret the Cauchy problem for the equations of linearised gravity on the initial Cauchy hypersurface $\Sigma$ as the problem of constructing solutions to the linearised system from \emph{freely prescribed seed data} on $\Sigma$. In section \ref{Seeddatafortheequationsoflinearisedgravity2} we provide a notion of such seed data. In particular, we will show in section \ref{Thewell-posednesstheorem} that all smooth solutions to the equations of linearised gravity arise from this smooth seed data. We in addition provide in section \ref{Pointwiseasymptoticallyflatseeddata} a stronger notion of regularity on this seed for which the solution map will be an isomorphism onto its image.

\subsubsection{Seed data for the equations of linearised gravity}\label{Seeddatafortheequationsoflinearisedgravity2}

Smooth seed data for the equations of linearised gravity is defined as follows.
\begin{definition}\label{defnseeddata}
	A smooth seed data set for the equations of linearised gravity consists of prescribing:
	\begin{itemize}
		\item four functions $\Phi_0, \Phi_1, \Psi_0, \Psi_1\in\smfunsigrad$
		\item two functions $\uV_0, \uV_1\in\smfunsig$
		\item two 1-forms $\oV_0, \oV_1\in\smonesig$
		\item four constants $\mfm, \mfa_{-1},\mfa_{0},\mfa_{1}\in\reals$
	\end{itemize}
\end{definition}
Here $\smonesig$ denotes the space of smooth 1-forms on $\Sigma_{0}$, $\smfunsig$ the space of smooth functions on $\Sigma_0$ and $\smfunsigrad$ the space of smooth functions on $\Sigma_{0}$ supported on $l\geq 2$ with this notion defined analogously as in section \ref{Theprojectionofsmfunontoandawayfromthel=0,1sphericalharmonics}.

\subsubsection{Pointwise asymptotically flat seed data}\label{Pointwiseasymptoticallyflatseeddata}

Now we provide a stronger notion of regularity on the seed data which we shall need for the decay statement of Theorem 2 in section \ref{Theorem 2:Boundedness,decayandasymptoticflatnessofinitial-data-normalisedsolutionstotheequationsoflinearisedgravity} in addition to the well-posedness theorem of section \ref{Thewell-posednesstheorem}.\newline

To define this notion of regularity we first introduce the following pointwise norm acting on $n$-covariant tensor fields on $\Sigma$:
\begin{align*}
\big|h\big|_{\overline{g}_M}:=\big|(\overline{g}^{-1}_M)^{i_1j_1}...(\overline{g}^{-1}_M)^{i_nj_n}h_{i_1...i_n}h_{j_1...j_n}\big|.
\end{align*}
\begin{definition}
Let $n\in\mathbb{N}_0$ and let $\delta\in\reals_{>0}$. Then we say that a smooth seed data set for the equations of linearised gravity is asymptotically flat with weight $\delta$ to order $n$ iff there exists a positive constant $C_n$ such that the following pointwise bounds hold on $\Sigma_0$ for every $i\in\{0,...,n-1\}$:
\begin{align*}
\big|(r\overline{\nabla})^{i+1}(r^{\frac{1}{2}+\delta}\Phi_0)\big|_{\overline{g}_M}+\big|(r\overline{\nabla})^{i}(r^{\frac{3}{2}+\delta}\Phi_1)\big|_{\overline{g}_M}+\big|(r\overline{\nabla})^{i+1}(r^{\frac{1}{2}+\delta}\Psi_0)\big|_{\overline{g}_M}+\big|(r\overline{\nabla})^{i}(r^{\frac{3}{2}+\delta}\Psi_1)\big|_{\overline{g}_M}&\leq C_n,\\
\big|(r\overline{\nabla})^{i+1}(r^{\frac{3}{2}+\delta}\uV_0)\big|_{\overline{g}_M}+\big|(r\overline{\nabla})^{i}(r^{\frac{5}{2}+\delta}\uV_1)\big|_{\overline{g}_M}+\big|(r\overline{\nabla})^{i+1}(r^{\frac{3}{2}+\delta}\oV_0)\big|_{\overline{g}_M}+\big|(r\overline{\nabla})^{i}(r^{\frac{5}{2}+\delta}\oV_1)\big|_{\overline{g}_M}&\leq C_n.
\end{align*}
\end{definition}

\subsection{The well-posedness theorem}\label{Thewell-posednesstheorem}

We shall now show that the space of smooth solutions to the equations of linearised gravity can be completely parametrised by smooth seed data. This is the appropriate statement of well-posedness in view of the existence of constraints.  In particular, in the remainder of the paper we are now free to view solutions to the equations of linearised simply in terms of the seed data from which they arise. We shall prove this result by using the decomposition of Theorem \ref{thmdecoupling} to reduce the well-posedness statement to the corresponding well-posedness of the Cauchy problem for the Regge--Wheeler and Zerilli equations and the residual pure gauge equation. The well-posedness of these Cauchy problems is thus the content of sections \ref{TheCauchyinitialvalueproblemfortheRegge--WheelerandZerilliequations} and \ref{TheCauchyinitialvalueproblemfortheresidualpuregaugeequation} with the full well-posedness statement for the equations of linearised gravity given in section \ref{Thewell-posednesstheoremfortheequationsoflinearisedgravity}.

\subsubsection{The Cauchy initial value problem for the Regge--Wheeler and Zerilli equations}\label{TheCauchyinitialvalueproblemfortheRegge--WheelerandZerilliequations}

The following proposition can be proved by using standard theory\footnote{In particular, recall that the boundary $\eh$ is a null hypersurface.} combined with, for instance, a spherical harmonic decomposition and so we omit the proof.

\begin{proposition}\label{propwellposednesszerilli}
	Let $\psi_0, \psi_1\in \Gamma_{l\geq 2}(\Sigma_0)$. 
	
	Then there exists a unique solution $\psi\in\smfunrad$ to the Regge--Wheeler equation \eqref{eqnRWeqn} on $\mcalm$ such that
	\begin{align*}
	\big(i_0^*\psi, i_0^*n(\psi)\big)=\big(\psi_0, \psi_1\big).
	\end{align*}
	and we denote by $\solnmap_\Phi(\psi_0, \psi_1)$ the corresponding solution map.
	
	In addition, there exists a unique solution $\psi\in\smfunrad$ to the Zerilli equation \eqref{eqnZereqn} on $\mcalm$ such that
	\begin{align*}
	\big(i_0^*\psi, i_0^*n(\psi)\big)=\big(\psi_0, \psi_1\big)
	\end{align*}
	and we denote by $\solnmap_\Psi(\psi_0, \psi_1)$ the corresponding solution map.
\end{proposition}

Here $n$ is the future-pointing unit normal to $\Sigma_{0}$.

\subsubsection{The Cauchy initial value problem for the residual pure gauge equation}\label{TheCauchyinitialvalueproblemfortheresidualpuregaugeequation}

The following proposition can again be proved using standard theory combined with, for instance, a spherical harmonic decomposition and so we omit the proof.
\begin{proposition}\label{propwellposegaugesol}
	Let $\uV_0, \uV_1\in \smfunsig$ and let $\oV_0, \oV_1\in \smonecovsig$. Then there exists a unique $V\in\smonecov$ solving
	\begin{align*}
	\Box_{g_M}V&=\fgaumap(\nabla\astrosun V)
	\end{align*}
	such that
	$$\left.\begin{array}{rr}
	\Big(i_0^*\big(V(n)\big), i_0^*V, i_0^*\big(\mcalL_nV(n)\big),i_0^*\mcalL_nV \Big)=\Big(\uV_0,\oV_0, \uV_1,\oV_1\Big).
\end{array}\right.$$
We denote by $\solnmap_V(\uV_0,\oV_0, \uV_1,\oV_1)$ the corresponding solution map.
\end{proposition}

\subsubsection{The well-posedness theorem for the equations of linearised gravity}\label{Thewell-posednesstheoremfortheequationsoflinearisedgravity}

We now finally state and prove the well-posedness theorem for the equations of linearised gravity.\newline

To state the theorem correctly it will now and in the sequel be more appropriate to view solutions to the equations of linearsed gravity as members of the following solutions space:
\begin{align*}
\solnspace:=\Big\{\glin\in\smsymtwocov\,\big|\, \glin\text{ solves }\eqref{eqnlinearisedeinsteinequations}-\eqref{eqnlorentzgauge}\Big\}.
\end{align*}
This has a natural vector space structure over $\reals$. It will in addition be more appropriate to view smooth seed data sets as members of the space
\begin{align*}
\seedspace:=\Big\{\big(\Phi_0, \Phi_1, \Psi_0, \Psi_1, \uV_0, \uV_1, \oV_0, \oV_1, \mfm, \mfa_{-1},\mfa_{0},\mfa_{1}\big)\in\smfunsigrad^4{\times}\smfunsig^2{\times}\smonecovsig^2{\times}\reals^4\Big\}
\end{align*}
equipped with the canonical vector space structure over $\reals$. The subspace $\seedspace^{n,\delta}\subset\seedspace$ will then denote the vector space of smooth seed data sets that are asymptotically flat with weight $\delta>0$ to order $n$, noting the embedding $\seedspace^{k,\delta}\subset\seedspace^{n,\delta}$ for every $k\geq n$.

The well-posedness theorem for the equations of linearised gravity it then given as follows.
\begin{theorem}\label{thmwellposedness}
Let $\solnmap:\seedspace\rightarrow\solnspace$ be the map defined by
\begin{align*}
\solnmap\Big(\big(\Phi_0, \Phi_1, \Psi_0, \Psi_1, \uV_0, \uV_1, \oV_0, \oV_1, \mfm, \mfa_{-1},\mfa_{0},\mfa_{1}\big)\Big)=\,\,\,\,\,&\gamma\big(\solnmap_\Phi(\Phi_0, \Phi_1), \solnmap_\Psi(\Psi_0, \Psi_1)\big)\\
+&\nabla\astrosun\big(\solnmap_V(\uV_0, \uV_1, \oV_0, \oV_1)\big)\\
+&\sum_{i=-1,0,1}\glin_{\mfm,\mfa_{i}}.
\end{align*}
Then
\begin{enumerate}[i)]
	\item $\solnmap$ is a linear surjection
	\item the restriction $\solnmap:\seedspace^{0,\delta}\rightarrow\solnmap\big(\seedspace^{0,\delta}\big)$ is an isomorphism for every $\delta>0$.
\end{enumerate}
\end{theorem}
\begin{proof}
The linearity of the map is clear. Moreover, that $\solnmap$ indeed maps into the solution space follows from Propositions \ref{proplinkerr} and \ref{proppuregauge} combined with Proposition \ref{propsolnlingravRWZ} (in addition to of course Propositions \ref{propwellposednesszerilli} and \ref{propwellposegaugesol}). The first part of the theorem then follows as a simple consequence of Theorem \ref{thmdecoupling}. 

To conclude the second part we first note from remark \ref{rmkpuregaugedistinctfromkerr} and Corollary \ref{corrsolnlingravRWZ} that if $\big(\Phi_0, \Phi_1, \Psi_0, \Psi_1, \uV_0, \uV_1, \oV_0, \oV_1, \mfm, \mfa_{-1},\mfa_{0},\mfa_{1}\big)\in\ker\solnmap\cap\seedspace$ then
\begin{align}\label{pop}
\gamma\big(\solnmap_\Phi(\Phi_0, \Phi_1), \solnmap_\Psi(\Psi_0, \Psi_1)\big)=\nabla\astrosun\big(\solnmap_V(\uV_0, \uV_1, \oV_0, \oV_1)\big)=\sum_{i=-1,0,1}\glin_{\mfm,\mfa_{i}}=0.
\end{align}
It therefore follows immediately from Proposition \ref{proplinkerr} that $\mfm=\mfa_{-1}=\mfa_{0}=\mfa_{1}=0$. Moreover, since $\widehat{\slashed{g}}_M=0$ then if $(f,g)\in\ker\gamma$ it must be that $(f,g)\in\ker\sn\astrosunhat\sdso$ and thus $f$ and $g$ are supported only on $l=0,1$. Hence \eqref{pop} implies that $\solnmap_\Phi(\Phi_0, \Phi_1)=\solnmap_\Psi(\Psi_0, \Psi_1)=0$ and thus by the uniqueness criterion of Proposition \ref{propwellposednesszerilli} it follows that $\Phi_0= \Phi_1=\Psi_0=\Psi_1=0$. Finally, \eqref{pop} in addition implies that $\solnmap_V(\uV_0, \uV_1, \oV_0, \oV_1)$ solves Killings equations and thus either vanishes or 
\begin{align*}
\big(\solnmap_V(\uV_0, \uV_1, \oV_0, \oV_1)\big)^\sharp\in\text{span}\{\pt, \Omega_1, \Omega_2, \Omega_3\}
\end{align*}
where we recall from section \ref{KillingfieldsoftheSchwarzschildmetric} that the latter span the Lie algebra of Killing fields. However, the regularity assumptions on the seed exclude this latter scenario and thus $\solnmap_V(\uV_0, \uV_1, \oV_0, \oV_1)\big)=0$ which yields $\uV_0= \uV_1=\oV_0=\oV_1=0$ by the uniqueness criterion of Proposition \ref{propwellposegaugesol}.

\end{proof}

The uniqueness criterion of part $ii)$ in the above thus motivates the following definition.
\begin{definition}\label{defnafsolutions}
Let $\glin$ be a smooth solution to the equations of linearised gravity. Then we say that $\glin$ arises from smooth seed data that is asymptotically flat with weight $\delta$ to order $n$ iff $\glin\in\solnmap\big(\seedspace^{n,\delta}\big)$.
\end{definition}
The remainder of the paper is then concerned with the above class of solutions to the equations of linearised gravity. Note this class is both manifestly non-empty and manifestly parametrised by elements of $\seedspace^{n,\delta}$ in a one-to-one fashion.

\section{Residual pure gauge and linearised Kerr normalised solutions to the equations of linearised gravity}\label{GaugeandlinearisedKerrnormalisedsolutionstotheequationsoflinearisedgravity}

In this section we consider solutions to the equations of linearised gravity which have been normalised via the addition of a particular member of each of the special classes of solutions introduced in section \ref{Specialsolutionstotheequationsoflinearisedgravity}. It is these and only these \emph{residual pure gauge and linearised Kerr normalised solutions} to the equations of linearised gravity that our decay statement of Theorem 2 in section \ref{Theorem 2:Boundedness,decayandasymptoticflatnessofinitial-data-normalisedsolutionstotheequationsoflinearisedgravity} shall hold.

An outline of this section is as follows. We begin in section \ref{Initial-datanormalisedsolutionstotheequationsoflinearisedgravity} by defining a class of solutions to the equations of linearised gravity by demanding that they arise from a particular class of seed data. Then in section \ref{Achievingtheinitialdatanormalisationforageneralsolution} we show that these \emph{initial-data normalised} solutions can in fact be realised by adding a particular residual pure gauge and linearised Kerr solution to a general solution of the equations of linearised gravity. Finally in section \ref{Globalpropertiesofinitial-data-normalisedsolutions} we state and prove certain global properties of such initial-data normalised solutions to the equations of linearised gravity.

\subsection{Initial-data normalised solutions to the equations of linearised gravity}\label{Initial-datanormalisedsolutionstotheequationsoflinearisedgravity}

With the correct class of solutions to the equations of gravity that one should analyse understood as a consequence of Theorem \ref{thmwellposedness} we now in this section identify a subclass for which we shall establish a decay statement in section \ref{Proofoftheorem2}.\newline

To define this subclass of solutions we must first introduce the subspace $\radseedspace\subset\seedspace$ given by 
\begin{align*}
\radseedspace:=\Big\{\big(\Phi_0, \Phi_1, \Psi_0, \Psi_1, 0, 0, 0, 0, 0, 0,0,0\big)\in\seedspace\Big\}.
\end{align*}
We then in turn define the subspace $\radseedspace^{n,\delta}\subset\radseedspace$ as $\radseedspace^{n,\delta}=\radseedspace\cap\seedspace^{n,\delta}$.

\begin{definition}\label{defnRWgauge}
We say that a smooth solution $\glin$ to the equations of linearised gravity is initial-data normalised iff $\glin\in\solnmap(\radseedspace)$. We will denote such initial-data normalised solutions to the equations of linearised gravity by $\gidnlin$.

In addition, we say that an initial-data normalised solution to the equations of linearised gravity arises from smooth seed data that is asymptotically flat with weight $\delta>0$ to order $n$ iff $\gidnlin\in\solnmap\big(\radseedspace^{n,\delta}\big)$.
\end{definition}

Note whether a solution to the equations of linearised gravity is initial-data normalised is manifestly a condition on the seed data from which it arises as in Theorem \ref{thmwellposedness}. It is moreover clear that space of initial-data normalised solutions to the equations of linearised gravity is non-empty.

\subsection{Achieving the initial-data normalisation for a general solution}\label{Achievingtheinitialdatanormalisationforageneralsolution}

We now show that any smooth solution arising under part $ii)$ of Theorem \ref{thmwellposedness} can be made initial-data normalised via the addition of a unique residual pure gauge solution and unique linearised Kerr solutions. It thus follows that establishing a decay statement for initial-data normalsied solutions to the equations of linearised gravity yields a decay statement for solutions to the equations of linearised gravity that holds up to the addition of the special solutions of section \ref{Specialsolutionstotheequationsoflinearisedgravity}.\newline

The result is as follows.
\begin{theorem}\label{thminitialdatagauge}
	Let $\glin$ be a smooth solution to the equations of linearised gravity arising from smooth seed data that is asymptotically flat with weight $\delta>0$ to order $n$. Then there exists a unique $V\in\smonecov$ and unique parameters $\mfm, \mfa_{-1}, \mfa_{0},\mfa_{1}\in\reals$ such that
	\begin{align*}
	\gidnlin:=\glin-\glin_V-\sum_{i=-1,0,1}\glin_{\mfm, \mfa_{i}}
	\end{align*}
	is initial-data normalised. Moreover, $\gidnlin$ arises from smooth seed data that is asymptotically flat with weight $\delta>0$ to order $n$.
\end{theorem}

\begin{proof}
By part $ii)$ of Theorem \ref{thmwellposedness} there exists a unique smooth seed data set \\ $\big(\Phi_0, \Phi_1, \Psi_0, \Psi_1, \uV_0, \uV_1, \oV_0, \oV_1, \mfm, \mfa_{-1},\mfa_{0},\mfa_{1}\big)\in\seedspace^{n,\delta}$ such that
\begin{align*}
\glin=\solnmap\Big(\big(\Phi_0, \Phi_1, \Psi_0, \Psi_1, \uV_0, \uV_1, \oV_0, \oV_1, \mfm, \mfa_{-1},\mfa_{0},\mfa_{1}\big)\Big).
\end{align*}
Defining therefore $V\in\smonecov$ by 
\begin{align*}
V=\solnmap_V\big(\uV_0, \uV_1, \oV_0, \oV_1\big)
\end{align*}
then the linearity of the solution map yields
\begin{align*}
\glin-\glin_V-\sum_{i=-1,0,1}\glin_{\mfm, \mfa_{i}}=\solnmap\Big(\big(\Phi_0, \Phi_1, \Psi_0, \Psi_1, 0, 0, 0, 0, 0, 0,0,0\big)\Big).
\end{align*}
\end{proof}

\subsection{Global properties of initial-data normalised solutions}\label{Globalpropertiesofinitial-data-normalisedsolutions}

In this final section we prove certain global properties of initial-data-normalised solutions which will be fundamental in establishing the boundedness and decay statements of Theorem 2 in section \ref{Precisestatementsofthemaintheorems}.\newline

Indeed, the following proposition establishes that initial-data normalised solutions to the equations of linearised gravity fall into the class of solutions identified in Proposition \ref{propsolnlingravRWZ}.

\begin{proposition}\label{propglobalproperties}
	Let $\gidnlin$ be an initial-data normalised solution to the equations of linearised gravity. Then
	\begin{align*}
	\gidnlin=\gamma\big(\Psilin, \Philin\big)
	\end{align*}
	where $\Philin$ and $\Psilin$ are the invariant quantities associated to $\gidnlin$ of Theorem \ref{thmgaugeinvariantquantintermsofRWandZ}.
\end{proposition}
\begin{proof}
Since $\gidnlin$ is initial-data normalised it follows that
\begin{align*}
\gidnlin&=\solnmap\Big(\big(\Phi_0, \Phi_1, \Psi_0, \Psi_1, 0, 0, 0, 0, 0, 0,0,0\big)\Big),\\
&=\gamma\big(\solnmap_\Psi(\Psi_0, \Psi_1), \solnmap_\Phi(\Phi_0, \Phi_1)\big)
\end{align*}
for some $\big(\Phi_0, \Phi_1, \Psi_0, \Psi_1, 0, 0, 0, 0, 0, 0,0,0\big)\in\radseedspace$. Computing as in Theorem \ref{thmgaugeinvariantquantintermsofRWandZ} then yields $\Philin=\solnmap_\Phi(\Phi_0, \Phi_1)$ and $\Psilin=\solnmap_\Psi(\Phi_0, \Phi_1)$.
\end{proof}

We end this section by noting the following interesting corollary which follows immediately from Theorem \ref{thminitialdatagauge}, Proposition \ref{propglobalproperties} and the fact that $\gamma$ has trivial kernel over $\smfunrad\times\smfunrad$.

\begin{corollary}\label{corrglobalproperties}
Let $\glin$ be a smooth solution to the equations of linearised gravity arising from smooth seed data that is asymptotically flat with weight $\delta>0$ to order $n$ such that $\Philin=\Psilin=0$. Then $\glin$ is the sum of residual pure gauge and linearised Kerr solutions.
\end{corollary}

\section{Precise statements of the main theorems}\label{Precisestatementsofthemaintheorems}

In this section we finally give precise statements of the main theorems of this paper. These statements take the form of \emph{boundedness and decay} bounds for solutions to the Regge--Wheeler and Zerilli equations in addition to \emph{boundedness and decay} bounds for the initial-data normalised solutions to the equations of linearised gravity of section \ref{GaugeandlinearisedKerrnormalisedsolutionstotheequationsoflinearisedgravity}, with the norms by which such bounds are measured to be defined in this section. The relation between the former and latter is clear from Proposition \ref{propglobalproperties}

An outline of this section is as follows. We begin in section \ref{Flux,integrateddecayandpointwisenorms} by defining the norms required to correctly state the two theorems. Then in section \ref{Theorem1:BoundednessanddecayforsolutionstotheRegge--WheelerandZerilliequations} we state Theorem 1 concerning boundedness and decay bounds for solutions to the Regge--Wheeler and Zerilli equations in the norms of section \ref{Flux,integrateddecayandpointwisenorms}. Finally in section \ref{Theorem 2:Boundedness,decayandasymptoticflatnessofinitial-data-normalisedsolutionstotheequationsoflinearisedgravity} we state Theorem 12 concerning boundedness and decay bounds for initial-data normalised solutions to the equations of linearised gravity in the norms of section \ref{Flux,integrateddecayandpointwisenorms}.

\subsection{Flux, integrated decay and pointwise norms}\label{Flux,integrateddecayandpointwisenorms}

The norms in question concern flux, integrated decay and pointwise norms acting on smooth functions and smooth, symmetric 2-covariant tensors. Defining the action of these norms on the former is thus the content of section \ref{Flux,integrateddecayandpointwisenormsonsmfun} which we then upgrade to smooth, symmetric $n$-covariant $Q$ tensors, smooth $\qmsm$ 1-forms and smooth, symmetric, traceless 2-covariant $S$ tensors in section \ref{Flux,integrateddecayandpointwisenormsonsmncovQ,smqmsmandsmsymtratwocovS}. This then allows defining in section \ref{Flux,integrateddecayandpointwisenormsonsmsymtwocov} the action of said norms on smooth, symmetric 2-covariant tensors by exploiting the decomposition of section \ref{Ageometricfoliationby2-spheres}.

\subsubsection{Flux, integrated decay and pointwise norms on $\smfun$}\label{Flux,integrateddecayandpointwisenormsonsmfun}

First we define these norms for smooth functions $\psi$ on $\mcalm$. 

In what follows, we remind the reader of the function $\taus$ defined in section \ref{Ageometricfoliationby2-spheres} in which the constant $R$ was also fixed. Moreover, we recall the $L^2$ norms on spheres defined in section \ref{Normsonspheres}.\newline

We associate to $\psi$ the energy norm
\begin{align*}
\mathbb{E}[\psi](\taus):=&\int_{2M}^R\Big(\normtwosphere{\pt\psi}{}{r}+\normtwosphere{\pr\psi}{}{r}+\normtwosphere{\sn\psi}{}{r}\Big)\exd r\\
+&\int_{R}^\infty\Big(\normtwosphere{D(r\psi)}{}{r}+\normtwosphere{\sn(r\psi)}{}{r}\Big)\exd r
\end{align*}
and the $r^p$-weighted norms
\begin{align*}
\mathbb{F}_p[\psi](\taus)&:=\int_{R}^\infty\Big(r^p\normtwosphere{D(r\psi)}{}{r}+\normtwosphere{\sn(r\psi)}{}{r}\Big)\exd r.
\end{align*}
Here $D:=\tfrac{1+\mu}{1-\mu}\pt+\pr$.

This leads to the weighted flux norm
\begin{align*}
\mathbb{F}[\psi]:=&\sup_{\taus\in[0, \infty)}\mathbb{E}[\psi](\taus)
+\sup_{\taus\in(-\infty, \infty)}\int_{R}^\infty \Big(r^2\normtwosphere{D(r\psi)}{}{r}+\normtwosphere{\sn(r\psi)}{}{r}\Big)\exd r.
\end{align*}
We morever define the initial flux norms along the initial Cauchy hypersurface $\Sigma$ of section \ref{TheinitialCauchyhypersurfaceSigma}:
\begin{align*}
\mathbb{D}[\psi]:=\int_{\Sigma}\Big(||n(\psi)||^2_{\overline{g}_M}+||\overline{\nabla}\psi||^2_{\overline{g}_M}\Big)\epsilon_{\overline{g}_M}
\end{align*}
with $\epsilon_{\overline{g}_M}$ the volume form associated to $\overline{g}_M$.

We further associate to $\psi$ the integrated local energy decay norm
\begin{align*}
\mathbb{I}_{\text{loc}}[\psi](\taus_1):=\int_{\taus_1}^\infty\int_{2M}^{R}\Big(\normtwosphere{\pt(r\psi)}{}{r}+\normtwosphere{\pr(r\psi)}{}{r}+\normtwosphere{\sn(r\psi)}{}{r}+\normtwosphere{r\psi}{}{r}\Big)\exd \taus\exd r
\end{align*}
and the $r^p$-weighted bulk norms
\begin{align*}
\mathbb{B}_p[\psi](\taus_1):=\int_{\tau^\star_1}^\infty\int_{R}^\infty r^p\Big( \normtwosphere{D(r\psi)}{}{r}+\normtwosphere{\sn(r\psi)}{}{r}\Big)\exd \taus\exd r.
\end{align*}
This leads to the integrated decay norm
\begin{align*}
\mathbb{M}[\psi]:=\int_{0}^\infty\int_{2M}^{\infty}\frac{1}{r^3}\Big(\normtwosphere{\pt(r\psi)}{}{r}+\normtwosphere{\pr(r\psi)}{}{r}+\normtwosphere{\sn(r\psi)}{}{r}+\normtwosphere{r\psi}{}{r}\Big)\exd \taus\exd r
\end{align*}
and the weighted bulk norm
\begin{align*}
\mathbb{I}[\psi]:=\int_{0}^\infty\int_{R}^\infty\Big(r \normtwosphere{D(r\psi)}{}{r}+r^{\beta_0} \normtwosphere{\sn(r\psi)}{}{r}\Big)\exd \taus\exd r.
\end{align*}
Here, $\beta_0>0$ is a fixed constant such that $1-\beta_0<<1$.\newline

Higher order flux norms are then defined according to, for $n\geq 1$ an integer, 
\begin{align*}
\mathbb{F}^n[\psi]:=&\sum_{i+j+k=0}^n\sup_{\taus\in[0, \infty)}\mathbb{E}[\pt^i\pr^j\sn^k\psi](\taus)\\
+&\sum_{i+j+k=0}^n\sup_{\taus\in(-\infty, \infty)}\int_{R}^\infty\Big( r^2\normtwosphere{D\underline{D}^i(rD)^j(r\sn)^k(r\psi)}{}{r}+\normtwosphere{\sn\underline{D}^i(rD)^j(r\sn)^k(r\psi)}{}{r}\Big)\exd r 
\end{align*}
with analogous definitions for the higher order energy and $r^p$-weighted norms. Here $\underline{D}:=(1-\mu)(\pt-\pr)$. Note to compute the higher order $L^2$ norms on spheres one simply uses the commutation formulae of section \ref{Commutationformulaeandusefulidentities} combined with Definition \ref{defnL2norm}.

Conversely, higher order initial flux norms are defined according to
\begin{align*}
\mathbb{D}^n[{\psi}]:=\sum_{i+j=0}^n\mathbb{D}[(r\overline{\nabla})^i\psi].
\end{align*}

Lastly, higher order integrated decay norms are defined according to
\begin{align*}
\mathbb{M}^n[\psi]:&=\sum_{i+j+k=0}^{n}\mathbb{M}[\pt^i\pr^j\sn^k\psi],\\
\mathbb{I}^n[\psi]:&=\sum_{i+j+k=0}^{n}\mathbb{I}[\underline{D}^i(rD)^j(r\sn)^k\psi]
\end{align*}
with analogous definitions for the higher order integrated local energy decay norms and the $r^p$-weighted bulk norms.

Finally then we have the following pointwise norm on the 2-spheres $\twosphere_{\taus, r}$:
\begin{align*}
|\psi|_{\taus, r}:=\sup_{\twosphere_{\taus, r}}|\psi|.
\end{align*}
Higher order pointwise norms are then defined in the now canonical way.

\subsubsection{Flux, integrated decay and pointwise norms on $\smsymncovQ$, $\smqmsm$ and $\smsymtratwocovS$}\label{Flux,integrateddecayandpointwisenormsonsmncovQ,smqmsmandsmsymtratwocovS}

Now we upgrade the norms of the previous section from $\smfun$ to $\smncovQ$, $\smqmsm$ and $\smsymtratwocovS$. Note however that generalising the norms $\mathbb{D}$ will be unnecessary.\newline

In fact, replacing the derivative operators $\pt, \pr, D$ and $\underline{D}$ by $\qn_{\pt}, \qn_{\pr},\qn_ D$ and $\qn_{\underline{D}}$ in the appropriate norms of section \ref{Flux,integrateddecayandpointwisenormsonsmfun} then those norms are equally well-defined on objects in $\smncovQ$, $\smqmsm$ and $\smsymtratwocovS$ by Definition \ref{defnL2norm}, the commutation formulae provided by the Riemann tensors in section \ref{Decomposingtheequationsoflinearisedgravity} and the calculus developed in section \ref{MixedQandStensoranalysis}\footnote{In particular, note that $\pt, \pr, D$ and $\underline{D}$ are $Q$ vector fields.}.

Finally, for the pointwise norms we define
\begin{align*}
|\qtau|_{{\taus, r}}:&=\sup_{\twosphere_{\taus, r}}\sum_{i+j=n}|\qtau_{ij}|, \qquad \qtau\in\smsymncovQ,\\
|{\mtau}|_{\taus, r}:&=\sup_{\twosphere_{\taus, r}}\Big(|\sn\even{\mtau}(\pt)|_{\sg}+|\sn\even{\mtau}(\px)|_{\sg}+|\shd\sn\odd{\mtau}(\pt)|_{\sg}+|\shd\sn\odd{\mtau}(\px)|_{\sg}\Big), \quad \mtau\in\smqmsm,\\
|\stau|_{{\taus, r}}:&=\sup_{\twosphere_{\taus, r}}|\stau|_{\sg}, \qquad \stau\in\smsymncovS.
\end{align*}
Here $\qtau_{ij}$ are the components of $\qtau$ in the natural frame (cf. section \ref{TheprojectionofsmoothQtensors,smoothsymmetrictracelessStensorsandsmoothqmsm1-formsontoandawayfromthel=0,1sphericalharmonics}) in particular utilising the calculus developed in section \ref{QandStensoranalysis} and section \ref{MixedQandStensoranalysis}.

The higher order pointwise norms are then defined in the now canonical way.

\subsubsection{Flux, integrated decay and pointwise norms on $\smsymtwocov$}\label{Flux,integrateddecayandpointwisenormsonsmsymtwocov}

We now finally upgrade the norms of the previous section to $\smsymtwocov$. To do this we will utilise the decomposition of section \ref{Ageometricfoliationby2-spheres}.\newline

Indeed for $\tau\in\smsymtwocov$ and $n\in\mathbb{N}_0$ we define
\begin{align*}
\mathbb{F}^n[\tau]:&=\mathbb{F}^n[\qhattau]+\mathbb{F}^n[\qtrtau]+\mathbb{F}^n[\mtau]+\mathbb{F}^n[\shattau]+\mathbb{F}^n[\strtau].
\end{align*}
The remaining norms are then defined analogously.

\subsection{Theorem \ref{mainthmdecayRWandZ}: Boundedness and decay for solutions to the Regge--Wheeler and Zerilli equations}\label{Theorem1:BoundednessanddecayforsolutionstotheRegge--WheelerandZerilliequations}

In this section we state Theorem 1 which concerns both a boundedness and decay statement for solutions to the Regge--Wheeler and Zerilli equations in the norms of section \ref{Flux,integrateddecayandpointwisenormsonsmfun}.

The proof of Theorem 1 is the content of section \ref{Proofoftheorem1}.\newline

The theorem statement is as below -- we note that in the statement we drop the superscript $(1)$ from all quantities under consideration as the theorem holds independently of the relation between the Regge--Wheeler and Zerilli equations and the equations of linearised gravity.

\begin{customthm}{1}\label{mainthmdecayRWandZ}
	Let $\Phi\in\smfunrad$ be a solution to the  Regge--Wheeler equation on $\Mg$:
	\begin{align*}
	\qbox\Phi+\slap\Phi=-\frac{6}{r^2}\frac{M}{r}\Phi.
	\end{align*}
	
	Then for any integer $n\geq 2$ the following estimates hold provided that the fluxes on the right hand side are finite:
	\begin{enumerate}[i)]
		\item the higher order flux and weighted bulk estimates
		\begin{align*}
		\norm{F}{r^{-1}\Phi}[n]+\norm{I}{r^{-1}\Phi}[n]&\lesssim\norm{D}{r^{-1}\Phi}[n]
		\end{align*}
		\item the higher order integrated decay estimate
		\begin{align*}
		\norm{M}{r^{-1}\Phi}[n]&\lesssim\norm{D}{r^{-1}\Phi}[n+1]
		\end{align*}
		\item finally, on any 2-sphere $\twosphere_{\taus,r}\subset\mcalm$ and any positive integers $i+j+k+l+m\geq n-2$, the pointwise decay bounds  
		\begin{align*}
		|\pt^i\pr^j\big((r-2M)D\big)^k(r\sn)^l\underline{D}^m\Phi|_{{\taus,r}}\lesssim\frac{1}{\sqrt{\taus}}\cdot\norm{D}{r^{-1}\Phi}[n].
		\end{align*}
	\end{enumerate}
	
	Let now $\Psi\in\smfunrad$ be a solution to the  Zerilli equation on $\Mg$:
	\begin{align*}
	\qbox\Psi+\slap\Psi=-\frac{6}{r^2}\frac{M}{r}\Psi+\frac{24}{r^3}\frac{M}{r}(r-3M)\zslapinv{1}\Psi+\frac{72}{r^3}\frac{M}{r}\frac{M}{r}(r-2M)\zslapinv{2}\Psi.
	\end{align*}

	Then for any integer $n\geq 2$ the following estimates hold provided that the fluxes on the right hand side are finite:
	\begin{enumerate}[i)]
		\item the higher order flux and weighted bulk estimates
		\begin{align*}
		\norm{F}{r^{-1}\Psi}[n]+\norm{I}{r^{-1}\Psi}[n]&\lesssim\norm{D}{r^{-1}\Psi}[n]
		\end{align*}
		\item the higher order ntegrated decay estimate
		\begin{align*}
		\norm{M}{r^{-1}\Psi}[n]&\lesssim\norm{D}{r^{-1}\Psi}[n+1]
		\end{align*}
		\item finally, on any 2-sphere $\twosphere{}{r}$ with $\taus\geq\taus_0$ and any positive integers $i+j+k+l+m\geq n-2$, the pointwise decay bounds  
		\begin{align*}
		|\pt^i\pr^j\big((r-2M)D\big)^k(r\sn)^l\underline{D}^m\Psi|_{_{\taus,r}}\lesssim\frac{1}{\sqrt{\taus}}\cdot\norm{D}{r^{-1}\Psi}[n].\newline
		\end{align*}
	\end{enumerate}
\end{customthm}

We make the following remarks regarding Theorem 1.
\begin{remark}
	The $r-2M$ weight in the pointwise bounds of $iii)$ is to regularise the operator $D$ at $\eh$. In particular, for $r\geq R$ one can replace $(r-2M)D$ with $rD$.
\end{remark}
\begin{remark}
	The contents of Theorem 1 regarding the Regge--Wheeler equation were originally proven by Holzegel in \cite{Holzegelultschwarz} (see also \cite{DHRlinstabschwarz}). Conversely, the contents of Theorem 1 regarding the Zerilli equation were originally proven in the independent works of the author \cite{Johnsondecayz} and Hung--Keller--Wang \cite{HKWlinstabschwarz}.
\end{remark}
\begin{remark}
	One can show from parts $i)$ of Theorem 1 that the quantities $|\Philin|$ and $|\Psilin|$ have finite limits on $\nullinf$ -- see \cite{Moschidisrpmethod} for details.
\end{remark}

\subsection{Theorem \ref{mainthmdecaysolutionslingrav}: Boundedness, decay and asymptotic flatness of initial-data-normalised solutions to the equations of linearised gravity }\label{Theorem 2:Boundedness,decayandasymptoticflatnessofinitial-data-normalisedsolutionstotheequationsoflinearisedgravity}

In this section we state Theorem 2 which concerns both a boundedness and decay statement for initial-data-normalised solutions $\gidnlin$ to the equations of linearised gravity in the flux and integrated decay norms of section \ref{Flux,integrateddecayandpointwisenormsonsmsymtwocov}. In addition, we provide a statement of asymptotic flatness for the solution $\gidnlin$. 

The proof of Theorem 2 is the content of section \ref{Proofoftheorem2}.\newline

The theorem statement is as given below.
\begin{customthm}{2}\label{mainthmdecaysolutionslingrav}
	Let $\gidnlin$ be a smooth initial-data normalised solution to the equations of linearised gravity arising from a smooth seed data set that is asymptotically flat with weight $\delta>0$ to order $n\geq4$ and let $\Philin$ and $\Psilin$ denote the invariant quantities associated to $\gidnlin$ in accordance with Theorem \ref{thmgaugeinvariantquantintermsofRWandZ}. 
	
	Then the initial flux norms
	\begin{align}\label{initialflux}
	\norm{D}{r^{-1}\Philin}[4]+\norm{D}{r^{-1}\Psilin}[4]
	\end{align}
	are finite and $\Philin$ and $\Psilin$ satisfy the conclusions of Theorem 1 with $n=4$.
	
	Moreover the following estimates hold on the solution $\gidnlin$:
	\begin{enumerate}[i)]
		\item the flux and weighted bulk estimates 
		\begin{align*}
		\norm{F}{\gidnlin}[2]+\norm{I}{\gidnlin}[2]&\lesssim\norm{D}{r^{-1}\Philin}[4]+\norm{D}{r^{-1}\Psilin}[4]
		\end{align*}
		\item the integrated decay estimates
		\begin{align*}
		\norm{M}{\gidnlin}[1]&\lesssim\norm{D}{r^{-1}\Philin}[4]+\norm{D}{r^{-1}\Psilin}[4]
		\end{align*}
		\item finally on any 2-sphere $\twosphere_{\taus, r}\subset\mcalm$ the pointwise decay bounds  
		\begin{align*}
		|r\gidnlin|_{{\taus,r}}\lesssim\frac{1}{\sqrt{\taus}}\cdot\Big(\norm{D}{r^{-1}\Philin}[4]+\norm{D}{r^{-1}\Psilin}[4]\Big).
		\end{align*}
	\end{enumerate}
\end{customthm}
We then have the immediate corollary:
\begin{corollary}\label{corrmaintheorem}
	Let $\glin$ be a smooth solution to the equations of linearised gravity arising from a smooth seed data set that is asymptotically flat with weight $\delta>0$ to order $n\geq4$. Then the initial-data normalised solution constructed from $\glin$ in accordance with Theorem \ref{thminitialdatagauge} satisfies the assumptions and hence the conclusions of Theorem 2. In particular, $\glin$ decays inverse polynomially to the sum of a residual pure gauge and linearised Kerr solutions.
\end{corollary}
We make the following remarks regarding Theorem 1.
\begin{remark}\label{reminitialdatanormalised}
Observe that it suffices to assume only that the initial flux bounds of \eqref{initialflux} are finite in order for the conclusions of the theorem to hold.
\end{remark}
\begin{remark}
	As in parts $iii)$ of Theorem 1 one can obtain higher order analogues of the pointwise decay bounds in part $iii)$ of the theorem statement although we decline to state these explicitly.
\end{remark}
\begin{remark}\label{rmkAF}
From part $i)$ of Theorem 2 one can show that the tensor $r\gidnlin$ has a finite limit on $\nullinf$. Note one could not show this for the solutions to the linearised system considered in our recent \cite{Johnsonlinstabschwarzold}.
\end{remark}

\section{Proof of Theorem \ref{mainthmdecayRWandZ}}\label{Proofoftheorem1}

In this section we prove Theorem 1. In fact, we prove only the following: first, we shall prove Theorem 1 only for the Zerilli equation as this is the simpler case. Second, we shall prove only the $\pt$-flux and Morawetz estimates discussed in section \ref{OVOutlineoftheproofofTheorem1} as the remaining estimates along with their higher order counterparts follow in an analogus fashion as discussed for the scalar wave equation in section \ref{OVAside:thescalarwaveequationontheSchwarzschildexteriorspacetime}. Finally, we shall prove these estimates only in the region to the future of $\Sigma_{0}\cap\{r\leq R\}$ as it will then be clear how to prove the estimates in the semi-global region.

An outline of this section is as follows. We begin in section \ref{IntegralidentitiesandintegralestimatesforsolutionstotheRegge--WheelerandZerilliequations} by deriving various integral identities and integral estimates that smooth solutions to the Zerilli equation must satisfy. Finally in section \ref{ThedegenerateenergyandMorawetzestimate} we prove the $\pt$ flux and Morawetz estimates for solutions to the Zerilli equation by exploiting the identities and estimates of section \ref{IntegralidentitiesandintegralestimatesforsolutionstotheRegge--WheelerandZerilliequations}. 

\subsection{Integral identities and integral estimates for solutions to the Regge--Wheeler and Zerilli equations}\label{IntegralidentitiesandintegralestimatesforsolutionstotheRegge--WheelerandZerilliequations}

We begin in this section by deriving various identities that solutions to the Zerilli equation must satisfy, recalling our earlier comment that we forgo explicitly analysing the solutions to the Regge--Wheeler equation in this paper. These identities shall then be utilised throughout the remainder of the section. \newline

Let us first recall for ease of reference the definition of the Zerilli equation from Definition \ref{defnRWandZ} for a functions $\Psi\in\smfunrad$ :
\begin{align}
\qbox\Psi+\slap\Psi&=-\frac{6}{r^2}\frac{M}{r}\Psi+\frac{24}{r^3}\frac{M}{r}(r-3M)\zslapinv{1}\Psi+\frac{72}{r^3}\frac{M}{r}\frac{M}{r}(r-2M)\zslapinv{2}\Psi\label{Zereqn}.
\end{align}
Here, we further recall the operator $\zslapinv{p}$ defined as in section \ref{Thefamilyofoperatorsslap}.

We then have that solutions to the above must satisfy the following set of identities.

In what follows, given two smooth functions $f$ and $g$ on $\mcalm$ we denote by $\langle f, g\rangle_{\twosphere_{\taus, r}}$ their $L^2$ product on any 2-sphere $\twosphere{}{r}$:
\begin{align*}
\langle f, g\rangle_{\twosphere_{\taus, r}}:=\frac{1}{r^2}\inttwosphere{}{r}{f\, g}
\end{align*}
In addition, for a $f$ a smooth function of $r$ on $\mcalm$ we define
\begin{align*}
f':=\pr f.
\end{align*}
Finally, we recall the mass aspect function $\mu=\frac{2M}{r}$.
\begin{lemma}\label{lemmaidentitesRWandZ}
	Let $\alpha, \beta$ and $w$ be smooth functions of $r$ on $\mcalm$.
	
	Let now $\Psi$ be a smooth solution to the Zerilli equation \eqref{Zereqn}. Then on any 2-sphere $\geomtwosphere{}{r}$ the following identities hold:
	\begin{align*}
	\pt&\Big[\opmu\alpha\,\normtwosphere{\pt\Psi}{}{r}+\ommu\alpha\,\normtwosphere{\pr\Psi}{}{r}+\alpha\,\normtwosphere{\sn_{\mfZ}\Psi}{}{r}\Big]\\
	-\pr&\Big[2\mu\alpha\,\normtwosphere{\pt\Psi}{}{r}+2\ommu\alpha\,\langle\pt\Psi, \pr\Psi\rangle_{\twosphere_{\taus, r}}\Big]\\
	&+2\mu\alpha'\,\normtwosphere{\pt\Psi}{}{r}\\
	=&-2\ommu\alpha'\,\langle\pt\Psi, \pr\Psi\rangle_{\twosphere_{\taus, r}},
	\end{align*}
	\begin{align*}
	\pt&\Big[2\opmu\beta\,\langle\pt\Psi, \pr\Psi\rangle_{\twosphere_{\taus, r}}-2\mu\beta\,\normtwosphere{\pr\Psi}{}{r}\Big]\\
	-\pr&\Big[\opmu\beta\,\normtwosphere{\pt\Psi}{}{r}-\ommu\beta\,\normtwosphere{\pr\Psi}{}{r}-\beta\,\normtwosphere{\sn_{\mfZ}\Psi}{}{r}\Big]\\
	&+\big(\opmu\beta\big)'\normtwosphere{\pt\Psi}{}{r}+\bigg(\ommu\beta'-\frac{\mu}{r}\beta\bigg)\normtwosphere{\pr\Psi}{}{r}\\
	&-\beta'\,\normtwosphere{\sn_{\mfZ}\Psi}{}{r}-\beta\,\normtwosphere{[\pr,\mfZ]\Psi}{}{r}\\
	=&-2\frac{\mu}{r}\,\beta\,\langle\pt\Psi, \pr\Psi\rangle_{\twosphere_{\taus, r}}
	\end{align*}
	and
	\begin{align*}
	-\pt&\bigg[\opmu w\,\langle\pt\Psi, \Psi\rangle_{\twosphere_{\taus, r}}-2\mu w\,\langle\pr\Psi, \Psi\rangle_{\twosphere_{\taus, r}}+\frac{\mu}{r}\frac{w}{2}\,\normtwosphere{\Psi}{}{r}\bigg]\\
	+\pr&\bigg[\ommu w\,\langle\pr\Psi, \Psi\rangle_{\twosphere_{\taus, r}}-\frac{1}{2}\bigg(\big(\ommu w\big)'+\frac{\mu}{r}\frac{w}{2}\bigg)\normtwosphere{\Psi}{}{r}\bigg]\\
	&+\opmu w\,\normtwosphere{\pt\Psi}{}{r}-\ommu w\,\normtwosphere{\pr\Psi}{}{r}-w\,\normtwosphere{\sn_{\mfZ}\Psi}{}{r}\\
	&+\frac{1}{2}\big(\ommu w'\big)'\normtwosphere{\Psi}{}{r}\\
	=&2\mu w\,\langle\pt\Psi, \pr\Psi\rangle_{\twosphere_{\taus, r}}.
	\end{align*}
	Here, for $f\in\smfunrad$ we define
	\begin{align*}
	\normtwosphere{\sn_{\mfZ}f}{}{r}:=\normtwosphere{\sn f}{}{r}+\frac{6}{r}\frac{\mu}{r}\Big((2-3\mu)^2+3\mu\ommu\Big)&\normtwosphere{\zslapinv{1}f}{}{r}\\
	-\frac{6}{r}\frac{\mu}{r}(2-3\mu)&\normtwosphere{(r\sn)\zslapinv{1}f}{}{r}\\
	-\frac{3}{r}\frac{\mu}{r}&\normtwosphere{f}{}{r}
	\end{align*}
	and
	\begin{align*}
	\normtwosphere{[\pr,\mfZ]f}{}{r}:=-\frac{2}{r}\normtwosphere{\sn f}{}{r}+\frac{9}{r^2}\frac{\mu}{r}\normtwosphere{f}{}{r}-\frac{108}{r^3}\mu^3\ommu(2-3\mu)&\normtwosphere{\zslapinv{2}f}{}{r}\\
	+\frac{108}{r^3}\mu^3\ommu&\normtwosphere{(r\sn)\zslapinv{2}f}{}{r}\\
	-\frac{18}{r^2}\frac{\mu}{r}\Big((2-3\mu)^2+\mu^2\Big)&\normtwosphere{\zslapinv{1}f}{}{r}\\
	+\frac{36}{r^2}\frac{\mu}{r}(1-2\mu)&\normtwosphere{(r\sn)\zslapinv{1}f}{}{r}.
	\end{align*}
\end{lemma}
\begin{proof}
We recall the commutation relations and integration by parts formulae of Lemma \ref{lemmacommrelationsandidentities} for the operator $\zslapinv{p}$ along with the coordinate form of the operators $\qbox$ and $\slap$ given in section \ref{Decomposingtheequationsoflinearisedgravity}. The lemma then follows after multiplying the Zerilli equation \eqref{Zereqn} successively by the smooth functions
	\begin{align*}
	\alpha\,\pt\Psi, \beta\,\pr\Psi, w\,\Psi\in\smfunrad
	\end{align*}
	and then integrating by parts on any 2-sphere $\geomtwosphere{}{r}$ with respect to the measure $r^{-2}\sepsilon$.
\end{proof}

Given a smooth solution $\Psi$ to the Zerilli equation on $\mcalm$ and three smooth radial functions $\alpha,\beta$ and $w$, the last three identities of Lemma \ref{lemmaidentitesRWandZ} motivate introducing the 1-form $\tilde{\mathbb{J}}^{\alpha, \beta,w}_{\taus,r}[\Psi]$ and function $ \widetilde{\mathbb{K}}^{\alpha, \beta,w}_{\taus,r}[\Psi]$ on $\mcalm$ defined according to
\begin{align*}
\tilde{\mathbb{J}}^{\alpha, \beta,w}_{\taus,r}[\Psi]:=\bigg[\opmu\alpha\,&\normtwosphere{\pt\Psi}{}{r}+2\opmu\beta\,\langle\pt\Psi, \pr\Psi\rangle_{\twosphere_{\taus, r}}+\big(\ommu\alpha-2\mu\beta\big)\normtwosphere{\pr\Psi}{}{r}\\+\alpha\,&\normtwosphere{\sn_{\mfZ}\Psi}{}{r}
-\opmu w\,\langle\pt\Psi, \Psi\rangle_{\twosphere_{\taus, r}}+2\mu w\,\langle\pr\Psi, \Psi\rangle_{\twosphere_{\taus, r}}\\
-\frac{\mu}{r}\frac{w}{2}\,&\normtwosphere{\Psi}{}{r}\bigg]\exd t^*\\
-\bigg[\big(2\mu\alpha+\opmu\beta\big)\,&\normtwosphere{\pt\Psi}{}{r}+2\ommu\alpha\,\langle\pt\Psi, \pr\Psi\rangle_{\twosphere_{\taus, r}}+\ommu\beta\,\normtwosphere{\pr\Psi}{}{r}\\-\beta\,&\normtwosphere{\sn_{\mfZ}\Psi}{}{r}
-\ommu w\,\langle\pr\Psi, \Psi\rangle_{\twosphere_{\taus, r}}\\
+\frac{1}{2}\bigg(\big(\ommu w\big)'-\frac{\mu}{r}\frac{w}{2}\bigg)&\normtwosphere{\Psi}{}{r}\bigg]\exd r
\end{align*}
and
\begin{align*}
\widetilde{\mathbb{K}}^{\alpha, \beta,w}_{\taus,r}[\Psi]=\Big(2\mu\alpha'\,+\big(\opmu\beta\big)'+\opmu w\Big)&\normtwosphere{\pt\Psi}{}{r}\\
2\bigg(\ommu\alpha'+\frac{\mu}{r}\,\beta-2\mu\,w\bigg)&\langle\pt\Psi, \pr\Psi\rangle_{\twosphere_{\taus, r}}\\
+\bigg(\ommu\beta'-\frac{\mu}{r}\beta-\ommu w\bigg)&\normtwosphere{\pr\Psi}{}{r}\\
-(\beta'+w)&\normtwosphere{\sn_{\mfZ}\Psi}{}{r}\\
-\beta\,&\normtwosphere{[\pr,\mfZ]\Psi}{}{r}\\
+\frac{1}{2}\big(\ommu w'\big)'&\normtwosphere{\Psi}{}{r}.
\end{align*}
Indeed, given in addition three real numbers $T^*>\taus_2\geq\taus_1$, then summing the identities in the first half of Lemma \ref{lemmaidentitesRWandZ} and integrating over the region\footnote{The hypersurfaces $\Sigma_{t^*}$ are defined by the level sets of $t^*$.}
\begin{align*}
\mathfrak{R}^{T^*}_{\taus_1, \taus_2}:=\bigg(\bigcup\limits_{t^*\in[\taus_1, T^*]}\Sigma_{t^*}\bigg)\bigcap\bigg(\bigcup\limits_{\taus\in[\taus_1, \taus_2]}\bigcup\limits_{r\in[2M, \infty)}\geomtwosphere{}{r}\bigg)
\end{align*}
with respect to the measure $r^{-2}\exd \taus\exd r\,\sepsilon$ yields the conservation law
\begin{align}
&\int_{2M}^{r(T^*, \taus_2)}\tilde{\mathbb{J}}^{\alpha, \beta,w}_{\taus_2, r}[\Psi](\partial_{\taus})\exd r+\int_{r(T^*,\taus_1)}^{r(T^*, \taus_2)}\tilde{\mathbb{J}}^{\alpha, \beta,w}_{{\taus(T^*,r),r}}[\Psi](\pt)\exd r\nonumber+\int_{\taus_1}^{\taus_2}\int_{2M}^{r(T^*, \taus)}\widetilde{\mathbb{K}}^{\alpha, \beta,w}_{\taus,r}[\Psi]\exd\taus\exd r\\
=&\int_{2M}^{r(T^*, \taus_1)}\tilde{\mathbb{J}}^{\alpha, \beta,w}_{\taus_1,r}[\Psi](\partial_{\taus})\exd r+\int_{\taus_1}^{ \taus_2}\tilde{\mathbb{J}}^{\alpha, \beta,w}_{\taus, 2M}[\Psi](\partial_r)\exd \taus.\label{conslawZ}
\end{align}
Here, $r(T^*, \taus_i)>R$ for $i=1,2$ is the unique value of $r$ such that $\taus(T^*, r)=\taus_i$ and we note that $\twosphere_{\taus(T^*,r),r}$ is the 2-sphere $\{T^*\}\times\{r\}\times \twosphere\subset\Sigma_{T^*}$. In addition, we have defined for $i=1,2$
\begin{align*}
\int_{2M}^{r(T^*, \taus_i)}\tilde{\mathbb{J}}^{\alpha, \beta,w}_{\taus_i,r}[\Psi](\partial_{\taus})\exd r:=\int_{2M}^{r(T^*, \taus_i)}\tilde{\mathbb{J}}^{\alpha, \beta,w}_{\taus_i,r}[\Psi](\pt)\exd r-\int_{R}^{r(T^*, \taus_i)}\frac{1+\mu}{1-\mu}\,\tilde{\mathbb{J}}^{\alpha, \beta,w}_{\taus_i,r}[\Psi](\pr)\exd r.\newline
\end{align*}
\\

We end this section with the following Proposition that arises from taking the limit $T^*\rightarrow\infty$ in the conservation law \eqref{conslawZ}.

\begin{proposition}\label{propfluxandbulkestimate}
	Let $\alpha, \beta$ and $w$ be smooth functions of $r$ on $\mcalm$ and let $T^*>\taus_2\geq\taus_1$ be three real numbers.
	
	We suppose that $\Psi$ is a smooth solution to the Zerilli equation \eqref{Zereqn} on $\mcalm$ for which at least one of the following two conditions hold:
	\begin{enumerate}[i)]
		\item for any $(\taus,r)$ the quantities $\tilde{\mathbb{J}}^{\alpha, \beta,w}_{\taus,r}[\Psi](\pt)$ and $-\tilde{\mathbb{J}}^{\alpha, \beta,w}_{\taus,2M}[\Psi](\pr)$ are non-negative
		\item for any $(\taus,r)$ the quantity $\tilde{\mathbb{J}}^{\alpha, \beta,w}_{\taus,r}[\Psi](\pt)$ is non-negative and uniformly in $T^*$ and $\taus_2$ it holds that \newline $\int_{\taus_1}^{\taus_2}\tilde{\mathbb{J}}^{\alpha, \beta,w}_{\taus,2M}[\Psi](\pr)\exd\taus\lesssim\int_{2M}^{r(T^*, \taus_1)}\tilde{\mathbb{J}}^{\alpha, \beta,w}_{\taus_1,r}[\Psi](\partial_{\taus})\exd r$
		\item for any $\taus$ the quantity $-\tilde{\mathbb{J}}^{\alpha, \beta,w}_{\taus,2M}[\Psi](\pr)$ is non-negative and uniformly in $T^*$ and $\taus_2$ it holds that\newline $-\int_{r(T^*,\taus_1)}^{r(T^*, \taus_2)}\tilde{\mathbb{J}}^{\alpha, \beta,w}_{{\taus(T^*,r),r}}[\Psi](\pt)\exd r\lesssim\int_{2M}^{r(T^*, \taus_1)}\tilde{\mathbb{J}}^{\alpha, \beta,w}_{\taus_1,r}[\Psi](\partial_{\taus})\exd r$
		\item uniformly in $T^*$ and $\taus_2$ it holds that\newline $\int_{\taus_1}^{\taus_2}\tilde{\mathbb{J}}^{\alpha, \beta,w}_{\taus,2M}[\Psi](\pr)\exd\taus-\int_{r(T^*,\taus_1)}^{r(T^*, \taus_2)}\tilde{\mathbb{J}}^{\alpha, \beta,w}_{{\taus(T^*,r),r}}[\Psi](\pt)\exd r\lesssim\int_{2M}^{r(T^*, \taus_1)}\tilde{\mathbb{J}}^{\alpha, \beta,w}_{\taus_1,r}[\Psi](\partial_{\taus})\exd r$
	\end{enumerate}	
	Then provided that the flux term on the right hand side is finite one has the estimate
	\begin{align*}
	\int_{2M}^{\infty}\tilde{\mathbb{J}}^{\alpha, \beta,w}_{\taus_2,r}[\Psi](\partial_{\taus})\exd r+\int_{\taus_1}^{\taus_2}\int_{2M}^{\infty}\widetilde{\mathbb{K}}^{\alpha, \beta,w}_{\taus,r}[\Psi]\exd\taus\exd r
	\lesssim\int_{2M}^{\infty}\tilde{\mathbb{J}}^{\alpha, \beta,w}_{\taus_1,r}[\Psi](\partial_{\taus})\exd r.
	\end{align*}
\end{proposition}

\subsection{The degenerate energy and Morawetz estimates}\label{ThedegenerateenergyandMorawetzestimate}

In this section we prove the $\pt$ flux and Morawetz estimates for solutions to the Zerilli equation.\newline

The proofs will proceed by applying Proposition \ref{propfluxandbulkestimate} with appropriate choices of the functions $\alpha,\beta$ and $w$. In particular, we henceforth assume that to the Zerilli equations under consideration satisfy the assumptions of Theorem 1 -- this ensures finiteness of the initial flux estimates that arise in the application of said proposition.\newline

\noindent\textbf{The degenerate energy estimate}\\

\noindent 
In order to derive such a flux estimate for solutions to the Zerilli equation which is moreover coercive will require controlling the terms that arise from the presence of the lower order terms that appears in the equations \eqref{Zereqn}. This is the content of the following lemma.
\begin{lemma}\label{lemmacontrollingangularforTenergy}
	Let $f\in\smfunrad$. Then on any 2-sphere $\geomtwosphere{}{r}$ one has the bounds
	\begin{align}\label{Zergradestimate}
	\normtwosphere{\sn f}{}{r}\lesssim\normtwosphere{\sn_{\mfZ} f}{}{r}\lesssim \normtwosphere{\sn f}{}{r}.
	\end{align}
\end{lemma}
\begin{proof}
	We recall from section \ref{IntegralidentitiesandintegralestimatesforsolutionstotheRegge--WheelerandZerilliequations} that
	\begin{align*}
	\normtwosphere{\sn_{\mfZ}f}{}{r}=\normtwosphere{\sn f}{}{r}-\frac{3}{r}\frac{\mu}{r}\normtwosphere{f}{}{r}-&\frac{6}{r}\frac{\mu}{r}(2-3\mu)\normtwosphere{(r\sn)\zslapinv{1}f}{}{r}\\
	+&\frac{6}{r}\frac{\mu}{r}\Big((2-3\mu)^2+3\mu\ommu\Big)\normtwosphere{\zslapinv{1}f}{}{r}.
	\end{align*}
	The upper bound in \eqref{Zergradestimate} follows from the elliptic estimates on the operator $\zslapinv{p}$ of Proposition \ref{propellipticestimateszlsap}.
	
	For the lower bound, we further recall for $f\in\smfunrad$ the Poincar\'e inequality of Lemma \ref{lemmapoincare} on any 2-sphere $\geomtwosphere{}{r}$:
	\begin{align}\label{Poinest}
	\frac{6}{r^2}\normtwosphere{ f}{}{r}\lesssim \normtwosphere{\sn f}{}{r}.
	\end{align}
For the lower bound of \eqref{Zergradestimate} we first observe that the coefficient of the term $\normtwosphere{\zslapinv{1}f}{}{r}$ in the expression $\normtwosphere{\sn_{\mfZ}f}{}{r}$ is non-negative on $r\geq 2M$ whereas the coefficient of the term $\normtwosphere{(r\sn)\zslapinv{1}f}{}{r}$ is non-negative for $2M\leq r\leq 3M$. In addition, we recall from Corollary \ref{correstimatesonzslapinv} the estimate on any 2-sphere $\twosphere{}{r}$
	\begin{align*}
	-\frac{1}{2+9\mu}\normtwosphere{f}{}{r}\leq-\normtwosphere{r\sn\zslapinv{1}f}{}{r}.
	\end{align*}
	Consequently, to establish the lower bound of \eqref{Zergradestimate} it suffices to establish the estimate
	\begin{align*}
	\normtwosphere{\sn f}{}{r}\lesssim\normtwosphere{\sn _\Upsilon f}{}{r}-\frac{3}{4}\frac{1}{r}\frac{\mu}{r}(2-3\mu)\normtwosphere{f}{}{r}
	\end{align*}
	for $2M\leq r\leq 3M$ and the estimate
	\begin{align*}
	\normtwosphere{\sn f}{}{r}\lesssim\normtwosphere{\sn _\Upsilon f}{}{r}
	\end{align*}
	for $r\geq 3M$, the latter of which was shown previously whereas the former follows easily from \eqref{Poinest}.
\end{proof}

The desired energy estimate for solutions to the Zerilli equations is then as stated below.

\begin{proposition}\label{propdegenergyestimate}
	Let $\taus_2,\taus_1\geq\taus_0$ be two real numbers.
	
	Let now $\Psi$ be as in Theorem \ref{mainthmdecayRWandZ}. Then one has the flux estimate
	\begin{align}
	&\int_{2M}^R\Big(||\pt\Psi||_{\mathsmaller{\twosphere_{\taus_2,r}}}^2+\ommu||\pr\Psi||_{\mathsmaller{\twosphere_{\taus_2,r}}}^2+||\sn\Psi||_{\mathsmaller{\twosphere_{\taus_2,r}}}^2\Big)\exd r\nonumber\\
	+&\int_{R}^\infty\Big(||D\Psi||_{\mathsmaller{\twosphere_{\taus_2,r}}}^2+||\sn\Psi||_{\mathsmaller{\twosphere_{\taus_2,r}}}^2\Big)\exd r\nonumber\\
	\lesssim&\int_{2M}^R\Big(||\pt\Psi||_{\mathsmaller{\twosphere_{\taus_1,r}}}^2+\ommu||\pr\Psi||_{\mathsmaller{\twosphere_{\taus_1,r}}}^2+||\sn\Psi||_{\mathsmaller{\twosphere_{\taus_1,r}}}^2\Big)\exd r\nonumber\\
	+&\int_{R}^\infty\Big(||D\Psi||_{\mathsmaller{\twosphere_{\taus_1,r}}}^2+||\sn\Psi||_{\mathsmaller{\twosphere_{\taus_1,r}}}^2\Big)\exd r\label{Zdegenenerest}.
	\end{align}
\end{proposition}
\begin{proof}
	We consider the three smooth radial functions $\alpha,\beta$ and $w$ on $\mcalm$ given by
	\begin{align*}
	\alpha&=1,\\
	\beta&=0,\\
	w&=0.
	\end{align*}

	From section \ref{IntegralidentitiesandintegralestimatesforsolutionstotheRegge--WheelerandZerilliequations} we have
	\begin{align*}
	\tilde{\mathbb{J}}^{1, 0,0}_{\taus,r}[\Psi](\pt)&=\opmu\normtwosphere{\pt\Psi}{}{r}+\ommu\normtwosphere{\pr\Psi}{}{r}+\normtwosphere{\sn_{\mfZ}\Psi}{}{r},\\
	\tilde{\mathbb{J}}^{1,0,0}_{\taus,r}[\Psi](\pr)&=-2\mu\,\normtwosphere{\pt\Psi}{}{r}-2\ommu\langle\pt\Psi, \pr\Psi\rangle_{\twosphere_{\taus, r}}
	\end{align*}
	and
	\begin{align*}
	\widetilde{\mathbb{K}}^{1,0,0}_{\taus,r}[\Psi]=0.
	\end{align*}
	Consequently, applying the Proposition \ref{propfluxandbulkestimate} (noting that condition $i)$ is satisfied) in conjunction with Lemma \ref{lemmacontrollingangularforTenergy} yields the estimate \eqref{Zdegenenerest}.
	
	This completes the proposition.
\end{proof}

An immediate consequence of the above computations combined with the conservations law \eqref{conslawZ} is the following estimates which will prove useful in the sequel.

\begin{corollary}\label{corrcontrolTflux}
	Let $T^*>\taus_2\geq\taus_1\geq\taus_0$ be three real numbers.
	
	Let now $\Psi$ be as in Theorem 1. Then one has the flux estimate
	\begin{align}
	&\int_{\taus_1}^{\taus_2}||\pt\Psi||_{\mathsmaller{\twosphere_{\taus,2M}}}^2\exd \taus+\int_{r(T^*,\taus_1)}^{r(T^*,\taus_2)}\tilde{\mathbb{J}}_{\taus(T^*,r),r}^{1,0,0}[\Psi](\pt)\exd r\nonumber\\
	\lesssim&\int_{2M}^R\Big(||\pt\Psi||_{\mathsmaller{\twosphere_{\taus_1,r}}}^2+\ommu||\pr\Psi||_{\mathsmaller{\twosphere_{\taus_1,r}}}^2+||\sn\Psi||_{\mathsmaller{\twosphere_{\taus_1,r}}}^2\Big)\exd r\nonumber\\
	+&\int_{R}^\infty\Big(||D\Psi||_{\mathsmaller{\twosphere_{\taus_1,r}}}^2+||\sn\Psi||_{\mathsmaller{\twosphere_{\taus_1,r}}}^2\Big)\exd r\label{ZTfluxest}.\newline
	\end{align}
\end{corollary}

\noindent\textbf{The Morawetz estimate}\\

The second such estimate we derive is an integrated local energy decay estimate which degenerates at both $\eh$ and $r=3M$.

As previously, deriving this estimate will require controlling the terms that arise from the presence of the `potential operator' in  \eqref{Zereqn}. This is the content of the following lemma.

In what follows, for a smooth radial function $f$ on $\mcalm$ we define $f^{*}:=\ommu f'$.
\begin{lemma}\label{lemmamorawetz}
	Let $\mff$ be the smooth radial function on $\mcalm$ defined according to
	\begin{align*}
	\mff:=4\bigg(1-\frac{3M}{r}\bigg)\bigg(1+\frac{3M}{r}\bigg).
	\end{align*}
	Let now $f\in\smfunrad$. Then on any 2-sphere $\geomtwosphere{}{r}$ one has the estimate
	\begin{align}\label{Zgradestimateformorawetz}
	\frac{1}{r}(2-3\mu)^2\normtwosphere{\sn f}{}{r}+\frac{1}{r^3}\normtwosphere{ f}{}{r}\lesssim&-\mff\bigg(\frac{\mu}{r}\,\normtwosphere{\sn_{\mfZ}f}{}{r}
	+\ommu\normtwosphere{[\pr,\mfZ]f}{}{r}\bigg)\nonumber\\
	&-\frac{1}{2}\frac{1}{1-\mu}\mff^{***}\normtwosphere{f}{}{r}.
	\end{align}
\end{lemma}
\begin{proof}
	We first recall from section \ref{IntegralidentitiesandintegralestimatesforsolutionstotheRegge--WheelerandZerilliequations} that
	\begin{align*}
	\normtwosphere{[\pr,\mfZ]f}{}{r}:=-\frac{2}{r}\normtwosphere{\sn f}{}{r}+\frac{9}{r^2}\frac{\mu}{r}\normtwosphere{f}{}{r}+&\frac{36}{r^2}\frac{\mu}{r}(1-2\mu)\normtwosphere{(r\sn)\zslapinv{1}f}{}{r}\\
	-&\frac{18}{r^2}\frac{\mu}{r}\Big((2-3\mu)^2+\mu^2\Big)\normtwosphere{\zslapinv{1}f}{}{r}\\
	+&\frac{108}{r^3}\mu^3\ommu\normtwosphere{(r\sn)\zslapinv{2}f}{}{r}\\
	-&\frac{108}{r^3}\mu^3\ommu(2-3\mu)\normtwosphere{\zslapinv{2}f}{}{r}.
	\end{align*}
	Subsequently, we compute that
	\begin{align*}
	-\mff\bigg(\frac{\mu}{r}\,\normtwosphere{\sn_{\mfZ}f}{}{r}
	+\ommu\normtwosphere{[\pr,\mfZ]f}{}{r}\bigg)
	-\frac{1}{2}\frac{1}{1-\mu}\mff^{***}&\normtwosphere{f}{}{r}\\
	=\frac{1}{r}(2+3\mu)(2-3\mu)^2&\normtwosphere{\sn f}{}{r}\\
	+\frac{1}{r^3}p_0(\mu)\,&\normtwosphere{f}{}{r}\\
	+\frac{1}{r^3}p_1(\mu)\,&\normtwosphere{(r\sn)\zslapinv{1}f}{}{r}\\
	+\frac{1}{r^3}p_2(\mu)\,&\normtwosphere{\zslapinv{1}f}{}{r}\\
	+\frac{1}{r^3}p_3(\mu)\,&\normtwosphere{(r\sn)\zslapinv{2}f}{}{r}\\
	+\frac{1}{r^3}p_4(\mu)\,&\normtwosphere{\zslapinv{2}f}{}{r}.
	\end{align*}
	Here, we have defined the polynomials
	\begin{align*}
	p_0(x):&=-12x\Big(3+5x-33x^2+27x^3\Big),\\
	p_1(x):&=6x\Big(24-64x-90x^2+144x^3+81x^4\Big)\\
	p_2(x):&=6x\Big(48-240x+264x^2+348x^3-837x^4+432x^5\Big),\\
	p_3(x):&=108x^3\Big(4-8x-5x^2+18x^3-9x^4\Big),\\
	p_4(x):&=108x^3\Big(8-28x+14x^2+51x^3-72x^4+25x^5\Big)
	\end{align*}
	which we observe are uniformly bounded on the domain $[0,1]$.
	Consequently, to establish the estimate \eqref{Zgradestimateformorawetz} it thus follows from both the Poincar\'e inequality of Lemma \ref{lemmapoincare} and the elliptic estimates of Proposition \ref{propellipticestimateszlsap} that it is in fact sufficient to demonstrate on any 2-sphere $\geomtwosphere{}{r}$ the bounds
	\begin{align}\label{Zbound}
	\normtwosphere{f}{}{r}\lesssim q(\mu)\normtwosphere{f}{}{r}+p_1(\mu)\,&\normtwosphere{(r\sn)\zslapinv{1}f}{}{r}\nonumber\\
	+p_2(\mu)\,&\normtwosphere{\zslapinv{1}f}{}{r}\nonumber\\
	+p_3(\mu)\,&\normtwosphere{(r\sn)\zslapinv{2}f}{}{r}\nonumber\\
	+p_4(\mu)\,&\normtwosphere{\zslapinv{2}f}{}{r}.
	\end{align}
	Here, $q(x)$ is the polynomial
	\begin{align*}
	q(x):=48-108x-168x^2+558x^3-324x^4.
	\end{align*}	
	Consequently, in order to prove estimate \eqref{Zbound} we first decompose $f$ into spherical harmonics as in section \ref{Theprojectionofsmfunontoandawayfromthel=0,1sphericalharmonics}:
	\begin{align*}
	f=\sum_{l=2}^{\infty}f^m_l\,Y^l_m.
	\end{align*} 
	Here, the convergence is pointwise. In particular, recalling that $\zslapinv{p}$ is the inverse operator $r^2\slap+2-3\mu$ applied $p$-times, it follows that for any integers $i,j\geq 0$ each mode $f^m_l\,Y^l_m$ satisfies the identity
	\begin{align*}
	\normtwosphere{(r\sn)^i\zslapinv{j}f^m_l\,Y^l_m}{}{r}=\frac{(\lambda-2)^i}{(\lambda+3\mu)^j}\normtwosphere{f^m_l\,Y^l_m}{}{r}
	\end{align*}
	where $\lambda:=(l-1)(l+2)$. Thus, to establish the bound \eqref{Zbound} it suffices to demonstrate positivity of the following expression over the domain $\lambda\geq 2$ and $x\in[0,1]$:
	\begin{align}\label{lambaestimate}
	\frac{1}{(\lambda+3x)^3}\Big(q_4(x)\lambda^4+q_3(x)\lambda^3+q_2(x)\lambda^2+q_1(x)\lambda+q_0(x)\Big).
	\end{align}
	where $q_4,q_3, q_2, q_1$ and $q_0$ are the polynomials
	\begin{align*}
	q_4(x):&=8-12x-18x^2+27x^3,\\
	q_3(x):&=16+12x-204x^2+288x^3-81x^4,\\
	q_2(x):&=x^2\Big(156-1224x+2484x^2-5508x^3\Big),\\
	q_1(x):&=x^3\Big(396-3456x+7614x^2-4662x^3\Big),\\
	q_0(x):&=3x^3\Big(108-1116x+2592x^2-1620x^3\Big).
	\end{align*}
	Indeed, we first claim that on this domain it holds that
	\begin{align*}
	(2-3x)^2(2+3x)\lesssim q_4(x)\lambda^4+q_3(x)\lambda^3+q_2(x)\lambda^2+q_1(x)\lambda+q_0(x).
	\end{align*}
	To verify this, we borrow successively from each polynomial to derive the estimate
	\begin{align*}
	q_4(x)\lambda^4+q_3(x)\lambda^3+q_2(x)\lambda^2+q_1(x)\lambda+q_0(x)\geq \frac{1}{2}q_4(x)&\lambda^4\\
	+\Big(q_3(x)-x^3\Big)&\lambda^3\\
	+\Big(q_2(x)+3x^3+2q_4(x)-\tfrac{77}{4}x^4\Big)&\lambda^2\\
	+\Big(q_1(x)+\tfrac{77}{2}x^4-13x^5\Big)&\lambda\\
	+q_0(x)+26x^5.
	\end{align*}
	Noting that $q_4(r)=(2-3x)^2(2+3x)$, the claim thus follows if positivity holds on the domain $[0,1]$ for each of the polynomials 
	\begin{align*}
	q_3(x)-x^3,\\
	q_2(x)+3x^3+2q_4(x)-\tfrac{77}{4}x^4\\
	q_1(x)+\tfrac{77}{2}x^4-13x^5\\
	q_0(x)+26x^5
	\end{align*}
	which is simple to verify.
	
	Finally, to establish positivity of the expression \eqref{lambaestimate} it remains to verify that the expression $q_4(x)\lambda^4+q_3(x)\lambda^3+q_2(x)\lambda^2+q_1(x)\lambda+q_0(x)$ is positive for $x=\frac{2}{3}$ as then continuity implies positivity on an open neighbourhood of $x=\frac{2}{3}$. However, this follows easily from an explicit computation and thus the lemma follows.
\end{proof}

The desired integrated local energy estimate for solutions to the Zerilli equation is then as follows.
\begin{proposition}\label{propmorawetz}
	Let $\Psi$ be as in Theorem 1. Then one has the bulk estimate
	\begin{align}
	\int_{\taus_0}^{\infty}&\int_{2M}^\infty\frac{1}{r^3}\bigg((2-3\mu)^2\Big(\normtwosphere{\pt\Psi}{}{r}+\ommu\normtwosphere{\pr\Psi}{}{r}+r^2\normtwosphere{\sn\Psi}{}{r}\Big)+\normtwosphere{\Psi}{}{r}\bigg)\exd\taus\exd r\nonumber\\
	\lesssim&\int_{2M}^R\Big(||\pt\Psi||_{\mathsmaller{\twosphere_{\taus_0,r}}}^2+\ommu||\pr\Psi||_{\mathsmaller{\twosphere_{\taus_0,r}}}^2+||\sn\Psi||_{\mathsmaller{\twosphere_{\taus_0,r}}}^2\Big)\exd r\nonumber\\
	+&\int_{R}^\infty\Big(||D\Psi||_{\mathsmaller{\twosphere_{\taus_0,r}}}^2+||\sn\Psi||_{\mathsmaller{\twosphere_{\taus_0,r}}}^2\Big)\exd r.\label{ZMorawetzestimate}
	\end{align}
\end{proposition}
\begin{proof}
	We consider the three smooth radial functions $\alpha,\beta$ and $w$ on $\mcalm$ given by
	\begin{align*}
	\alpha&=\mu\,\mff,\\
	\beta&=\ommu\mff,\\
	w&=-\ommu\mff'
	\end{align*}
	where $\mff$ is as in Lemma \ref{lemmamorawetz}. 
	
	Then from section \ref{IntegralidentitiesandintegralestimatesforsolutionstotheRegge--WheelerandZerilliequations} we have
	\begin{align*}
	\tilde{\mathbb{J}}^{\alpha, \beta,w}_{\taus,r}[\Psi](\pt)=\mff\Big(\mu\opmu &\normtwosphere{\pt\Psi}{}{r}+2\ommu\opmu\langle\pt\Psi, \pr\Psi\rangle_{\twosphere_{\taus, r}}\\
	-&\mu\ommu\normtwosphere{\pr\Psi}{}{r}
	+\mu\,\normtwosphere{\sn_{\mfZ}\Psi}{}{r}\Big)\\
	+\ommu\mff'\bigg(\opmu&\langle\pt\Psi, \Psi\rangle_{\twosphere_{\taus, r}}-2\mu\,\langle\pr\Psi, \Psi\rangle_{\twosphere_{\taus, r}}+\frac{1}{2}\frac{\mu}{r}\,\normtwosphere{\Psi}{}{r}\bigg),\\
	\tilde{\mathbb{J}}^{\alpha, \beta,w}_{\taus,r}[\Psi](\pr)=-\mff\Big(\big(1+\mu^2\big)\,&\normtwosphere{\pt\Psi}{}{r}+2\mu\ommu\,\langle\pt\Psi, \pr\Psi\rangle_{\twosphere_{\taus, r}}+\ommu^2\,\normtwosphere{\pr\Psi}{}{r}\\
	-\ommu\,&\normtwosphere{\sn_{\mfZ}\Psi}{}{r}\Big)\\
	+\ommu^2\mff\,&\langle\pr\Psi, \Psi\rangle_{\twosphere_{\taus, r}}-\frac{1}{2}\bigg(\big(\ommu^2 \mff'\big)'-\frac{\mu}{r}\frac{1-\mu}{2}\mff'\bigg)\normtwosphere{\Psi}{}{r}
	\end{align*}
	and
	\begin{align*}
	\widetilde{\mathbb{K}}^{\alpha, \beta,w}_{\taus,r}[\Psi]=2&\mff'\,\normtwosphere{\mu\pt\Psi+\ommu\pr\Psi}{}{r}\\
	-&\mff\bigg(\frac{\mu}{r}\,\normtwosphere{\sn_{\mfZ}\Psi}{}{r}
	+\ommu\normtwosphere{[\pr,\mfZ]\Psi}{}{r}\bigg)
	-\frac{1}{2}\frac{1}{1-\mu}\mff^{***}\normtwosphere{\Psi}{}{r}.
	\end{align*}
	Now, as $\mff$ and its derivatives are uniformly bounded on $\mcalm$, we have from Cauchy--Schwarz combined with the Poincar\'e inequality of Lemma \ref{lemmapoincare} the estimates
	\begin{align}
	-\tilde{\mathbb{J}}^{\alpha, \beta,w}_{\taus,r}[\Psi](\pt)&\lesssim\tilde{\mathbb{J}}^{1, 0,0}_{\taus,r}[\Psi](\pt),\label{aliZ}\\
	-\tilde{\mathbb{J}}^{\alpha, \beta,w}_{\taus,2M}[\Psi](\pr)&\lesssim\tilde{\mathbb{J}}^{1, 0,0}_{\taus,2M}[\Psi](\pr)\label{bobZ}
	\end{align}
	and
	\begin{align}
	\int_{2M}^\infty\tilde{\mathbb{J}}^{\alpha, \beta,w}_{\taus,r}[\Psi](\partial_{\taus})\exd r\lesssim\int_{2M}^\infty\tilde{\mathbb{J}}^{1, 0,0}_{\taus,r}[\Psi](\partial_{\taus})\exd r\label{cliveZ}.
	\end{align}
	Here, the 1-form $\tilde{\mathbb{J}}^{1, 0,0}_{\taus,r}[\Psi]$ is as in the proof of Proposition \ref{propdegenergyestimate}. Consequently, applying Proposition \ref{propfluxandbulkestimate} (noting that condition $iv)$ is satisfied by estimates \eqref{aliZ} and \eqref{bobZ} combined with  Corollary \ref{corrcontrolTflux}) in conjunction with Lemma \ref{lemmamorawetz} yields the estimate 
	\begin{align}
	\int_{\taus_0}^{\infty}&\int_{2M}^\infty\frac{1}{r^3}\Big(\normtwosphere{\mu\pt\Psi+\ommu\pr\Psi}{}{r}+(2-3\mu)^2r^2\normtwosphere{\sn\Psi}{}{r}+\normtwosphere{\Psi}{}{r}\Big)\exd\taus\exd r\nonumber\\
	\lesssim&\int_{2M}^R\Big(||\pt\Psi||_{\mathsmaller{\twosphere_{\taus_0,r}}}^2+\ommu||\pr\Psi||_{\mathsmaller{\twosphere_{\taus_0,r}}}^2+||\sn\Psi||_{\mathsmaller{\twosphere_{\taus_0,r}}}^2\Big)\exd r\nonumber\\
	+&\int_{R}^\infty\Big(||D\Psi||_{\mathsmaller{\twosphere_{\taus_0,r}}}^2+||\sn\Psi||_{\mathsmaller{\twosphere_{\taus_0,r}}}^2\Big)\exd r.\label{ZfirstMorawetzestimate}
	\end{align}
	Here, we have combined estimate \eqref{cliveZ} with the second half Proposition \ref{propdegenergyestimate} to control the flux terms arising in the first half of Proposition \ref{propfluxandbulkestimate}. 
	
	We consider now the three smooth radial functions $\alpha,\beta$ and $w$ on $\mcalm$ given by
	\begin{align*}
	\alpha&=\mu\,\mfg,\\
	\beta&=\ommu\mfg,\\
	w&=0
	\end{align*}
	where $\mfg$ is the smooth radial function
	\begin{align*}
	\mfg:=-\frac{2}{r^2}\bigg(1-\frac{3M}{r}\bigg)^3\bigg(1-\frac{2M}{r}\bigg).
	\end{align*}
	
	Then from section \ref{IntegralidentitiesandintegralestimatesforsolutionstotheRegge--WheelerandZerilliequations} we have
	\begin{align*}
	\tilde{\mathbb{J}}^{\alpha, \beta,w}_{\taus,r}[\Psi](\pt)=\mfg\Big(\mu\opmu &\normtwosphere{\pt\Psi}{}{r}+2\ommu\opmu\langle\pt\Psi, \pr\Psi\rangle_{\twosphere_{\taus, r}}\nonumber\\
	-\mu\ommu&\normtwosphere{\pr\Psi}{}{r}
	+\mu\,\normtwosphere{\sn_{\mfZ}\Psi}{}{r}\Big)\\
	\tilde{\mathbb{J}}^{\alpha, \beta,w}_{\taus,r}[\Psi](\pr)=-\mfg\Big(\big(1+\mu^2\big)\,&\normtwosphere{\pt\Psi}{}{r}+2\mu\ommu\,\langle\pt\Psi, \pr\Psi\rangle_{\twosphere_{\taus, r}}+\ommu^2\,\normtwosphere{\pr\Psi}{}{r}\\
	-\ommu\,&\normtwosphere{\sn_{\mfZ}\Psi}{}{r}\Big)
	\end{align*}
	and
	\begin{align*}
	\widetilde{\mathbb{K}}^{\alpha, \beta,w}_{\taus,r}[\Psi]=&\mfg'\,\normtwosphere{\pt\Psi}{}{r}\\
	+&\mfg'\,\normtwosphere{\mu\pt\Psi+\ommu\pr\Psi}{}{r}\\
	-&\mfg\bigg(\frac{\mu}{r}\,\normtwosphere{\sn_{\mfZ}\Psi}{}{r}
	+\ommu\normtwosphere{[\pr,\mfZ]\Psi}{}{r}\bigg)\\
	-\ommu&\mfg'\,\normtwosphere{\sn_{\mfZ}\Psi}{}{r}.
	\end{align*}
	Now, as $\mfg$ and its derivatives are uniformly bounded on $\mcalm$, we have from Cauchy--Schwarz combined with the Poincar\'e inequality of Lemma \ref{lemmapoincare} the estimates
	\begin{align}
	-\tilde{\mathbb{J}}^{\alpha, \beta,w}_{\taus,r}[\Psi](\pt)&\lesssim\tilde{\mathbb{J}}^{1, 0,0}_{\taus,r}[\Psi](\pt),\label{aligZ}\\
	-\tilde{\mathbb{J}}^{\alpha, \beta,w}_{\taus,2M}[\Psi](\pr)&\lesssim\tilde{\mathbb{J}}^{1, 0,0}_{\taus,2M}[\Psi](\pr)\label{bobgZ}
	\end{align}
	and
	\begin{align}
	\int_{2M}^\infty\tilde{\mathbb{J}}^{\alpha, \beta,w}_{\taus,r}[\Psi](\partial_{\taus})\exd r\lesssim\int_{2M}^\infty\tilde{\mathbb{J}}^{1, 0,0}_{\taus,r}[\Psi](\partial_{\taus})\exd r\label{clivegZ}.
	\end{align}
	In addition, as both $\mfg'$ and $\mfg$ vanish to second order at $r=3M$ with both $\mfg'$ and $\frac{1}{r}\mfg$ vanishing to third order as $r\rightarrow\infty$, we have from the elliptic estimates of Proposition \ref{propellipticestimateszlsap} on any 2-sphere $\geomtwosphere{}{r}$ the estimate
	\begin{align}
	\mfg'\,&\normtwosphere{\mu\pt\Psi+\ommu\pr\Psi}{}{r}\nonumber\\
	-\mfg&\bigg(\frac{\mu}{r}\,\normtwosphere{\sn_{\mfZ}\Psi}{}{r}
	+\ommu\normtwosphere{[\pr,\mfZ]\Psi}{}{r}\bigg)\nonumber\\
	-\ommu\mfg'&\,\normtwosphere{\sn_{\mfZ}\Psi}{}{r}\nonumber\\
	\lesssim\frac{1}{r^3}\,&\normtwosphere{\mu\pt\Psi+\ommu\pr\Psi}{}{r}+
	\frac{1}{r}(2-3\mu)^2\,\normtwosphere{\sn\Psi}{}{r}
	+\frac{1}{r^3}\,\normtwosphere{\Psi}{}{r}\label{prattZ}.
	\end{align}
	Consequently, Proposition \ref{propfluxandbulkestimate} combined with Proposition \ref{propdegenergyestimate}, Corollary \ref{corrcontrolTflux}, estimate \eqref{prattZ} and the fact that the function $\mfg'$ is non-negative on $\mcalm$ yields the improved estimate 
	\begin{align*}
	\int_{\taus_0}^{\infty}\int_{2M}^\infty\frac{1}{r^3}\Big(&\normtwosphere{\mu\pt\Psi+\ommu\pr\Psi}{}{r}\\
	+&(2-3\mu)^2\Big(\normtwosphere{\pt\Psi}{}{r}+r^2\normtwosphere{\sn\Psi}{}{r}+\normtwosphere{\Psi}{}{r}\Big)\exd\taus\exd r\nonumber\\
	\lesssim&\int_{2M}^R\Big(||\pt\Psi||_{\mathsmaller{\twosphere_{\taus_0,r}}}^2+\ommu||\pr\Psi||_{\mathsmaller{\twosphere_{\taus_0,r}}}^2+||\sn\Psi||_{\mathsmaller{\twosphere_{\taus_0,r}}}^2\Big)\exd r\\
	+&\int_{R}^\infty\Big(||D\Psi||_{\mathsmaller{\twosphere_{\taus_0,r}}}^2+||\sn\Psi||_{\mathsmaller{\twosphere_{\taus_0,r}}}^2\Big)\exd r
	\end{align*}
	from which estimate \eqref{ZMorawetzestimate} follows.
	
	This completes the proposition.
\end{proof}

We make the following remarks regarding Proposition \ref{propmorawetz}.
\begin{remark}
	The use of the function $\mff$ to derive the estimate \eqref{Zgradestimateformorawetz} of Lemma \ref{lemmamorawetz} first appeared in the work \cite{Me} of the author where it was motivated by earlier works of Holzegel \cite{Holz} on the Regge--Wheeler equation. See also \cite{H--K--W}.
\end{remark}

\section{Proof of Theorem 2}\label{Proofoftheorem2}

In this section we prove Theorem 2.

An outline of this section is as follows. We begin in section \ref{BoundednessanddecayforthepuregaugeandlinearisedKerrinvariantquantities} by showing that the invariant quantities of section \ref{DecouplingtheequationsoflinearisedgravityuptoresidualpuregaugeandlinearisedKerrsolutions:theRegge--WheelerandZerilliequations} associated to the initial-data normalised solution of the theorem statement satisfy the assumptions and hence the conclusions of Theorem 1. Finally in section \ref{CompletingtheproofofTheorem2} we combine the boundedness and decay bounds of Theorem 1 with Corollary \ref{corrglobalproperties} to complete the proof of Theorem 2.

\subsection{Boundedness and decay for the pure gauge and linearised Kerr invariant quantities}\label{BoundednessanddecayforthepuregaugeandlinearisedKerrinvariantquantities}

We begin in this section by first applying Theorem 1 to the invariant quantities $\Philin$ and $\Psilin$ associated to the solution $\gidnlin$ of Theorem 2. Indeed, this is immediately applicable courtesy of Theorem \ref{thmgaugeinvariantquantintermsofRWandZ} and the asymptotic flatness of the seed. We thus have:
\begin{proposition}\label{propascendingthehierarchy}
	Let $\gidnlin$ be as in the statement of Theorem 1. 
	Then the quantities $\Philin$ and $\Psilin$ assoicated to $\gidnlin$ satisfy the assumptions and hence the conclusions of Theorem 1.
\end{proposition}

\subsection{Completing the proof of Theorem 2}\label{CompletingtheproofofTheorem2}

In this section we complete the proof of Theorem 2 with the aid of Corollary \ref{corrglobalproperties}.

We first note the following lemmas.
\begin{lemma}\label{propptprframe}
	The linearly independent vector fields $\pt$ and $\pr$ determine a spherically symmetric $\qm$-frame on $\mcalm$ such that
	\begin{align}\label{qmlengthsofframe}
	\qg(\pt,\pt)=-\ommu,\qquad\qg(\pt,\pr)=\mu,\qquad\qg(\pr,\pr)=1+\mu.
	\end{align}
	In addition one has the (smooth) connection coefficients
	\begin{align*}
	\qn_{\pt}\pt&=\frac{\mu}{2}\frac{\mu}{r}\,\pt+\frac{1}{2}\frac{\mu}{r}\ommu\,\pr,\qquad\qn_{\pt}\pr=\frac{1}{2}\frac{\mu}{r}\opmu\,\pt-\frac{\mu}{2}\frac{\mu}{r}\,\pr,\\
	\qn_{\pr}\pt&=-\qn_{\pt}\pr,\,\,\,\qquad\qquad\qquad\qquad\qn_{\pr}\pr=\frac{1}{2}\frac{\mu}{2}(2+\mu)\,\pt-\frac{1}{2}\frac{\mu}{r}\opmu\,\pr.
	\end{align*}
\end{lemma}
\begin{proof}
	Computation.
\end{proof}

\begin{lemma}\label{lemmaqdL}
	Given a smooth function $f$ on $\mcalm$ we define the operator $\qexdL f$ acting on smooth vector fields according
	\begin{align*}
	\qexdL f:=\qexd f-\qhd\qexd f.
	\end{align*}
	Then for a vector field $V=\alpha\,\underline{D} +\beta\, D$ where $\alpha$ and $\beta$ are smooth functions on $\mcalm$ it holds that
	\begin{align*}
	\qexdL_V f=2\beta Df.
	\end{align*}
\end{lemma}
\begin{proof}
	Computation.
\end{proof}

\begin{proof}[Proof of Theorem 2]
	We have from Corollary \ref{corrglobalproperties} that the solution $\gidnlin$ satisfies
	\begin{align*}
	\qhatglin&=\qn\otimeshat\qexdL\big(r\Psilin\big)+6\mu\exd r\qastrosunhat\zslapinv{1}\qexd\Psilin,\\
	\qtrglin&=0,\\
	\mglin&=\sdso\Big(\qexdL\big(r\Psilin\big)-2\qexd r\,\Psilin, \qexdL\big(r\Philin\big)-2\qexd r\,\Philin\Big),\\
	\shatglin&=r\sn\otimeshat\sdso\Big(\Psilin, \Philin\Big),\\
	\strglin&=4\qexdL_P\Psilin+12\mu r^{-1}\ommu\zslapinv{1}\Psilin
	\end{align*}
	where $\Philin$ and $\Psilin$ are as in Proposition \ref{propascendingthehierarchy}. All the estimates that were derived for solutions to the Regge--Wheeler and Zerilli equations in section \ref{Proofoftheorem1} can thus be shown to hold for the solution $\gidnlin$ (but with an additional $r$ weight placed on $\gidnlin$) by dilligently commuting and evaluating the above expressions in the frame $\{\pt,\pr\}$ , keeping careful track of $r$-weights, and then applying the higher order estimates of Theorem 1. This in particular yields the pojntwise decay bounds of part $iii)$ in the statement of Theorem 2 courtesy of the Sobolev embedding on 2-spheres. Since however this would be rather cumbersome to carry out in practice we only note the key points:
	\begin{itemize}
		\item Lemmas \ref{propptprframe} and \ref{lemmaqdL} allows one to perfom all necessary computations in the $\{\pt,\pr\}$ frame. In particular, the connection coefficients in this frame are of order $O(r^{-2})$ and hence play no role when evaluating the (commuted) tensorial expressions 
		\item to control higher order angular derivatives of the solution  one commutes with the family of angular operators $\smcA{k}, \vmcA{k}$ and $\tmcA{k}$ of section \ref{ThefamilyofoperatorssA}, noting the commutation relations from the relevant Riemann tensors in section \ref{Decomposingtheequationsoflinearisedgravity}, and then apply the elliptic estimates of Proposition \ref{propellipticestimatesonA}
		\item by definition of the flux and integrated decay norms the derivatives $D$ and $\sn$ always appear with an additional $r$-weight, thus gaining in in $r$
		\item by Lemma \ref{lemmaqdL} `contracting' the operator $\qexdL$ in the frame $\{\pt,\pr\}$ always returns a $D$ derivative which gains an $r$-weight by the previous point
		\item to bound the (commuted) terms involving the operator $\zslapinv{1}$ one applies the commutation relations of Lemma \ref{lemmacommrelationsandidentities} along with the estimates of Proposition \ref{propellipticestimateszlsap}
	\end{itemize}
\end{proof}

\appendix

\section{Decomposing the equations of linearised gravity}\label{Decomposingtheequationsoflinearisedgravity}

In this section we present the equations that result from decomposing the equations of linearised gravity using the formalism of section \ref{Ageometricfoliationby2-spheres}. 

First we compute the connection coefficients of $\qn$ and $\sn$ along with related geometric formulae. We let
$p\in\mcalm$ arbitrary and choose a coordinate system $(x^\alpha){\times}(x^A)$ about $p$ where $\alpha=0,1$ with $(x^0, x^1)=(t^*,r)$ and $A=2,3$ with $(x^2, x^3)$ a coordinate system about $\pi_{\twosphere}(p)\in\twosphere$. We then define $\partial_\alpha:=\partial_{x^\alpha}$ and $\partial_A:=\partial_{x^A}$ which together form a local frame for $\mcalm$ about $p$. The former also defines a local frame for $TQ$ about $p$ with the latter forming a local frame for $TS$ about $p$ as each $\pA$ is orthogonal to each $\pa$. Since $(\qg^{-1})^{\alpha\beta}=(\g^{-1})^{\alpha\beta}$ and $(\sg^{-1})^{AB}=(\g^{-1})^{AB}$ the Koszul formula therefore yields at $p$
\begin{align*}
\qn_{\pa}\pb=\Gamma^{\gamma}_{\alpha\beta}\pg,\qquad  \sn_{\pA}\pB=\Gamma^{C}_{AB}\partial_C
\end{align*}
with $\Gamma$ the Christoffel symbols of $g_M$, $\beta,\gamma=0,1$ and $B,C=2,3$\footnote{Here we use the standard subscript and superscript notation for contracting a tensor in a frame.}. This immediately yields for $f\in\smfun$ 
\begin{align*}
\qbox f&={(\qg^{-1})}^{\alpha\beta}\Big(\pa\pb f-\Gamma^{\gamma}_{\alpha\beta}\pg f\Big),\\
\slap f&={(\sg^{-1})}^{AB}\Big(\pA\pB f-\Gamma^{C}_{AB}\pC f\Big).
\end{align*}

We now express the wave operator $\Box$ acting on smooth functions, smooth 1-forms and smooth symmetric 2-covariant tensor fields relative to $\qn$ and $\sn$. First we compute 
\begin{align*}
\nabla_{\pa}\pb=\Gamma^{\gamma}_{\alpha\beta}\pg,\qquad \nabla_{\pa}\pA=\nabla_{\pA}\pa=r_\alpha\,\pA,\qquad \nabla_{\pA}\pB=-\frac{1}{r}r^\alpha(\sg)_{AB}\,\pa+\Gamma^{C}_{AB}\partial_C
\end{align*}
and therefore by definition
\begin{align*}
\qn_{\pa}\pA=\nabla_{\pa}\pA=r_\alpha\,\pA.
\end{align*}
Here $r_\alpha:=\qexd r_{\alpha}$. For $\omega\in\smonecov$ this implies
\begin{align*}
(\nabla\omega)_{\alpha\beta}&=\big(\qn\qomega\big)_{\alpha\beta},\\
(\nabla\omega)_{\alpha A}&=\big(\qn\somega\big)_{\alpha A},\\
(\nabla\omega)_{A\alpha}&=\big(\sn\qomega \big)_{A\alpha}-\frac{1}{r}r_\alpha\,\somega_A,\\
(\nabla\omega)_{AB}&=\big(\sn\somega\big)_{AB}+\frac{1}{r}r^\alpha(\sg)_{AB}\,\qomega_\alpha
\end{align*}
and
\begin{align*}
(\nabla\nabla\omega)_{\alpha\beta\gamma}&=\big(\qn\qn\qomega\big)_{\alpha\beta\gamma},\\
(\nabla\nabla\omega)_{AB\alpha}&=\big(\sn\sn\qomega\big)_{AB\alpha}+\frac{1}{r}(\sg)_{AB}\,\big(\qn_{\qP}\qomega\big)_{\alpha}-\frac{1}{r}r_\alpha\big(\sn\astrosun\somega\big)_{AB}-\frac{1}{r^2}r_\alpha(\sg)_{AB}\,\qomega_{\qP},\\
(\nabla\nabla\omega)_{\alpha\beta A}&=\big(\qn\qn\somega\big)_{\alpha\beta A},\\
(\nabla\nabla\omega)_{AB\Gamma}&=\big(\sn\sn\somega\big)_{AB\Gamma}+\frac{1}{r}(\sg)_{AB}\,\big(\qn_{\qP}\somega\big)_{\Gamma}+\frac{2}{r}(\sg)_{A(C}\sn_{B)}\qomega_{\qP}-\frac{1}{r^2}r^\alpha r_\alpha (\sg)_{AC}\somega_B.
\end{align*}
For $\gamma\in\smsymtwocov$ we have
\begin{align*}
(\nabla\gamma)_{\alpha\beta\gamma}&=\big(\qn\qtau\big)_{\alpha\beta\gamma},\\
(\nabla\gamma)_{\alpha A\beta}&=(\nabla\gamma)_{\alpha\beta A}=\big(\qn\mtau\big)_{\alpha A\beta},\\
(\nabla\gamma)_{\alpha AB}&=\big(\qn\stau\big)_{\alpha AB},\\
(\nabla\gamma)_{A\alpha\beta}&=\big(\sn\qtau\big)_{A\alpha\beta}-\frac{2}{r}r_{(\alpha}\mtau_{\beta)A},\\
(\nabla\gamma)_{A\alpha B}&=(\nabla\gamma)_{AB\alpha}=\big(\sn\mtau\big)_{A\alpha B}+\frac{1}{r}r^\beta(\sg)_{AB}\,\qtau_{\alpha\beta}-\frac{1}{r}r_{\alpha}\,\stau_{AB},\\
(\nabla\gamma)_{ABC}&=\big(\sn\stau\big)_{ABC}+\frac{2}{r}r^{\alpha}(\sg)_{A(B}\mtau_{C)\alpha}
\end{align*}
and
\begin{align*}
(\nabla\nabla\gamma)_{\alpha\beta\gamma\delta}&=\big(\qn\qn\qtau\big)_{\alpha\beta\gamma\delta}\\
(\nabla\nabla\gamma)_{AB\alpha\beta}&=\big(\sn\sn\qtau\big)_{AB\alpha\beta}+\frac{1}{r}(\sg)_{AB}\big(\qn_{\qP}\qtau\big)_{\alpha\beta}-\frac{4}{r}\big(\sn\mtau\big)_{(AB)(\alpha}r_{\beta)}-\frac{2}{r^2}r^\gamma(\sg)_{AB}r_{(\alpha}\qtau_{\beta)\gamma}+\frac{2}{r^2}r_\alpha r_\beta\stau_{AB},\\
(\nabla\nabla\gamma)_{\alpha\beta\gamma A}&=\big(\qn\qn\mtau\big)_{\alpha\beta\gamma A}+\frac{2}{r^2}r_\alpha r_\beta\mtau_{\gamma A}\\
(\nabla\nabla\gamma)_{AB\alpha \Gamma}&=\big(\sn\sn\mtau\big)_{AB\alpha\Gamma }+\frac{1}{r}(\sg)_{AB}\big(\qn_{\qP}\mtau\big)_{\alpha\Gamma}+\frac{2}{r}r^\beta(\sg)_{\Gamma(A}\sn_{B)}\qtau_{\alpha\beta}-\frac{2}{r}r_\alpha\big(\sn\stau\big)_{(AB)\Gamma}\\
&-\frac{4}{r^2}r^\beta r_{(\alpha}\mtau_{\beta)(\Gamma}(\sg)_{B)A}\\
(\nabla\nabla\gamma)_{\alpha\beta AB}&=\big(\qn\qn\stau\big)_{\alpha\beta AB}\\
(\nabla\nabla\gamma)_{AB\Gamma\Delta}&=\big(\sn\sn\stau\big)_{AB\Gamma\Delta}+\frac{1}{r}(\sg)_{AB}\big(\qn_{\qP}\stau\big)_{\Gamma\Delta}+\frac{2}{r}r^\alpha\Big(\big(\sn\mtau\big)_{A\alpha(\Gamma}(\sg)_{\Delta)B}+\big(\sn\mtau\big)_{B\alpha(\Gamma}(\sg)_{\Delta)A}\Big)\\
&+\frac{2}{r^2}(g_M)_{A(\Gamma}(\sg)_{\Delta)B}\qtau_{\qP\qP}-2\qP(r)(\sg)_{A(\Gamma}\stau_{\Delta)B}
\end{align*}
with $\delta=0,1$ and $\Delta=2,3$. This yields
\begin{align*}
\Box_{g_M} f &=\qbox f+\slap f+\frac{2}{r}\qn_{\qP}f,\\
\big(\Box_{g_M}\omega\big)_\alpha&=\big(\qbox\qomega\big)_{\alpha}+\big(\slap\qomega\big)_\alpha+\frac{2}{r}\big(\qn_{\qP}\qomega\big)_{\alpha}-\frac{2}{r}r_\alpha\,\sdiv\somega-\frac{2}{r^2}r_\alpha\,\qomega_{\qP},\\
\big(\Box_{g_M}\omega\big)_A&=\big(\qbox\somega\big)_A+\big(\slap\somega\big)_A+\frac{2}{r}\big(\qn_{\qP}\somega\big)_{A}+\frac{2}{r}\sn_A\somega-\frac{1}{r^2}r^\alpha r_\alpha\,\somega_A,\\
\big(\Box_{g_M}\gamma\big)_{\alpha\beta}&=\big(\qbox\qtau\big)_{\alpha\beta}+\big(\slap\qtau\big)_{\alpha\beta}+\frac{2}{r}\big(\qn_{\qP}\qtau\big)_{\alpha\beta}-\frac{4}{r}\big(\sdiv\mtau\big)_{(\alpha}r_{\beta)}-\frac{4}{r^2}r^\gamma r_{(\alpha}\qtau_{\beta)\gamma}+\frac{2}{r^2}r_\alpha r_\beta\strtau,\\
\big(\Box_{g_M}\gamma\big)_{\alpha A}&=\big(\qbox\mtau\big)_{\alpha A}+\big(\slap\mtau\big)_{\alpha A}+\frac{2}{r}\big(\qn_{\qP}\mtau\big)_{\alpha A}+\frac{2}{r}r^\beta\sn_{A}\qtau_{\alpha\beta}-\frac{2}{r}r_\alpha\big(\sdiv\stau\big)_{A}-\frac{1}{r^2}r^\beta r_\beta\mtau_{\alpha A}\\
&-\frac{3}{r^2}r^\beta r_{\alpha}\mtau_{\beta A},\\
\big(\Box_{g_M}\gamma\big)_{AB}&=\big(\qbox\stau\big)_{AB}+\big(\slap\stau\big)_{AB}+\frac{2}{r}\big(\qn_{\qP}\stau\big)_{AB}+\frac{2}{r}\big(\sn\astrosun\mtau_{\qP}\big)_{AB}+\frac{2}{r^2}(g_M)_{AB}\qtau_{\qP\qP}-\frac{2}{r^2}r^\alpha r_\alpha\stau_{AB}
\end{align*}
after noting that $\g^{-1}=\qg^{-1}+\sg^{-1}$.

Noting finally the Riemann tensors
\begin{align*}
\widetilde{\text{Riem}}_{\alpha\beta\gamma\delta}&=\widetilde{\text{R}}\,(g_M)_{\alpha[\gamma}(g_M)_{\delta]\beta},\\
\slashed{\text{Riem}}_{AB\Gamma\Delta}&=\slashed{\text{R}}\,(g_M)_{A[\Gamma}(g_M)_{\Delta]B}
\end{align*}
with 
\begin{align*}
\widetilde{\text{R}}&=\frac{2}{r}\frac{\mu}{r},\\
\slashed{\text{R}}&=\frac{2}{r^2}
\end{align*}
we have
\begin{align*}
\text{Riem}_{\alpha\beta\gamma\delta}&=\widetilde{\text{Riem}}_{\alpha\beta\gamma\delta},\\
\text{Riem}_{A\alpha B\beta}&=-\frac{1}{r}(\sg)_{AB}\big(\qn r\big)_{\alpha\beta},\\
\text{Riem}_{AB\Gamma\Delta}&=(1-r^\alpha r_\alpha)\slashed{\text{Riem}}_{AB\Gamma\Delta}
\end{align*}
which gives the following.
\begin{proposition}\label{propdecomposedlinearisedsystem}
Let $\zeta:\smsymtwocov\rightarrow\smonecov$ be an $\reals$-linear map and suppose that $\gamma\in\smsymtwocov$ solves
\begin{align*}
\big(\Box \gamma\big)_{ab}-2\tensor{\textnormal{Riem}}{^c_{ab}^d}\gamma_{cd}&=2\nabla_{(a}\zeta^\gamma_{b)},\\
\big(\textnormal{div}\gamma\big)_a-\frac{1}{2}\nabla_a\tr\gamma&=\zeta^\gamma_{a}
\end{align*}
where $\zeta^\gamma:=\zeta(\gamma)$. Then the projections of $\gamma$ onto $\smsymtwocovQ, \smqmsm$ and $\smsymtwocovS$ satisfy the following system of equations\footnote{Note we do not always decompose objects into their trace and tracefree parts in order to present each respective equation within one line.}:
\begin{align*}
\qbox\qhatgamma+\slap\qhatgamma+\frac{2}{r}\qn_{\qP}\qhatgamma-\frac{2}{r}\qn r\qastrosunhat\sdiv\mgamma-\frac{2}{r^2}\qn r\qastrosunhat\qhatgamma_{\qP}-\frac{2\mu}{r^2}\qhatgamma-\frac{1}{r^2}\qn r\qastrosunhat\qn r\big(\qtrgamma-\strgamma\big)&=\qn\astrosunhat\qzeta^\gamma,\\
\qbox\qtrgamma+\slap\qtrgamma+\frac{2}{r}\qn_{\qP}\qtrgamma-\frac{4}{r}\sdiv\mgamma_{\qP}+\frac{4}{r^2}\qgamma_{\qP\qP}-\frac{2}{r^2}(1-2\mu)\big(\qtrgamma-\strgamma\big)&=-2\qdiv\qzeta,\\
\qbox\mgamma+\slap\mgamma+\frac{2}{r}\qn_{\qP}\mgamma+\frac{2}{r}\sn\astrosun\qgamma_{\qP}-\frac{2}{r}\qn r\astrosun\sdiv\sgamma-\frac{3}{r^2}\qn r\astrosun\mgamma_{\qP}-\frac{1}{r^2}(1-2\mu)\mgamma&=\qn\astrosun\szeta^\gamma+\sn\astrosun\qzeta^\gamma,\\
\qbox\shatgamma+\slap\shatgamma+\frac{2}{r}\qn_{\qP}\shatgamma+\frac{2}{r}\sn\astrosunhat\mgamma_{\qP}-\frac{2}{r^2}\opmu\shatgamma&=\sn\astrosunhat\szeta^\gamma,\\
\qbox\strgamma+\slap\strgamma+\frac{2}{r}\qn_{\qP}\strgamma+\frac{4}{r}\sdiv\mgamma_{\qP}+\frac{4}{r^2}\qhatgamma_{\qP\qP}+\frac{2}{r^2}(1-2\mu)\big(\qtrgamma-\strgamma\big)&=2\sdiv\szeta^\gamma+\frac{4}{r}\qzeta^\gamma_{\qP},\\
-\qdiv\qhatgamma+\sdiv\mgamma-\frac{1}{2}\qn\strgamma+\frac{2}{r}\qhatgamma_{\qP}+\frac{1}{r}\qn r\big(\qtrgamma-\strgamma\big)&=\qzeta^\gamma,\\
-\qdiv\mgamma-\frac{1}{2}\sn\qtrgamma+\sdiv\shatgamma+\frac{3}{r}\mgamma_{\qP}&=\szeta^\gamma.
\end{align*}
Here we have set $\qzeta^\gamma:=\widetilde{\zeta(\gamma)}$ and $\szeta^\gamma:=\zeta(\gamma)-\qzeta^\gamma$.
\end{proposition}

We then finally have the decomposition of the equations of linearised gravity -- note we focus only on the modes $l\geq 2$ as the $l=0,1$ modes are understood to be linearised Kerr plus residual pure gauge.
\begin{corollary}\label{corrdecomposedeqnslingrag}
Let $\glin$ be a smooth solution to the equations of linearised gravity with vanishing projection to $l=0,1$. Then the projections of $\glin$ onto $\smsymtwocovQ, \smqmsm$ and $\smsymtwocovS$ satisfy the following system of equations:
\begin{align*}
\qbox\qhatglin+\slap\qhatglin+\frac{2}{r}\qn_{\qP}\qhatglin+\frac{2}{r}\qn\astrosunhat\Big(\qhatglin_P\Big)-\frac{2}{r}\qn r\qastrosunhat\sdiv\mglin-\frac{2\mu}{r^2}\qhatglin-\frac{1}{r^2}\qn r\qastrosunhat\qn r\Big(\qtrglin-\strglin\Big)&=\qn\astrosunhat\qzeta[\Psilin],\\
\qbox\qtrglin+\slap\qtrglin&=0,\\
\qbox\mglin+\slap\mglin+\frac{2}{r}\qn_{\qP}\mglin-\frac{2}{r}\qn\astrosun\Big(\mglin_{\qP}\Big)-\frac{2}{r}\qn r\astrosun\sdiv\shatglin-\frac{1}{r^2}\qn r\astrosun\mglin_{\qP}-\frac{1}{r^2}(1-2\mu)\mglin&=\qn\astrosun\szeta[\Psilin,\Philin]\\
&+\sn\astrosun\qzeta[\Psilin],\\
\qbox\shatglin+\slap\shatglin+\frac{2}{r}\qn_{\qP}\shatglin-\frac{2}{r^2}\opmu\shatglin&=\sn\astrosunhat\szeta[\Psilin, \Philin],\\
\qbox\strglin+\slap\strglin+\frac{2}{r}\qn_{\qP}\strglin-\frac{4}{r^2}\qhatglin_{\qP\qP}+\frac{2}{r^2}\Big(\strglin-\qtrglin\Big)&=2\sdiv\szeta[\Psilin, \Philin]\\
&+\frac{4}{r}\qzeta[\Psilin]_{\qP},\\
-\qdiv\qhatglin+\sdiv\mglin-\frac{1}{2}\qn\strglin&=\qzeta[\Psilin],\\
-\qdiv\mglin-\frac{1}{2}\sn\qtrglin+\sdiv\shatglin+\frac{1}{r}\mglin_{\qP}&=\szeta[\Psilin, \Philin].
\end{align*}
Here we have set
\begin{align*}
\qzeta[\Psilin]:&=-\frac{1}{r^2}\qhd\qexd\Big(r^3\mfZ\Psilin\Big),\\
\szeta[\Psilin, \Philin]:&=\frac{2}{r}(1-2\mu)\sdso\Big(\Psilin, \Philin\Big)+r\sn\mfZ\Psilin
\end{align*}
where $\Psilin$ and $\Philin$ are the invariant quantities associated to $\glin$.
\end{corollary}

We also have the additional corollary regarding the decomposition of the residual pure gauge equation.
\begin{corollary}\label{corrdecomposedgaugeequation}
Let $V$ be a smooth solution to the residual pure gauge equation \eqref{eqnpuregauge} with vanishing projection to $l=0,1$. Then the projections of $V$ onto $\smoneconQ$ and $\smoneconS$ satisfy the following system of equations:
\begin{align*}
\qbox\qV+\slap\qV-\frac{2}{r}\qn\big(\qV_{\qP}\big)+\frac{2}{r}\qn r\,\sdiv\sV+\frac{2}{r^2}\qn r\,\qV_{\qP}&=0,\\
\qbox\sV+\slap\sV+\frac{1}{r^2}\ommu\sV&=0.
\end{align*}
\end{corollary}

\bibliographystyle{siam}

\end{document}